\documentclass[11pt]{article}
\usepackage[utf8]{inputenc}
\usepackage[english]{babel}
\usepackage{graphicx}         
\usepackage[pdftex, colorlinks=true, linkcolor=black, citecolor=black]{hyperref}
\usepackage{braket}
\usepackage{amsmath}
\usepackage{amssymb}
\usepackage{amsthm}
\usepackage{mathtools}
\usepackage{bbm}		
\usepackage{array}		
\usepackage{color}
\usepackage[usenames,dvipsnames]{xcolor}
\usepackage{caption}
\usepackage{enumitem}	
\usepackage{subcaption}
\usepackage{wrapfig}

\allowdisplaybreaks

\theoremstyle{definition}\newtheorem{definition}{Definition}[section]
\theoremstyle{plain}\newtheorem{thm}{Theorem}
\newtheorem{lem}[definition]{Lemma}

\newtheorem{prop}[definition]{Proposition}

\theoremstyle{remark}\newtheorem{remark}{Remark}[section]

\newcommand{\lemit}[1]{\begin{enumerate}[label={(\alph*)}, ref={\thelem\alph*}]{#1}\end{enumerate}}	
\newcommand{\remit}[1]{\begin{enumerate}[label={(\alph*)}, ref={\theremark\alph*}]{#1}\end{enumerate}}

\renewcommand{\hat}[1]{\widehat{#1}}											
\renewcommand{\tilde}[1]{\widetilde{#1}}										

\newcommand{\ls}{\lesssim}													
\newcommand{\gs}{\gtrsim}													

\newcommand{\lr}[1]{\left\langle #1 \right\rangle} 							
\newcommand{\llr}[1]{\left\llangle #1 \right\rrangle}							
\newcommand{\norm}[1]{\lVert#1\rVert}   										
\newcommand{\onorm}[1]{\lVert#1\rVert_\mathrm{op}}							

\newcommand{\R}{\mathbb{R}}	
\newcommand{\N}{\mathbb{N}}												

\newcommand\mydots{,\makebox[1em][c]{.\hfil.\hfil.},}							
\newcommand\mycdots{\makebox[1em][c]{$\cdot$\hfil$\cdot$\hfil$\cdot$}}		

\newcommand{\Tr}{\mathrm{Tr}}
\renewcommand{\d}{\mathop{}\!\mathrm{d}}

\renewcommand{\i}{\mathrm{i}}
\newcommand{\e}{\mathrm{e}}

\newcommand{\Hb}{H_{\mu,\beta}}
\newcommand{\Hm}{H_\mu}
\newcommand{\hb}{h_\beta}

\newcommand{\psiNe}{\psi^{N,\varepsilon}}

\newcommand{\Eb}{E_{\wb}}
\newcommand{\Ecal}{\mathcal{E}_{\bb}}
\newcommand{\efrak}{\mathfrak{e}_\beta}
\newcommand{\efrako}{\mathfrak{e}_1}

\newcommand{\Wb}{\mathcal{W}_{\beta,\eta}}
\newcommand{\Wbt}{\mathcal{W}_{\bt,\eta}}
\newcommand{\Vbar}{\overline{\mathcal{V}}}
\newcommand{\Vbbar}{\overline{\Vbar}}

\newcommand{\Vp}{V^\parallel}
\newcommand{\phe}{\varphi^\varepsilon}
\newcommand{\chie}{\chi^\varepsilon}

\newcommand{\wb}{w_{\mu,\beta}}
\newcommand{\wm}{w_\mu}
\newcommand{\wbot}{\wb^{(12)}}

\newcommand{\bb}{b_\beta}	
\newcommand{\bNe}{b_{\beta,N,\varepsilon}}	
\newcommand{\btNe}{b_{\bt,N,\varepsilon}}	

\newcommand{\Tex}{T^\mathrm{ex}_{\Vp}}										

\newcommand{\alwb}{\alpha_{\wb}^<}
\newcommand{\alwm}{\alpha_{\wm}^<}
\newcommand{\alUf}{\alpha_{\Ubt\fb}^<}
\newcommand{\awm}{\alpha_{\wm}}

\newcommand{\pp}{p^\Phi}														
\newcommand{\pc}{p^{\chie}}														
\newcommand{\qp}{q^\Phi}														
\newcommand{\qc}{q^{\chie}}	

\newcommand{\bo}{{\beta_1}}
\newcommand{\bz}{{\beta_2}}
\newcommand{\bt}{{\tilde{\beta}}}							

\newcommand{\Ubt}{U_{\mu,\bt}}
\newcommand{\Rbt}{\varrho_\bt}
\newcommand{\kbt}{\kappa_{\bt}}
\newcommand{\gb}{g_\bt}	
\newcommand{\gbot}{g_\bt^{(12)}}
\newcommand{\gbij}{g_\bt^{(ij)}}											
\newcommand{\fb}{f_\bt}	
\newcommand{\fbot}{f_\bt^{(12)}}

\newcommand{\fblk}{f_\bt^{(lk)}}	
\newcommand{\fbrs}{f_\bt^{(rs)}}	
\newcommand{\fbij}{f_\bt^{(ij)}}

\newcommand{\he}{h_\varepsilon}												
\newcommand{\heot}{\he^{(12)}}

\newcommand{\heij}{\he^{(ij)}}

\newcommand{\te}{H_\varepsilon}											
\newcommand{\teij}{\te^{(ij)}}
\newcommand{\teot}{\te^{(12)}}

\newcommand{\wbar}{\overline{\wb}}

\newcommand{\vbar}{\overline{v}_\rho}
\newcommand{\vbaro}{\overline{v}_1}
\newcommand{\ombar}{\overline{\omega}}

\newcommand{\hr}{\overline{h}_{\varrho_\beta,\rho}}
\newcommand{\hro}{\overline{h}_{\varrho_\beta,1}}

\newcommand{\hrt}{\overline{h}_{\Rbt,\rho}}
\newcommand{\hrto}{\overline{h}_{\Rbt,1}}

\newcommand{\wso}{\overline{\omega}_{\sigma_1}}
\newcommand{\wst}{\overline{\omega}_{\sigma_2}}
\newcommand{\hbarss}{\overline{h}_{\sigma_1,\sigma_2}}
\newcommand{\hbbarss}{\overline{\overline{h}}_{\sigma_1,\sigma_2}}
\newcommand{\hso}{\overline{h}^{(1)}_{\sigma_1,\sigma_2}}
\newcommand{\hst}{\overline{h}^{(2)}_{\sigma_1,\sigma_2}}

\newcommand{\wbbar}{\overline{\overline{\wb}}}
\newcommand{\wbbarot}{\wbbar^{(12)}}
\newcommand{\vbbar}{\overline{\overline{v}}_\rho}
\newcommand{\vbbaro}{\overline{\overline{v}}_1}
\newcommand{\nbbaro}{\overline{\overline{\nu}}_1}
\newcommand{\vbbarbz}{\overline{\overline{v}}_{\mu^\bz}}
\newcommand{\ombbar}{\overline{\ombar}}
\newcommand{\Ufbbar}{\overline{\overline{\Ubt\fb}}}

\newcommand{\hbbar}{\overline{\overline{h}}_{\varrho_\beta,\rho}}
\newcommand{\hbbarot}{\hbbar^{(12)}}
\newcommand{\hbbaroth}{\hbbar^{(13)}}
\newcommand{\hbbarbz}{\overline{\overline{h}}_{\Rbt,\mu^\bz}}
\newcommand{\hbbaro}{\overline{\overline{h}}_{\mu^\bz,1}}
\newcommand{\hbbaroo}{\overline{\overline{h}}_{\varrho_\beta,1}}
\newcommand{\hbbart}{\overline{\overline{h}}_{\Rbt,\rho}}

\newcommand{\A}{\mathcal{A}}
\newcommand{\Abar}{\overline{\mathcal{A}}}
\newcommand{\Ao}{{\A_1}}
\newcommand{\Abaro}{{\Abar_1}}
\newcommand{\B}{\mathcal{B}}
\newcommand{\Bbar}{\overline{\mathcal{B}}}
\newcommand{\Bo}{{\B_1}}
\newcommand{\Bbaro}{{\Bbar_1}}
\newcommand{\Cbar}{\overline{\mathcal{C}}}
\newcommand{\Cbaro}{{\Cbar_1}}
\newcommand{\charAbaro}{\mathbbm{1}_\Abaro}
\newcommand{\charAo}{\mathbbm{1}_\Ao}
\newcommand{\charBbaro}{\mathbbm{1}_\Bbaro}
\newcommand{\charBo}{\mathbbm{1}_\Bo}
\newcommand{\charCbaro}{\mathbbm{1}_\Cbaro}
\newcommand{\charAbarox}{\mathbbm{1}_{\overline{\mathcal{A}}_1^x}}
\newcommand{\charAox}{\mathbbm{1}_{\mathcal{A}_1^x}}

\DeclareMathOperator*{\supp}{\mathrm{supp}}

\makeatletter
\DeclareFontFamily{OMX}{MnSymbolE}{}
\DeclareSymbolFont{MnLargeSymbols}{OMX}{MnSymbolE}{m}{n}
\SetSymbolFont{MnLargeSymbols}{bold}{OMX}{MnSymbolE}{b}{n}
\DeclareFontShape{OMX}{MnSymbolE}{m}{n}{
    <-6>  MnSymbolE5
   <6-7>  MnSymbolE6
   <7-8>  MnSymbolE7
   <8-9>  MnSymbolE8
   <9-10> MnSymbolE9
  <10-12> MnSymbolE10
  <12->   MnSymbolE12
}{}
\DeclareFontShape{OMX}{MnSymbolE}{b}{n}{
    <-6>  MnSymbolE-Bold5
   <6-7>  MnSymbolE-Bold6
   <7-8>  MnSymbolE-Bold7
   <8-9>  MnSymbolE-Bold8
   <9-10> MnSymbolE-Bold9
  <10-12> MnSymbolE-Bold10
  <12->   MnSymbolE-Bold12
}{}

\let\llangle\@undefined
\let\rrangle\@undefined
\DeclareMathDelimiter{\llangle}{\mathopen}%
                     {MnLargeSymbols}{'164}{MnLargeSymbols}{'164}
\DeclareMathDelimiter{\rrangle}{\mathclose}%
                     {MnLargeSymbols}{'171}{MnLargeSymbols}{'171}
\makeatother

 \newcommand\smallO[1]{
        \mathchoice
            {
                \ensuremath{\mathop{}\mathopen{}{\scriptstyle\mathcal{O}}\mathopen{}\left(#1\right)}
            }
            {
                \ensuremath{\mathop{}\mathopen{}{\scriptstyle\mathcal{O}}\mathopen{}\left(#1\right)}
            }
            {
                \ensuremath{\mathop{}\mathopen{}{\scriptscriptstyle\mathcal{O}}\mathopen{}\left(#1\right)}
            }
            {
                \ensuremath{\mathop{}\mathopen{}{o}\mathopen{}\left(#1\right)}
            }
    }

\title{Derivation of the 2d Gross--Pitaevskii equation for strongly confined 3d bosons}
\author{Lea Boßmann\thanks{Fachbereich Mathematik, Eberhard Karls Universität Tübingen, Auf der Morgenstelle 10, 72076 Tübingen, Germany; and
Institute of Science and Technology Austria, Am Campus 1, 3400 Klosterneuburg, Austria.
E-mail: \texttt{lea.bossmann@ist.ac.at}}~ }
\date{\today}

\begin{document}
\maketitle

\begin{abstract}
\noindent
We study the dynamics of a system of $N$ interacting bosons in a disc-shaped trap, 
which is realised by an external potential that confines the bosons in one spatial dimension to an interval of length of order $\varepsilon$.
The interaction is non-negative and scaled in such a way that its scattering length is of order $\varepsilon/N$, while its range is proportional to $(\varepsilon/N)^{\beta}$ with scaling parameter $\beta\in(0,1]$. 

\noindent We consider the simultaneous limit $(N,\varepsilon)\to(\infty,0)$ and assume that the system initially exhibits Bose--Einstein condensation. We prove that condensation is preserved by the $N$-body dynamics, where the time-evolved condensate wave function is the solution of a two-dimensional non-linear equation. The strength of the non-linearity depends on the scaling parameter $\beta$. For $\beta\in(0,1)$, we obtain a cubic defocusing non-linear Schrödinger equation, while the choice $\beta=1$ yields a Gross--Pitaevskii equation featuring the scattering length of the interaction. In both cases, the coupling parameter depends on the confining potential.
\end{abstract}

\section{Introduction}\label{sec:intro}
For two decades, it has been experimentally possible to realise quasi-two dimensional Bose gases in disc-shaped traps~\cite{gorlitz2001, rychtarik2004, smith2005}. The study of such systems is of particular physical interest since they permit the detection of inherently two-dimensional effects and serve as models for different statistical physics phenomena~\cite{hadzibabic2006, hadzibabic2008, yefsah2011}.
In this article, our aim is to contribute to the mathematically rigorous understanding of such systems.
We consider a Bose--Einstein condensate of $N$ identical, non-relativistic, interacting bosons in a disc-shaped trap, which effectively confines the particles in one spatial direction to an interval of length $\varepsilon$. 
We study the dynamics of this system in the simultaneous limit $(N,\varepsilon)\to(\infty,0)$, where the Bose gas becomes quasi two-dimensional.
To describe the $N$ bosons, we use the coordinates
$$z=(x,y)\in\R^{2+1}\,,$$
where $x$ denotes the two longitudinal dimensions and $y$ is the transverse dimension.
The confinement in the $y$-direction is modelled by the scaled potential 
$\frac{1}{\varepsilon^2}V^\perp\left(\tfrac{y}{\varepsilon}\right)$ for $0<\varepsilon\ll 1$ and some $V^\perp:\R\rightarrow\R$.  
In units such that $\hbar=1$ and $m=\frac12$, the Hamiltonian is given by
\begin{equation}\label{H}
\Hb(t)=\sum\limits_{j=1}^N\left(-\Delta_j+\frac{1}{\varepsilon^2}V^\perp\left(\frac{y_j}{\varepsilon}\right)+\Vp(t,z_j)\right)+\sum\limits_{1\leq i<j\leq N} \wb(z_i-z_j)\,,
\end{equation}
where $\Delta$ denotes the Laplace operator on $\R^3$ and $\Vp:\R\times\R^3\to\R$ is an additional external potential, which may depend on time. 
The interaction $\wb$ between the particles is purely repulsive and scaled in dependence of the parameters $N$ and $\varepsilon$.
In this paper, we consider two fundamentally different scaling regimes, corresponding to different choices of the scaling parameter $\beta\in \R$:
$\beta\in(0,1)$ yields the non-linear Schrödinger (NLS) regime, while $\beta=1$ is known as the Gross--Pitaevskii regime.
Making use of the parameter
$$\mu:=\frac{\varepsilon}{N}\,,$$
the Gross--Pitaevskii regime is realised by scaling an interaction $w:\R^3\to\R$, which is compactly supported, spherically symmetric and non-negative,
as
\begin{equation}\label{int:GP}
\wm(z)=\frac{1}{\mu^{2}}w\left(\frac{z}{\mu}\right)\,.
\end{equation}
For the NLS regime, we will consider a more generic form of the interaction (see Definition~\ref{def:W}). For the length of this introduction, let us focus on the special case 
\begin{equation}\label{interaction}
\wb(z)=\mu^{1-3\beta} \,w\left(\mu^{-\beta}z\right) 
\end{equation}
with $\beta\in(0,1)$. Clearly,~\eqref{int:GP} equals~\eqref{interaction} with the choice $\beta=1$.
Both scaling regimes describe very dilute gases, and we comment on their physical relevance below.
\\

The $N$-body wave function $\psiNe(t)\in L^2_+(\R^{3N}):=\otimes_\mathrm{sym}^N L^2(\R^3)$ at time $t\in\R$ is determined by the Schrödinger equation
\begin{equation}\label{SE}
\begin{cases}\i\tfrac{\d}{\d t}\psi^{N,\varepsilon}(t)=\Hb(t)\psi^{N,\varepsilon}(t)\\[8pt]
\psi^{N,\varepsilon}(0)=\psi^{N,\varepsilon}_0\end{cases}
\end{equation}
with initial datum $\psi^{N,\varepsilon}_0\in L^2_+(\R^{3N}).$
We assume that this initial state exhibits Bose--Einstein condensation, i.e., that the one-particle reduced density matrix $\gamma^{(1)}_{\psiNe_0}$  of $\psiNe_0$,
\begin{equation}\label{rdm}
\gamma^{(1)}_{\psiNe_0}:=\Tr_{2\mydots N}|\psiNe_0\rangle\langle\psiNe_0|\,,
\end{equation}
converges to a projection onto the so-called condensate wave function $\phe_0\in L^2(\R^3)$.
At low energies, the  strong confinement in the transverse direction causes the condensate wave function to factorise in the limit $\varepsilon\to 0$ into a longitudinal part $\Phi_0\in L^2(\R^2)$ and a transverse part $\chie\in L^2(\R)$, 
$$\phe_0(z)=\Phi_0(x)\chie(y) $$
(see Remark \ref{rem:assumptions:gs}).
The transverse part $\chie$ is given by the normalised ground state of $-\tfrac{\d^2}{\d y^2}+\tfrac{1}{\varepsilon^2}V^\perp(\tfrac{y}{\varepsilon})$, which is defined by
$$\left(-\tfrac{\d^2}{\d y^2}+\tfrac{1}{\varepsilon^2}V^\perp\left(\tfrac{\cdot}{\varepsilon}\right)\right)\chie=\tfrac{E_0}{\varepsilon^2}\chie \,.$$
Here, $E_0$ denotes the minimal eigenvalue of the unscaled operator $-\tfrac{\d^2}{\d y^2}+V^\perp$, corresponding to the  normalised ground state $\chi$.
The relation of $\chie$ and $\chi$ is 
\begin{equation}\label{eqn:chie}
\chie(y):=\tfrac{1}{\sqrt{\varepsilon}}\,\chi\left(\tfrac{y}{\varepsilon}\right)\,.
\end{equation}
By~\cite[Theorem 1]{griesemer2004}, $\chie$  is exponentially localised on a scale of order $\varepsilon$ for suitable confining potentials $V^\perp$, such as harmonic potentials or smooth, bounded potentials that admit at least one bound state below the essential spectrum. \\

In this paper, we derive an effective description of the many-body dynamics $\psiNe(t)$. We show that if the system initially forms a Bose--Einstein condensate with factorised condensate wave function, then the dynamics generated by $\Hb(t)$ preserve this property.
Under the assumption that
$$ \lim\limits_{(N,\varepsilon)\to(\infty,0)}\Tr_{L^2(\R^3)}\left|\gamma^{(1)}_{\psiNe_0}-|\phe_0\rangle\langle\phe_0|\right|=0\,,$$
where the limit $(N,\varepsilon)\to(\infty,0)$ is taken along a suitable sequence, 
we show that
$$ \lim\limits_{(N,\varepsilon)\to(\infty,0)}\Tr_{L^2(\R^3)}\left|\gamma^{(1)}_{\psiNe(t)}-|\phe(t)\rangle\langle\phe(t)|\right|=0$$
with time-evolved condensate wave function $\phe(t)=\Phi(t)\chie$.
While the transverse part of the condensate wave function remains in the ground state, merely undergoing phase oscillations, the longitudinal part is subject to a non-trivial time evolution. 
We show that this evolution is determined by the two-dimensional non-linear equation
\begin{equation}\label{NLS}
\begin{cases}
\i\tfrac{\partial}{\partial t}\Phi(t,x)=\left(-\Delta_x+\Vp(t,(x,0))+\bb|\Phi(t,x)|^2\right)\Phi(t,x)=:\hb(t)\Phi(t,x)\\[8pt]
\Phi(0)=\Phi_0\,.\end{cases}
\end{equation}
The coupling parameter $\bb$ in~\eqref{NLS} depends on the scaling regime and is given by
$$\bb=\begin{cases}
	\displaystyle\norm{w}_{L^1(\R^3)}\int_{\R}|\chi(y)|^4\d y &\quad \text{ for } \beta\in(0,1),\\[10pt]
	\displaystyle 8\pi a\int_{\R}|\chi(y)|^4\d y &\quad \text{ for } \beta=1,
\end{cases}$$
where $a$ denotes the scattering length of $w$ (see Section~\ref{subsec:GP} for a definition). 
The evolution equation~\eqref{NLS} provides an effective description of the dynamics.
Since the $N$ bosons interact, it contains an effective one-body potential, which is given by the probability density $N|\Phi(t)|^2$ times the two-body scattering process times a factor $\int_\R|\chie(y)|^4\d y$ from the confinement. At low energies, the scattering is to leading order described by the $s$-wave scattering length $a_{\mu,\beta}$ of the interaction $\wb$, which scales as $a_{\mu,\beta}\sim\mu$ for the whole parameter range $\beta\in(0,1]$ (see~\cite[Lemma~A.1]{erdos2007}) and characterises the length scale of the inter-particle correlations.

For the regime $\beta\in(0,1)$, we find $a_{\mu,\beta}\ll\mu^\beta$, i.e., the scattering length is negligible compared to the range of the interaction in the limit $(N,\varepsilon)\to(\infty,0)$. In this situation, the first order Born approximation $8\pi a_{\mu,\beta}\approx\int_{\R^3}\wb(z)\d z$ is a valid description of the scattering length and yields above coupling parameter $\bb$ for $\beta\in(0,1)$. 

In the scaling regime $\beta=1$, the first order Born approximation breaks down since $a_{\mu,1}\sim\mu$, which implies that the correlations are visible on the length scale $\mu$ of the interaction even in the limit $(N,\varepsilon)\to(\infty,0)$. Consequently, the coupling parameter $b_1$ contains the full scattering length,
which makes~\eqref{NLS} a Gross--Pitaevskii equation.\\

Physically, the scaling $\beta=1$ is relevant because it corresponds to an $(N,\varepsilon)$-independent interaction via a suitable coordinate transformation.
In the Gross--Pitaevskii regime, the kinetic energy per particle (in the longitudinal directions) is of the same order  as the total energy per particle (without counting the energy from the confinement or the external potential). 
For $N$ bosons which interact via a potential with scattering length $A$ in a trap with longitudinal extension $L$ and transverse size $\varepsilon L$, the former scales as $E_\mathrm{kin}\sim L^{-2}$. The latter can be computed as $E_\mathrm{total}\sim A\varrho_{3d}\sim AN/(L^3\varepsilon)$, where $\varrho_{3d}$ denotes the particle density. Both quantities being of the same order implies the scaling condition $A/L\sim\varepsilon/N$.

The choice $A\sim 1$ entails $L\sim N/\varepsilon$ and corresponds to an $(N,\varepsilon)$-independent interaction potential. Hence, to capture $N$ bosons in a strongly asymmetric trap while remaining in the Gross--Pitaevskii regime, one must increase the longitudinal length scale of the trap as $N/\varepsilon$ and the transverse scale as $N$.
For our analysis, we choose to work instead in a setting where $L\sim 1$, thus we consider interactions with scattering length $A\sim \varepsilon/N$. 
Both choices are related by the coordinate transform $z\mapsto (\varepsilon/N)z$, which comes with the time rescaling $t\mapsto (\varepsilon/N)^2t$ in the $N$-body Schrödinger equation~\eqref{SE}.

For the scaling regime $\beta\in(0,1)$, there is no such coordinate transform relating $\wb$ to a physically relevant $(N,\varepsilon)$-independent interaction. 
We consider this case mainly because the derivation of the Gross--Pitaevskii equation for $\beta=1$ relies on the corresponding result for $\beta\in(0,1)$. The central idea of the proof is to approximate the interaction $\wm$ by an appropriate potential with softer scaling behaviour covered by the result for $\beta\in(0,1)$, and to control the remainders from this substitution.
We follow the approach developed by Pickl in~\cite{pickl2015}, which was adapted to the problem with strong confinement in~\cite{NLS} and~\cite{GP}, where an effectively one-dimensional NLS resp.\ Gross--Pitaevskii equation was derived for three-dimensional bosons in a cigar-shaped trap. The model considered in~\cite{NLS,GP} is analogous to our model~\eqref{H} but with a two-dimensional confinement, i.e., where $(x,y)\in\R^{1+2}$.
Since many estimates are sensitive to the dimension and need to be reconsidered, the adaptation to our problem with one-dimensional confinement is non-trivial. 
A detailed account of the new difficulties is given in Remarks~\ref{rem:differences:NLS:1d} and~\ref{rem:differences:GP:1d}.
\medskip

To the best of our knowledge, the only existing derivation of a two-dimensional evolution equation from the three\--di\-men\-sion\-al $N$-body dynamics is by Chen and Holmer in~\cite{chen2013}.
Their analysis is restricted to the range $\beta\in(0,\frac25)$, which in particular does not include the physically relevant Gross--Pitaevskii case.
In this paper, we extend their result to the full regime $\beta\in(0,1]$ and include a larger class of confining traps as well as a possibly time-dependent external potential. 
We impose different conditions on the parameters $N$ and $\varepsilon$, which are stronger than in \cite{chen2013} for small $\beta$ but much less restrictive for larger $\beta$ (see Remark \ref{rem:admissibility}).
Related results for a cigar-shaped confinement were obtained in~\cite{NLS, GP,  chen2017, keler2016}.
\\

Regarding the situation without strong confinement, the first mathematically rigorous justification of a three-dimensional NLS equation from the quantum many-body dynamics of three-dimensional bosons with repulsive interactions was by Erd\H{o}s, Schlein and Yau in~\cite{erdos2007}, who extended their analysis to the Gross-Pitaevskii regime in~\cite{erdos2010}.
With a different approach, Pickl derived effective evolution equations for both regimes~\cite{pickl2015}, providing also estimates of the rate of convergence. Benedikter, De Oliveira and Schlein  proposed a third and again different strategy in~\cite{benedikter2015}, which was then adapted by Brennecke and Schlein in~\cite{brennecke2017} to yield the optimal rate of convergence. For two-dimensional bosons, effective NLS dynamics of repulsively interacting bosons were first derived by Kirkpatrick,  Schlein and Staffilani in~\cite{kirkpatrick2011}. This result was extended to more singular scalings of the interaction, including the Gross--Pitaevskii regime, by Leopold, Jeblick and Pickl in \cite{jeblick2016}, and two-dimensional attractive interactions were covered in \cite{chen2017_2, jeblick2017, lewin2017}.
Further results concerning the derivation of effective dynamics for interacting bosons were obtained, e.g., in \cite{adami2007,anapolitanos2017,chong2016,jeblick2018,knowles2010,michelangeli2017_2,michelangeli2017,triay2019}.

The dimensional reduction of non-linear one-body equations was studied in~\cite{abdallah2005_2} by Ben Abdallah, M\'ehats, Schmeiser and Weishäupl, who consider an $n+d$-dimensional NLS equation with a $d$-dimensional quadratic confining potential. In the limit where the diameter of this confinement converges to zero, they obtain an effective $n$-dimensional NLS equation. A similar problem for a cubic NLS equation in a quantum waveguide, resulting in a limiting one-dimensional equation, was covered by M\'ehats and Raymond in~\cite{mehats2017},
and the corresponding problem for the linear Schrödinger equation was studied, e.g., in \cite{wachsmuth,deoliveira2014}.
\\

The remainder of the paper is structured as follows: in Section~\ref{sec:main}, we state our assumptions and present the main result. 
The strategy of proof for the NLS scaling  is explained in Section~\ref{subsec:NLS}, while the Gross--Pitaevskii scaling is covered in Section~\ref{subsec:GP}. Section~\ref{subsec:proof} contains the proof of our main result, which depends on five propositions. Section~\ref{sec:preliminaries} collects some auxiliary estimates, which are used in Sections~\ref{sec:NLS} and~\ref{sec:GP} to prove the propositions for $\beta\in(0,1)$ and $\beta=1$, respectively.
\\

\noindent{\bf Notation.} We use the notations $A\ls B$, $A\gs B$ and $A\sim B$ to indicate that there exists a constant $C>0$ independent of $\varepsilon, N, t, \psi^{N,\varepsilon}_0,\Phi_0$ such that $A\leq CB$, $A\geq CB$ or $A=CB$, respectively. This constant may, however, depend on the quantities fixed by the model, such as $V^\perp$, $\chi$ and $\Vp$.
Besides, we will exclusively use the symbol $\,\hat{\cdot}\,$ to denote the weighted many-body operators from Definition~\ref{def:p} and use the abbreviations $$\llr{\cdot,\cdot}:=\lr{\cdot,\cdot}_{L^2(\R^{3N})},\quad \norm{\cdot}:=\norm{\cdot}_{L^2(\R^{3N})}
\quad\text{and}\quad\onorm{\cdot}:=\norm{\cdot}_{\mathcal{L}(L^2(\R^{3N}))}.$$
Finally, we write $x^+$ and $x^-$ to denote $(x+\sigma)$ and $(x-\sigma)$ for any fixed $\sigma>0$, which is to be understood in the following sense: Let the sequence $(N_n,\varepsilon_n)_{n\in\N}\to(\infty,0)$. Then
\begin{eqnarray*}
f(N,\varepsilon)\ls N^{-x^-} &:\Leftrightarrow& \text{ for any }\sigma>0, \; f(N_n,\varepsilon_n)\ls N_n^{-x+\sigma} \; \text{ for sufficiently large }n\,, \\
f(N,\varepsilon)\ls\varepsilon^{x^-} &:\Leftrightarrow &\text{ for any }\sigma>0, \; f(N_n,\varepsilon_n)\ls \varepsilon_n^{x-\sigma} \; \text{ for sufficiently large }n\,,\\
f(N,\varepsilon)\ls\mu^{x^-} &:\Leftrightarrow &\text{ for any }\sigma>0, \; f(N_n,\varepsilon_n)\ls \mu_n^{x-\sigma} \; \text{ for sufficiently large }n\,.
\end{eqnarray*}
Note that these statements concern fixed $\sigma$ in the limit $(N,\varepsilon)\to(\infty,0)$ and do in general not hold uniformly as $\sigma\to0$.
In particular, the implicit constants in the notation $\ls$ may depend on $\sigma$.

\section{Main result}\label{sec:main}
Our aim is to derive an effective description of the dynamics $\psiNe(t)$ in the simultaneous limit $(N,\varepsilon)\to(\infty,0)$. To this end, we consider families of initial data $\psiNe_0$ along sequences $(N_n,\varepsilon_n)$ with the following two properties:
\begin{definition}\label{def:admissible}
Let $\left\{(N_n,\varepsilon_n)\right\}_{n\in\N}\subset\mathbb{N}\times(0,1)$ such that $\lim_{n\rightarrow\infty}(N_n,\varepsilon_n)=(\infty,0)$, and let $\mu_n:=\varepsilon_n/N_n$.
The sequence is called 
\begin{itemize}
\item ($\Theta$\,-)\emph{admissible}, if 
$$\lim\limits_{n\rightarrow\infty}\frac{\varepsilon^{\Theta}_n}{\mu_n}=N_n\varepsilon_n^{\Theta-1}=0\,,$$
\item ($\Gamma$-)\emph{moderately confining}, if 
$$\lim\limits_{n\rightarrow\infty}\frac{\varepsilon_n^\Gamma}{\mu_n}=N_n\varepsilon_n^{\Gamma-1}=\infty\,.$$
\end{itemize}
\end{definition}

Our result holds for sequences $(N,\varepsilon)$ that are $(\Theta,\Gamma)_\beta$-admissible with parameters
\begin{equation}\label{values:Gamma:Theta}
\begin{cases}
\displaystyle\frac{1}{\beta}=\Gamma<\Theta<\frac{3}{\beta}&  \quad \beta\in(0,1)\,,\\[10pt]
\displaystyle \;\, 1<\Gamma<\Theta\leq 3 & \quad \beta=1\,.
\end{cases}
\end{equation}
The admissibility condition implies that $\varepsilon^{\beta\Theta}/\mu^\beta\ll 1$. 
Hence, by imposing this condition, we ensure that  the diameter $\varepsilon$ of the confining potential does not shrink too slowly compared to the range $\mu^\beta$ of the interaction.
Consequently, the energy gap above the transverse ground state, which scales as $\varepsilon^{-2}$, is always large enough to sufficiently suppress transverse excitations.
Clearly, it is necessary to choose $\Theta>1$, and the condition is weaker for larger $\Theta$.

In the proof, we require the admissibility condition to control the orthogonal excitations in the transverse direction (see Remark~\ref{rem:differences:NLS:1d}), which results in the respective upper bound for $\Theta$.
The threshold $\Theta=3^+$ admits $N\sim\varepsilon^{-2}$, which has a physical implication: if the confinement is realised by a harmonic trap $V^\perp(y)=\omega^2y^2$, the frequency $\omega_\varepsilon$ of the rescaled oscillator $\varepsilon^{-2}V^\perp(y/\varepsilon)$ scales as $\omega_\varepsilon=\omega\varepsilon^{-2}$. Hence, $\Theta=3^+$ means that the frequency of the confining trap grows proportionally to $N$.

The moderate confinement condition implies that, for sufficiently large $N$ and small $\varepsilon$,
\begin{equation}\label{eqn:mod:conf}
\frac{\mu}{\varepsilon^\Gamma}=N^{-1}\varepsilon^{1-\Gamma}\ll 1 
\quad \Leftrightarrow\quad
\begin{cases}
	\displaystyle \frac{\mu^\beta}{\varepsilon}\ll 1 &\beta\in(0,1)\\[10pt]
	\displaystyle \frac{\mu}{\varepsilon^\Gamma} \ll1 & \beta=1\,.
\end{cases}
\end{equation}
Moderate confinement means that $\varepsilon$ does not shrink too fast compared to $\mu^\beta$. 
For $\beta\in(0,1)$, it implies that the interaction is always supported well within the trap.
This is automatically true for $\beta=1$ because $\mu/\varepsilon=N^{-1}$, but we require a somewhat stronger condition to handle the Gross--Pitaevskii scaling (see Remark \ref{rem:differences:GP:1d}). This leads to the additional moderate confinement condition for $\beta=1$ with parameter $\Gamma>1$, which is clearly a weaker restriction for smaller $\Gamma$, and we expect this to be a purely technical condition (see Remark \ref{rem:mod:conf}).
The upper bound $\Gamma<\Theta$ is necessary to ensure the mutual compatibility of admissibility and moderate confinement.

From a technical point of view, the moderate confinement condition allows us to compensate for certain powers of $\varepsilon^{-1}$ in terms of powers of $N^{-1}$, while the admissibility condition admits the control of powers of $N$ by powers of $\varepsilon$.
\\

To visualise the restrictions due to admissibility and moderate confinement, we plot in Figure~\ref{fig} the largest possible subset of the parameter space $\mathbb{N}\times[0,1]$ which can be  covered by our analysis. A sequence $(N,\varepsilon)\to(\infty,0)$ passes through this space from the top right to the bottom left corner. The two boundaries correspond to the two-stage limits where first ${N\to\infty}$ at constant $\varepsilon$ and subsequently $\varepsilon\to0$, and vice versa. The edge cases are not contained in our model.

The sequences $(N,\varepsilon)\to(\infty,0)$ within the dark grey region in Figure~\ref{fig} are covered by our analysis and yield an NLS or Gross--Pitaevskii equation, respectively.
Naturally, these restrictions are meaningful only for sufficiently large $N$ and small $\varepsilon$, which implies that mainly the section of the plot around the bottom left corner is of importance. 
The white region in figures (a) to (c) is excluded from our analysis by the admissibility condition. In figure (d), there is an additional prohibited region due to moderate confinement. 
Note that Chen and Holmer impose constraints which are weaker for small $\beta$ and stronger for larger $\beta\in(0,\frac25)$, which are discussed in Remark \ref{rem:admissibility} and plotted in Figure \ref{fig:CH}.

The light grey region in Figure~\ref{fig}, which is present for $\beta\in(0,1)$, is not contained in Theorem~\ref{thm} as a consequence of the moderate confinement condition. 
We expect the dynamics in this region to be described by an effective equation with coupling parameter $\bb=0$ since it corresponds to the condition $\varepsilon/\mu^\beta\ll 1$, implying that the the confinement shrinks much faster than the interaction. 
Consequently, the interaction is predominantly supported in a region that is essentially inaccessible to the bosons, which results in a free evolution equation.
For $\beta<\frac13$ and a cigar-shaped confinement by Dirichlet boundary conditions, this was shown in~\cite{keler2016}.\\

\begin{figure}[t]
\begin{subfigure}[t]{.5\linewidth}
    \centering\includegraphics[scale=0.45]{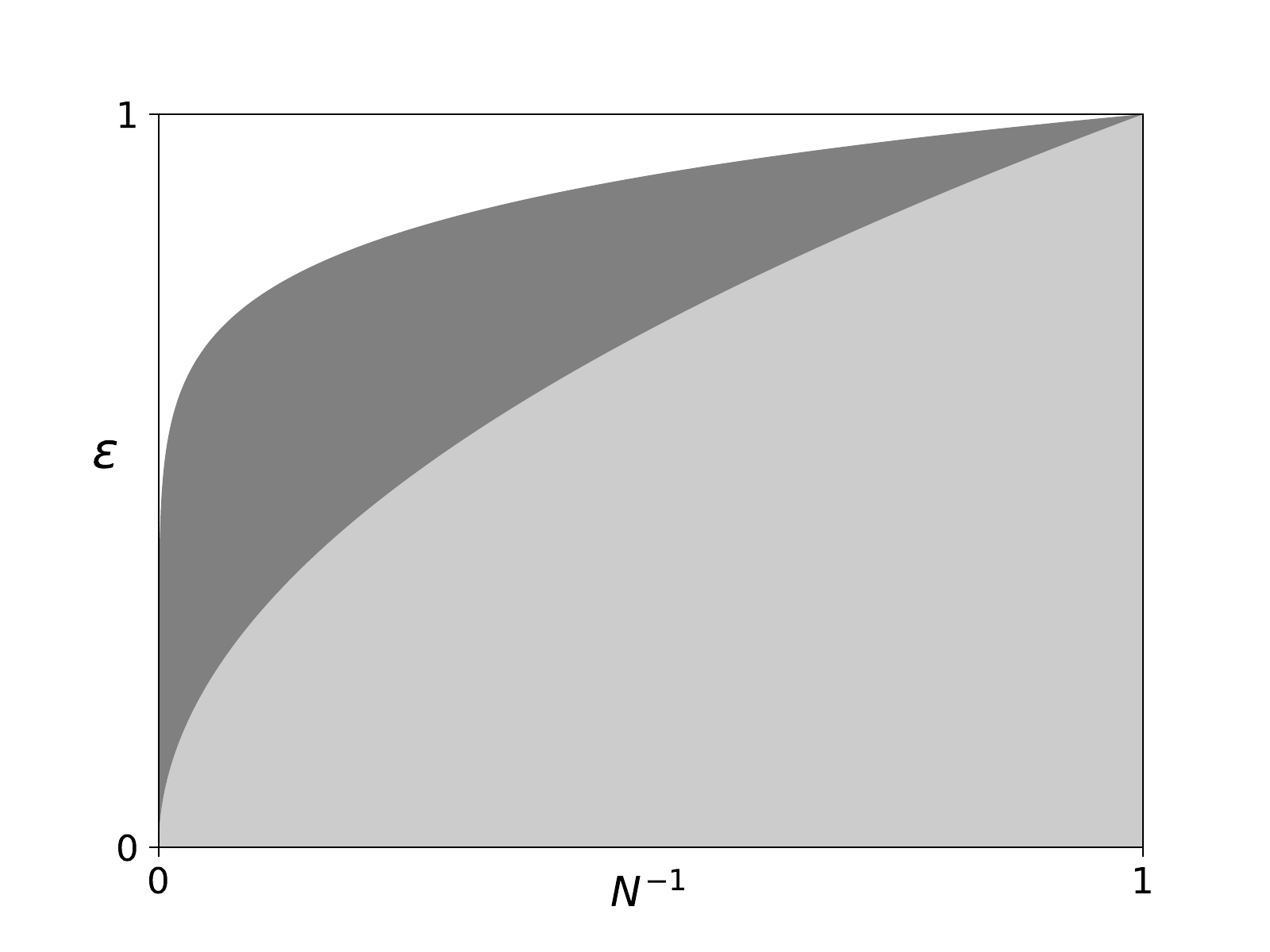}
    \caption{$\beta=\frac13$}
  \end{subfigure}
  \begin{subfigure}[t]{.5\linewidth}
    \centering\includegraphics[scale=0.45]{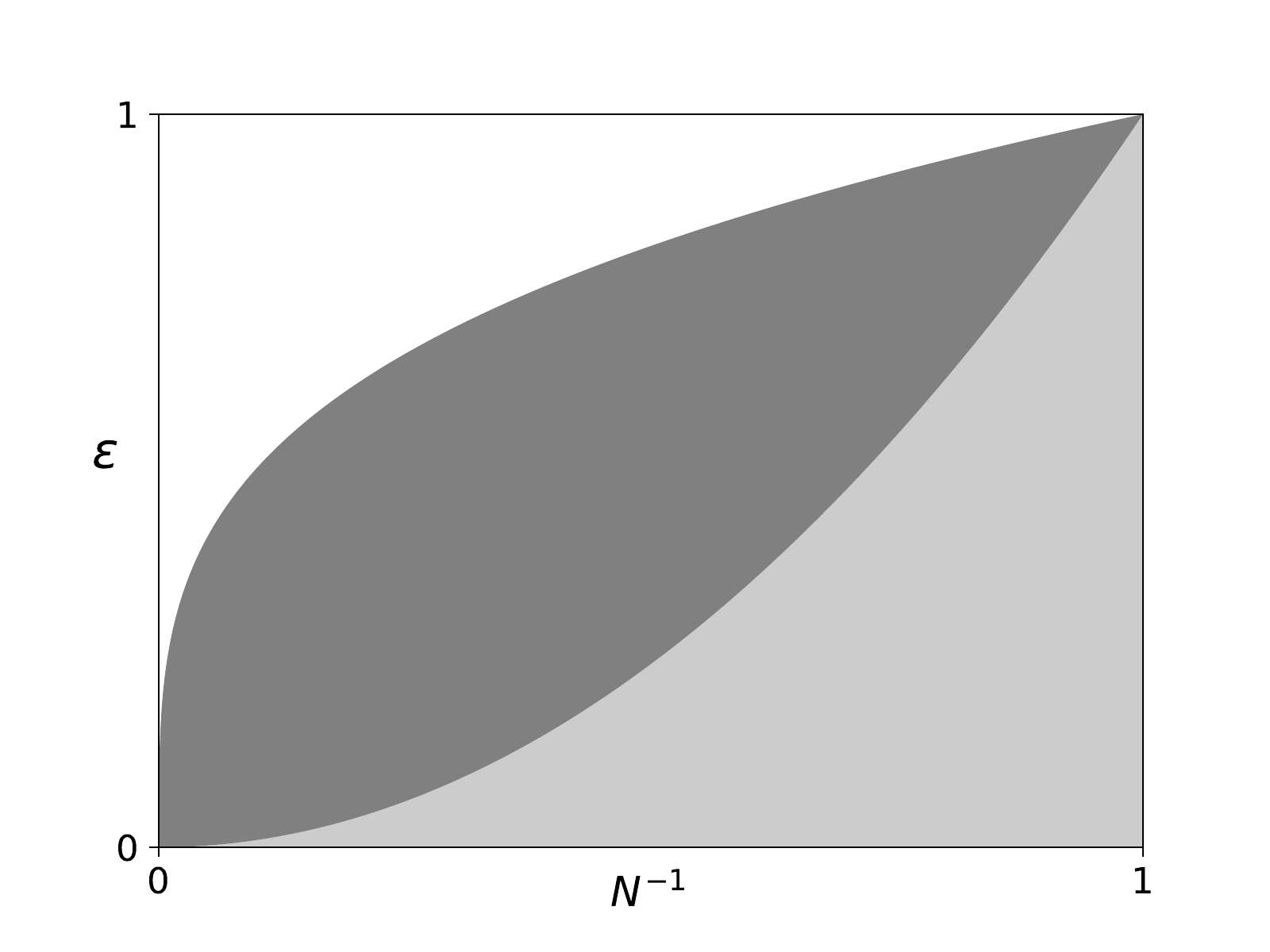}
    \caption{$\beta=\frac23$}
  \end{subfigure}
  \begin{subfigure}[t]{.5\linewidth}
    \centering\includegraphics[scale=0.45]{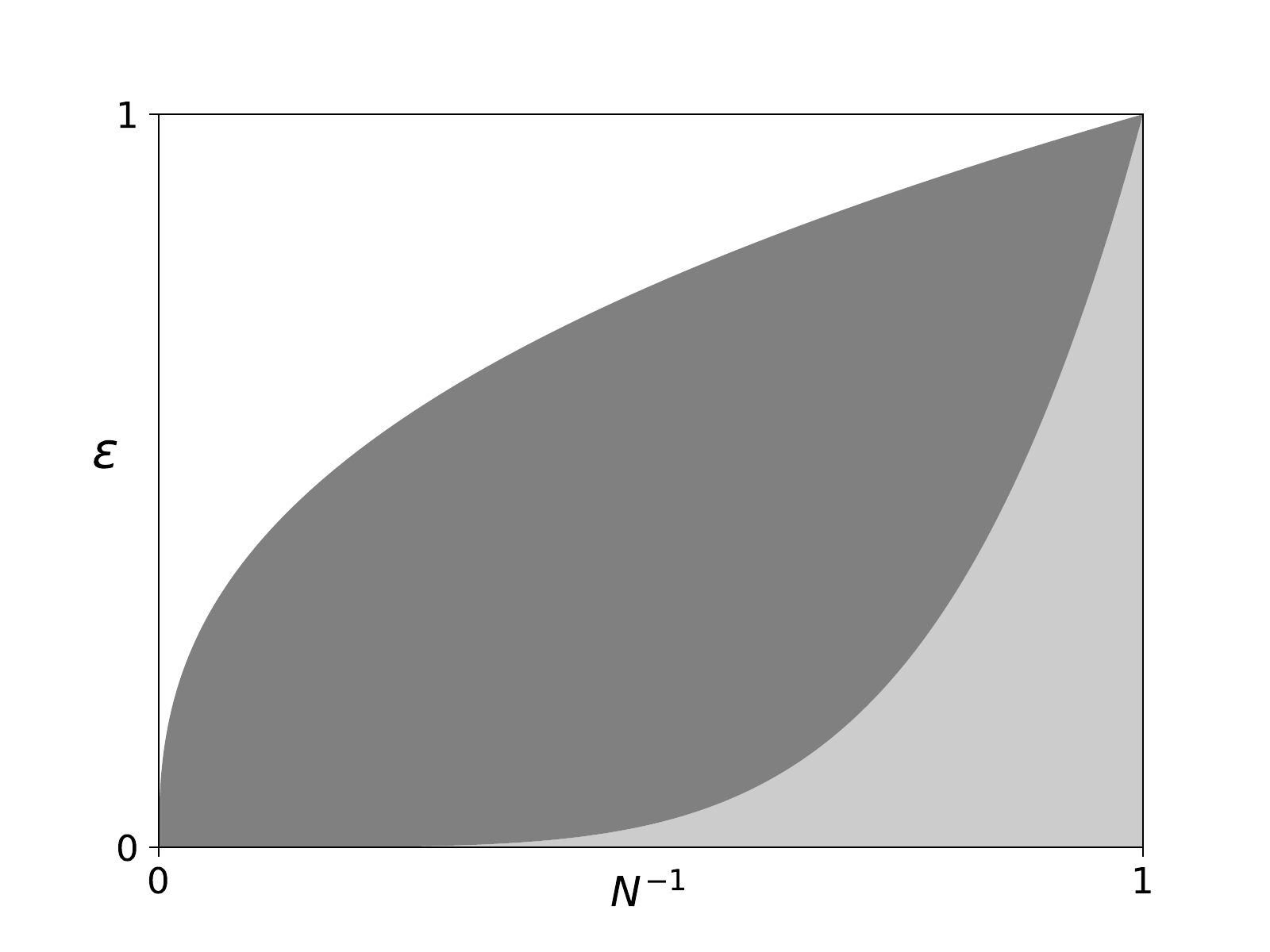}
    \caption{$\beta=\frac56$}
  \end{subfigure}
  \begin{subfigure}[t]{.5\linewidth}
    \centering\includegraphics[scale=0.45]{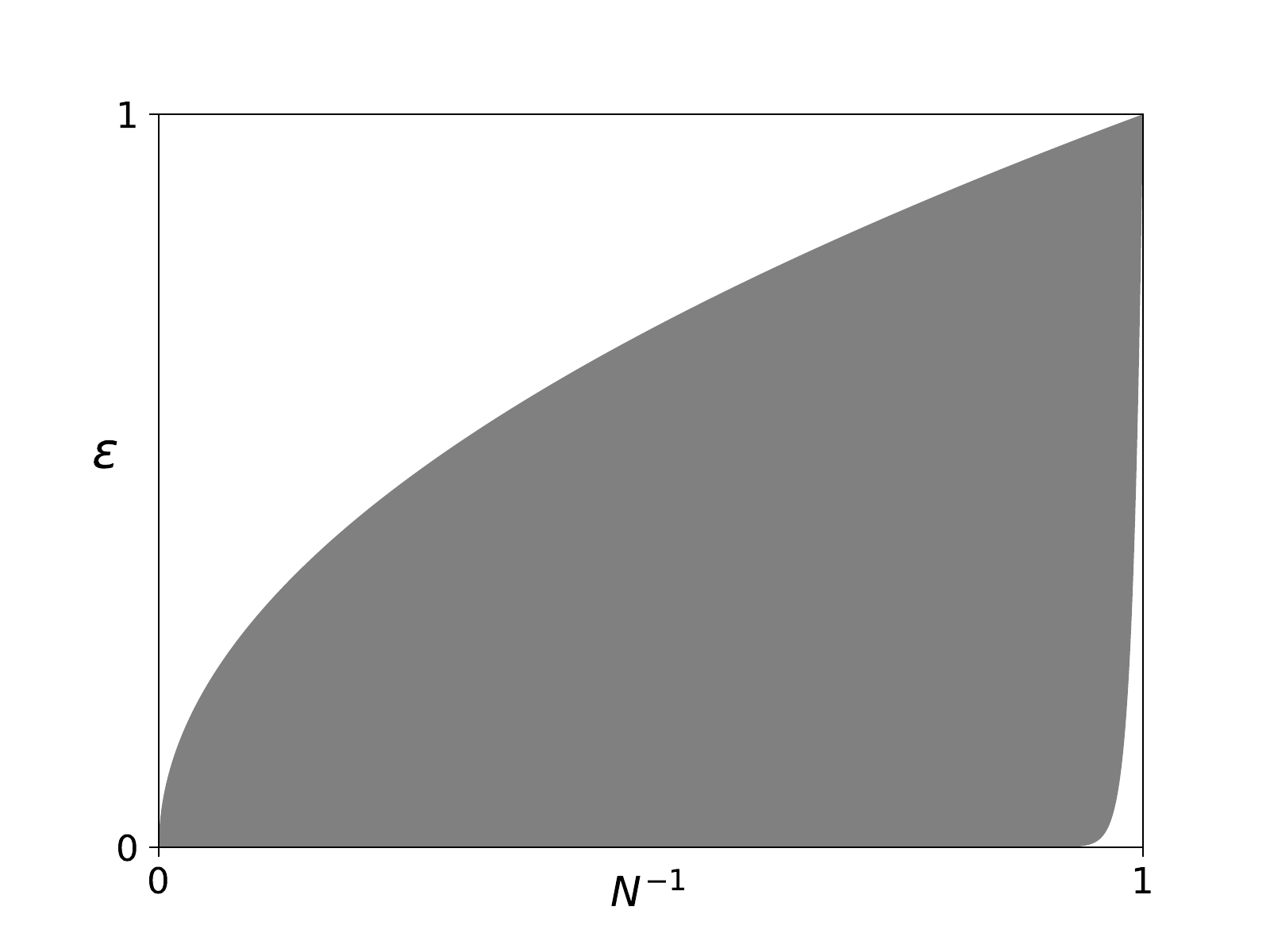}
    \caption{$\beta=1$}
  \end{subfigure}
  \caption{Best possible coverage of the parameter space $\N\times[0,1]$ for some exemplary choices of $\beta\in(0,1)$ and for $\beta=1$. We chose the least restrictive conditions satisfying Definition~\ref{def:admissible}, i.e., $(\Theta,\Gamma)_\beta=(\frac{3}{\beta}^-,\frac{1}{\beta})$ and $(\Theta,\Gamma)_1=(3,1^+)$. To make the moderate confinement condition $\Gamma=1^+$ for $\beta=1$ visible, we implemented it as $\Gamma=1.01$.
  Theorem~\ref{thm} applies in the dark grey area, while the white region is excluded from our analysis. In the light grey part, we expect the dynamics to be effectively described by a free evolution equation. Plotted with Matplotlib~\cite{matplotlib}.}
  \label{fig}
\end{figure}

As mentioned above, we will consider interactions in the NLS scaling regime $\beta\in(0,1)$  which are of a more generic form than~\eqref{interaction}.
\begin{definition}\label{def:W}
Let $\beta\in(0,1)$ and $\eta>0$. Define the set $\Wb$ as the set containing all families 
$$\wb:(0,1)\to L^\infty(\R^3,\R), \quad \mu\mapsto \wb,$$ 
such that for any $\mu\in(0,1)$
$$
\begin{cases} 
	(a)\;\,\norm{\wb}_{L^\infty(\R^3)}\ls \mu^{1-3\beta},
	\vspace{0.2cm}\\\vspace{0.2cm}
	(b)\;\, \wb \text{ is non-negative and spherically symmetric},\\\vspace{0.2cm}
	(c)\;\, \varrho_\beta:=\mathrm{diam}(\supp{\wb})\sim\mu^\beta ,\\
	(d)\;\, 	\lim\limits_{(N,\varepsilon)\to(\infty,0)}\mu^{-\eta}
	\left|\bNe(\wb)-	\lim\limits_{(N,\varepsilon)\to(\infty,0)}\bNe(\wb)\right|	=0,
\end{cases}$$
where
$$\bNe(\wb):= N\int\limits_{\R^3}\wb(z)\d z\int\limits_{\R}|\chie(y)|^4\d y=\mu^{-1}\int\limits_{\R^3}\wb(z)\d z\int\limits_{\R}|\chi(y)|^4\d y\,.$$
In the sequel, we will abbreviate  $\bNe(\wb)\equiv \bNe$.
\end{definition}
Condition (d) in Definition~\ref{def:W} regulates how fast the $(N,\varepsilon)$-dependent coupling parameter $\bNe$ converges to its limit as $(N,\varepsilon)\to(\infty,0)$. For the special case~\eqref{interaction}, we find that $\bNe=\norm{w}_{L^1(\R^3)}\int_\R|\chi(y)|^4\d y$ is independent of $N$ and $\varepsilon$,
hence this interaction is contained in $\Wb$ for any choice of $\eta>0$.\\

\noindent Throughout the paper, we will use two notions of one-particle energies:
\begin{itemize}
\item The \emph{``renormalised'' energy per particle}: for $\psi\in\mathcal{D}(\Hb(t)^\frac12)$,
\begin{equation}\label{E^psi}
\Eb^\psi(t):=\tfrac{1}{N}\llr{\psi,\Hb(t)\psi}-\tfrac{E_0}{\varepsilon^2},
\end{equation}
where $E_0$ denotes the lowest eigenvalue of $-\tfrac{\d^2}{\d y^2}+V^\perp(y)$. 
\item The \emph{effective energy per particle}: for $\Phi\in H^1(\R^2)$ and $b\in\R$,
\begin{equation}\label{E^Phi}
\mathcal{E}_b^\Phi(t):=\lr{\Phi,\left(-\Delta_x+\Vp(t,(x,0))+\tfrac{b}{2}|\Phi|^2\right)\Phi}_{L^2(\R^2)}.
\end{equation}
\end{itemize}

\noindent We can now state our assumptions:

\begin{itemize}
\item[A1] \emph{Interaction potential.} 
\begin{itemize}
\item [\textbullet] $\beta\in(0,1)$: \; Let $\wb\in \Wb$ for some $\eta>0$. 
\item [\textbullet] $\beta=1$:\hspace{26pt} Let $\wm$ be given by~\eqref{int:GP} with $w\in L^\infty(\R^3,\R)$ spherically symmetric,\\ 
\hphantom{$\beta=1$:\hspace{26pt} }non-negative and with $\supp w\subseteq \{z\in\R^3:|z|\leq 1\}$.
\end{itemize}
\item[A2] \emph{Confining potential.} Let $V^\perp:\R\rightarrow\R$ such that $-\tfrac{\d^2}{\d y^2}+V^\perp$ is self-adjoint and has a non-degenerate ground state $\chi$ with energy $E_0<\inf\sigma_\mathrm{ess}(-\Delta_y+V^\perp)$. 
Assume that the negative part of $V^\perp$ is bounded and that 
$\chi\in\mathcal{C}^2_\mathrm{b}(\R)$, i.e., $\chi$ is bounded and twice continuously differentiable with bounded derivatives. We choose  $\chi$ normalised and real. 
\item[A3] \emph{External field.} Let $\Vp:\R\times\R^3\rightarrow\R$ such that for fixed $z\in\R^3$, $\Vp(\cdot,z)\in\mathcal{C}^1(\R)$. 
Further, assume that for each fixed $t\in \R$, 
$\Vp(t,\cdot),\dot{\Vp}(t,\cdot)\in L^\infty(\R^3)\cap \mathcal{C}^1(\R^3)$ and $\partial_y\Vp(t,\cdot),\,\partial_y\dot{\Vp}(t,\cdot)\in L^\infty(\R^3)$.
\item[A4] \emph{Initial data.} Let $(N,\varepsilon)\to(\infty,0)$ be admissible and moderately confining with parameters $(\Theta,\Gamma)_\beta$ satisfying \eqref{values:Gamma:Theta}.
Assume that the family of initial data $\psi^{N,\varepsilon}_0\in\mathcal{D}(\Hb(0))\cap L^2_+(\R^{3N})$ with $\norm{\psi^{N,\varepsilon}_0}^2=1$, is close to a condensate with condensate wave function $\phe_0=\Phi_0\chie$ for some normalised $\Phi_0\in H^4(\R^2)$, i.e.,
\begin{equation}\label{A4:1}
\lim\limits_{(N,\varepsilon)\rightarrow(\infty,0)}\Tr_{L^2(\R^3)}\Big|\gamma^{(1)}_{\psi^{N,\varepsilon}_0}-|\Phi_0\chie\rangle\langle\Phi_0\chie|\Big|=0\,.
\end{equation}
Further, let
\begin{equation}\label{A4:2}
\lim\limits_{(N,\varepsilon)\rightarrow(\infty,0)}\left|E^{\psi^{N,\varepsilon}_0}_{\wb}(0)-\mathcal{E}^{\Phi_0}_{\bb}(0)\right|=0.
\end{equation}
\end{itemize}
In our main result, we prove the persistence of condensation in the state $\phe(t)=\Phi(t)\chie$ for initial data $\psiNe_0$ from \emph{A4}. Naturally, we are interested in times for which the condensate wave function $\Phi(t)$ exists, and, moreover, we require $H^4(\R^2)$-regularity of $\Phi(t)$ for the proof.
Let us therefore introduce the maximal time of $H^4(\R^2)$-existence,
\begin{equation}\label{def:T:ex}
\Tex:=\sup\left\{t\in\R_0^+: \norm{\Phi(t)}_{H^4(\R^2)}<\infty\right\},
\end{equation}
where $\Phi(t)$ is the solution of~\eqref{NLS} with initial datum $\Phi_0$ from \emph{A4}.

\begin{remark}\label{rem:existence:solutions}
The regularity of the initial data is for many choices of $\Vp$ propagated by the evolution~\eqref{NLS}. For several classes of external potentials, global existence in $H^4(\R^2)$-sense and explicit bounds on the growth of $\norm{\Phi(t)}_{H^4(\R^2)}$ are known:
\begin{itemize}
\item The case without external field,  $\Vp=0$, was covered in~\cite[Corollary 1.3]{sohinger2012}: for initial data $\Phi_0\in H^k(\R^2)$ with $k>0$, there exists  $C_k>0$ depending on $\norm{\Phi_0}_{H^k(\R^2)}$  such that 
$$\norm{\Phi(t)}_{H^k(\R^2)}\leq C_k(1+|t|)^{\frac47k^+}\norm{\Phi_0}_{H^k(\R^2)}$$ 
for all $t\in\R$.
If the initial data are further restricted to the set $$\Sigma^k:=\bigg\{f\in L^2(\R^2): \norm{f}_{\Sigma^k}:=\sum_{|\alpha|+|\beta|\leq k}\norm{x^\alpha\partial_x^\beta f}_{L^2(\R^2)}<\infty\bigg\}\subset H^k(\R^2)\,,$$ the bound is even uniform in $t\in\R$. This is, for $\Phi_0\in\Sigma^k$,  there exists $C>0$ such that 
$$\norm{\Phi(t)}_{H^k(\R^2)}<C$$ 
for all $t\in\R$ \cite[Section 1.2]{carles2013}.
\item For time-dependent external potentials $\Vp(t,(x,0))$ that are at most quadratic in $x$ uniformly in time, global existence of $H^k(\R^2)$-solutions with double exponential growth was shown in~\cite[Corollary 1.4]{carles2013} for initial data $\Phi_0\in\Sigma^k$:

Assume that $\Vp(\cdot,(\cdot,0))\in L^\infty_\mathrm{loc}(\R\times\R^2)$ is real-valued such that the map $x\mapsto\Vp(t,(x,0))$ is $\mathcal{C}^\infty(\R^2)$, the map $x\mapsto V(t,(x,0))$ is $\mathcal{C}^\infty(\R^2)$ for almost all $t\in\R$, and
the map $t\mapsto\sup_{|x|\leq 1}|\Vp(t,(x,0))|$ is $L^\infty(\R)$. Moreover, let $\partial_x^\alpha\Vp(\cdot,(\cdot,0))\in L^\infty(\R\times\R^2)$ for all $\alpha\in\N^2$ with $|\alpha|\geq2$.
Let $\Phi_0\in\Sigma^k(\R^2)$ with $k\geq 2$. Then there exists a constant $C>0$ such that 
$$\norm{\Phi(t)}_{H^k(\R^2)}\leq C\e^{\e^{Ct}}$$
for all $t\in\R$. 
In case of a time-independent harmonic potential and initial data  $\Phi_0\in\Sigma^k$, this can be improved to an exponential rather than double exponential bound. Note, however, that unbounded potentials $\Vp(t,z)$ are excluded by assumption~\emph{A3}.  
\end{itemize}
\end{remark}

\begin{thm} \label{thm}
Let $\beta\in(0,1]$ and assume that the potentials $\wb$, $V^\perp$ and $\Vp$ satisfy \emph{A1 -- A3}.
Let $\psiNe_0$ be a family of initial data satisfying \emph{A4}, let $\psiNe(t)$ denote the solution of~\eqref{SE} with initial datum $\psiNe_0$, and let $\gamma^{(1)}_{\psiNe(t)}$ denote its one-particle reduced density matrix as in~\eqref{rdm}.
Then for any $0\leq T<\Tex$,
\begin{eqnarray}
\lim\limits_{(N,\varepsilon)\to(\infty,0)}\;\sup\limits_{t\in[0,T]}\,\Tr\left|\gamma^{(1)}_{\psiNe(t)}-|\Phi(t)\chie\rangle\langle\Phi(t)\chie|\right|&=&0,\label{eqn:thm:1}\\
\lim\limits_{(N,\varepsilon)\to(\infty,0)}\;\sup\limits_{t\in[0,T]}\,\left|E^{\psiNe(t)}_{\wb}(t)-\mathcal{E}_{\bb}^{\Phi(t)}(t)\right|&=&0\,,\label{eqn:thm:2}
\end{eqnarray}
where the limits are taken along the sequence from \emph{A4}.
Here, $\Phi(t)$ is the solution of~\eqref{NLS} with initial datum $\Phi(0)=\Phi_0$ from \emph{A4} and with coupling parameter
\begin{equation}\label{eqn:bb}
\bb:=\begin{cases}
	\displaystyle\lim\limits_{(N,\varepsilon)\to(\infty,0)}\bNe &\quad \text{ for }  \beta\in(0,1)\,,\\[10pt]
	\displaystyle 8\pi a\int_{\R}|\chi(y)|^4\d y &\quad \text{ for } \beta=1\,
\end{cases}
\end{equation}
with $\bNe$ from Definition~\ref{def:W} and with $a$ the scattering length of $w$ as defined in \eqref{eqn:def:a}.
\end{thm}

\begin{remark}\label{rem:assumptions}
\remit{
\item Due to assumptions \emph{A1--A3}, the Hamiltonian $\Hb(t)$ is for any $t\in\R$ self-adjoint on its time-independent domain $\mathcal{D}(\Hb)$. 
Since we assume continuity of $t\mapsto \Vp(t)\in\mathcal{L}(L^2(\R^3))$, \cite{griesemer2017} implies that the family $\left\{\Hb(t)\right\}_{t\in\R}$ generates a unique, strongly continuous, unitary time evolution that leaves $\mathcal{D}(\Hb)$ invariant.
By imposing the further assumptions on $\Vp$, we can control the growth of the one-particle energies and the interactions of the particles with the external potential.
Note that it is physically important to include time-dependent external traps, since this admits non-trivial dynamics even if the system is initially prepared in an eigenstate.
\item Assumption \emph{A4} states that the system is initially a Bose--Einstein condensate which factorises in a longitudinal and a transverse part. In~\cite[Theorems 1.1 and 1.3]{schnee2007}, Schnee and Yngvason prove that both parts of the assumption are fulfilled by the ground state of $\Hb(0)$ for $\beta=1$ and $\Vp(0,z)=V(x)$ with $V$ locally bounded and diverging as $|x|\to\infty$.
\label{rem:assumptions:gs}
\item Our proof yields an estimate of the rate of the convergence~\eqref{eqn:thm:1}, which is of the form
\begin{equation*}
\Tr\left|\gamma_{\psiNe(t)}^{(1)}-|\Phi(t)\chie\rangle\langle\Phi(t)\chie|\right|
\ls \Big(A(0)+R_{\beta,\Theta,\Gamma,\eta}(N,\varepsilon)\Big)^\frac12\e^{f(t)}
\end{equation*}
with
\begin{eqnarray*}
A(0)&:=&\left|E^{\psiNe_0}_{\wb}(0)-\mathcal{E}_{\bb}^{\Phi_0}(0)\right|+\left(\Tr\left|\gamma_{\psiNe_0}^{(1)}-|\Phi_0\chie\rangle\langle\Phi_0\chie|\right|\right)^\frac12\,,\\
R_{\beta,\Theta,\Gamma,\eta}(N,\varepsilon)&\ls & N^{-n_1}+\varepsilon^{n_2}+\left(\frac{\varepsilon^\Theta}{\mu}\right)^{n_3}+\left(\frac{\mu}{\varepsilon^\Gamma}\right)^{n_4}
\end{eqnarray*}
for some $n_1\mydots n_4>0$ and some function $f:\R\to\R$ which is bounded uniformly in both $N$ and $\varepsilon$.
The coefficients $n_1$ to $n_4$ can be recovered from the bounds in Propositions~\ref{prop:gamma^<} and~\ref{prop:gamma:GP} by optimising \eqref{eqn:rate:of:conv:1} and \eqref{eqn:rate:of:conv:2} over the free parameters and making use of Lemma~\ref{lem:equivalence}.
We do not expect this rate to be optimal.
}
\end{remark}
\bigskip

\begin{remark}\label{rem:admissibility}
The sequences $(N,\varepsilon)\to(\infty,0)$  covered by Theorem~\ref{thm} are restricted by admissibility and moderate confinement condition (Definition~\ref{def:admissible} and \eqref{values:Gamma:Theta}). 
To conclude this section, let us discuss these constraints:
\remit{
\item By \eqref{values:Gamma:Theta}, the weakest possible constraints are given by $(\Theta,\Gamma)_\beta=(\frac{3}{\beta}^-,\frac{1}{\beta})$ for $\beta\in(0,1)$ and $(\Theta,\Gamma)_1=(3,1^+)$ for $\beta=1$. 
Instead of choosing these least restrictive values, we present Theorem~\ref{thm} and all estimates in explicit dependence of the parameters $\Theta$ and $\Gamma$, making it more transparent where the conditions enter the proof.
Moreover, the rate of convergence  improves for more restrictive choices of the parameters $\Gamma$ and $\Theta$.
\item In~\cite{chen2013}, Chen and Holmer prove Theorem~\ref{thm} for the regime $\beta\in(0,\frac25)$ under 
different assumptions on the sequence $(N,\varepsilon)$. The subset of the parameter range $\mathbb{N}\times[0,1]$ covered by their analysis is visualised in Figure \ref{fig:CH}.\\
While no admissibility condition is required for their proof, they impose a moderate confinement condition which is equivalent to our condition for $\beta\in(0,\frac{3}{11}]$.  For larger $\beta\in(\frac{3}{11},\frac{2}{5})$, they restrict the parameter range much stronger%
\footnote{More precisely, Chen and Holmer consider sequences $(N,\varepsilon)$ such that $N\gg\varepsilon^{-2\nu(\beta)}$, where $\nu(\beta):=\max\left\{\frac{1-\beta}{2\beta},\frac{5\beta/4-1/12}{1-5\beta/2},\frac{\beta/2+5/6}{1-\beta},\frac{\beta+1/3}{1-2\beta}\right\}$. 
For the regime $\beta\in(0,\frac{3}{11}]$, this implies $\nu(\beta)=\frac{1-\beta}{2\beta}$, which is equivalent to the choice $\Gamma=\frac{1}{\beta}$ and thus exactly our moderate confinement condition. 
For $\beta\in(\frac{3}{11},\frac13]$, one obtains $\nu(\beta)=\frac{\beta+1/3}{1-2\beta}$, which corresponds to the choice $\Gamma=\frac{5}{3-6\beta}>\frac{1}{\beta}$, 
and for $\beta\in(\frac13,\frac25)$, one concludes $\nu(\beta)=\frac{5\beta/4-1/12}{1-5\beta/2}$, corresponding to $\Gamma=\frac{5}{6-15\beta}>\frac{1}{\beta}$. Since the moderate confinement condition is weaker for smaller $\Gamma$, we conclude that our condition is weaker for $\beta>\frac{3}{11}$.},
and their condition becomes so restrictive with increasing $\beta$ that it delimitates the range of scaling parameters to $\beta\in(0,\frac25)$.
\item 
No restriction comparable to the admissibility condition is needed for the ground state problem in~\cite{schnee2007}. 
Given the work~\cite{mehats2017} where the strong confinement limit of the three-dimensional NLS equation is taken, this suggests that our result should  hold without any such restriction. However, for the present proof, the condition is indispensable (see Remarks~\ref{rem:differences:NLS:1d} and~\ref{rem:differences:GP:1d}).
\item As argued above, the moderate confinement condition for $\beta\in(0,1)$ is optimal, in the sense that we expect a free evolution equation if $\mu^\beta/\varepsilon\to\infty$.
For $\beta=1$, we require that $\mu/\varepsilon^\Gamma\to0$ for $\Gamma>1$. Note that the choice $\Gamma=1$ would mean no restriction at all because $\mu/\varepsilon=N^{-1}$. Our proof works for $\Gamma$ that are arbitrarily close to $1$. However, since the estimates are not uniform in $\Gamma$, the case $\Gamma=1$ is excluded.
To our understanding, the constraint $\Gamma>1$ is purely technical. Note that such a restriction is neither required for the ground state problem in \cite{schnee2007}, nor in \cite{GP}, where the dynamics for cigar-shaped case with strong confinement in two directions is studied.
\label{rem:mod:conf}
\item Although no moderate confinement condition appears the cigar-shaped problem~\cite{GP}, our analysis covers a considerably larger subset of the parameter space $\mathbb{N}\times[0,1]$ than is included in~\cite{GP}.
In that work, the admissibility condition is given as $N\varepsilon^{\frac{2}{5}^-}\to 0$, which is much more restrictive than our condition.
}
\end{remark}
\begin{figure}[t]
  \begin{subfigure}[t]{.5\linewidth}
    \centering\includegraphics[scale=0.45]{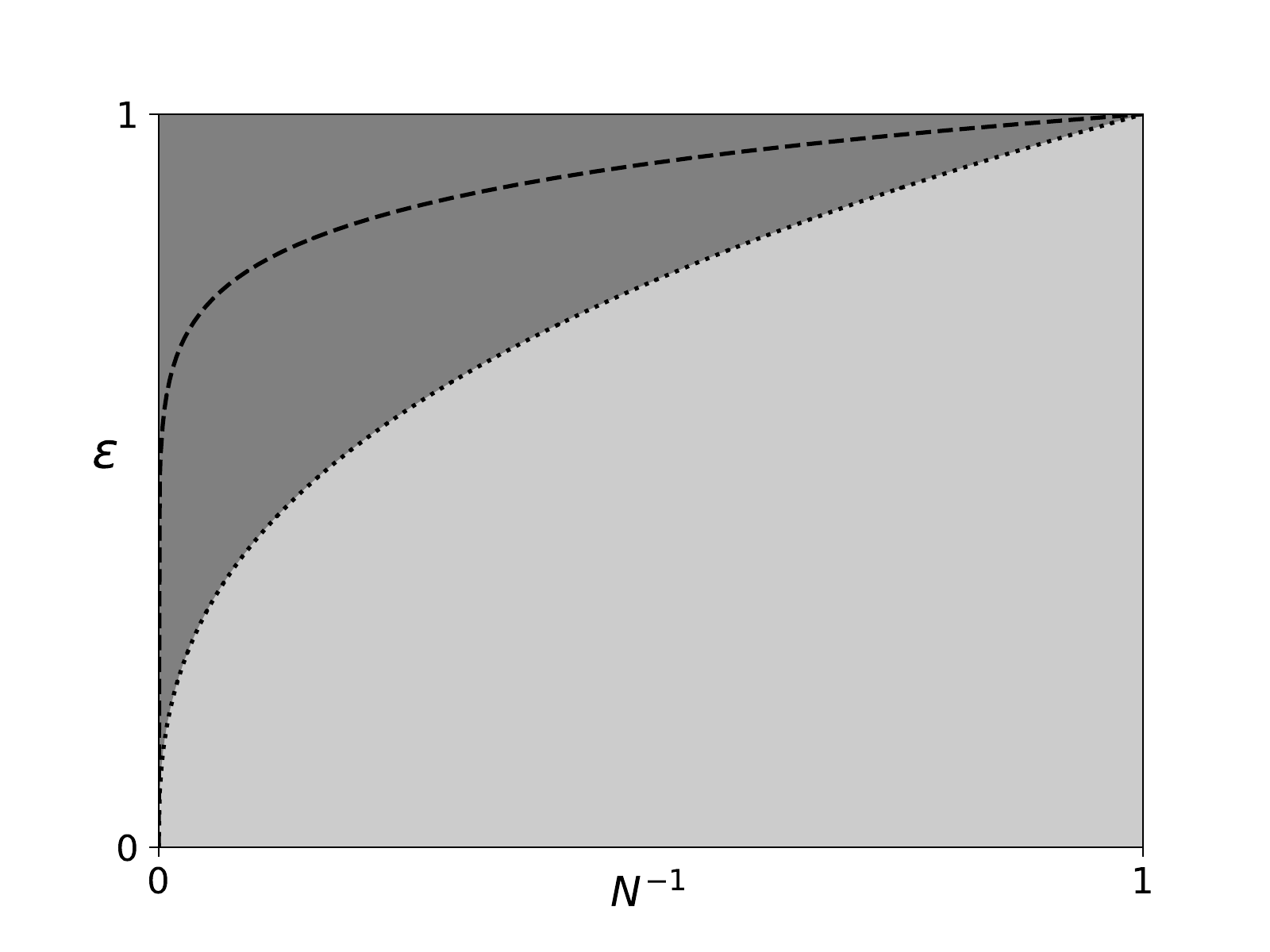}
    \caption{$d=2$, $\beta=\frac{3}{11}$}
  \end{subfigure}
  \begin{subfigure}[t]{.5\linewidth}
    \centering\includegraphics[scale=0.45]{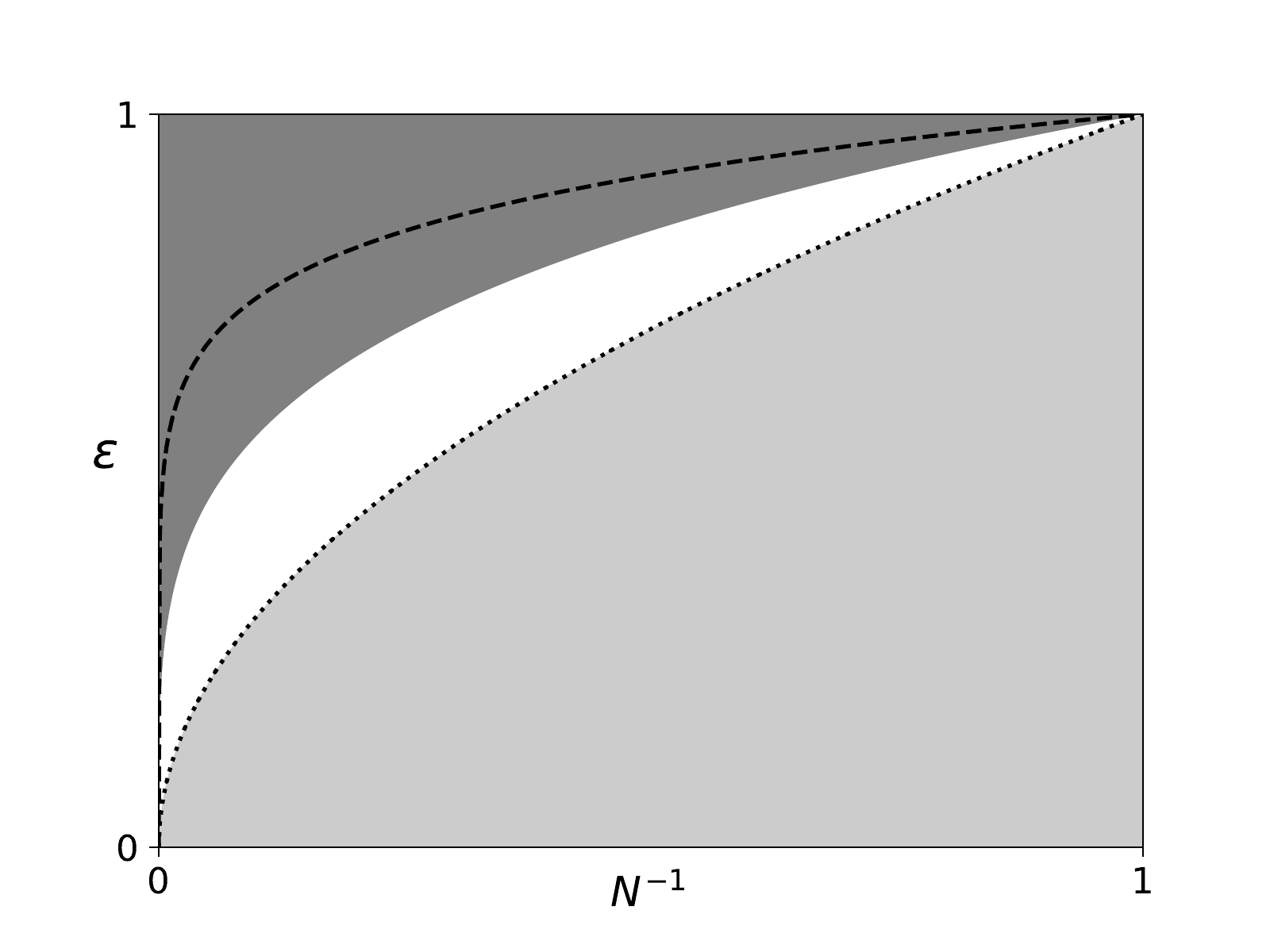}
    \caption{$d=2$, $\beta=\frac{1}{3}$}
  \end{subfigure}
  \begin{subfigure}[t]{.5\linewidth}
    \centering\includegraphics[scale=0.45]{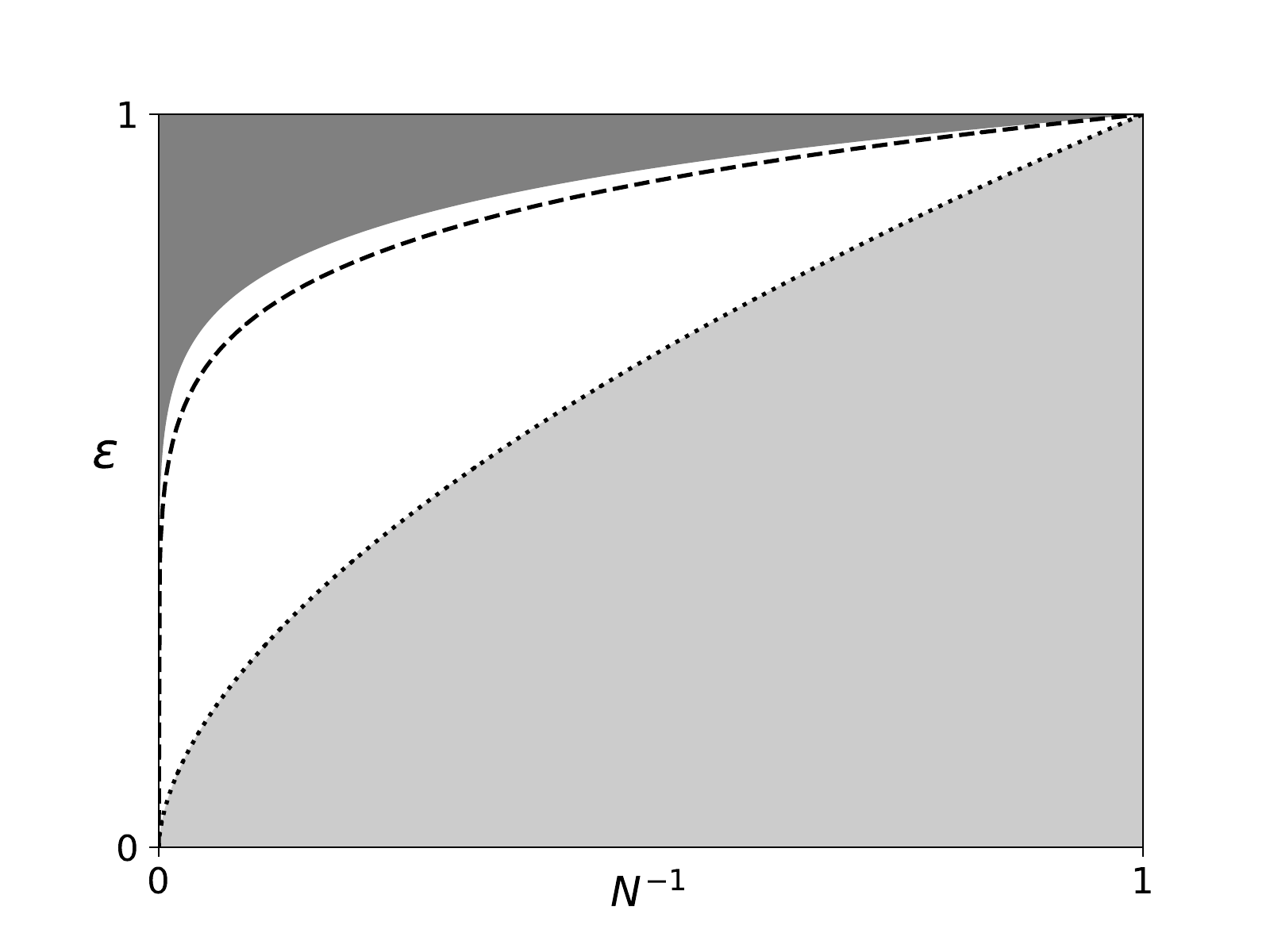}
    \caption{$d=2$, $\beta=\frac{11}{30}$}
  \end{subfigure}
  \begin{subfigure}[t]{.5\linewidth}
    \centering\includegraphics[scale=0.45]{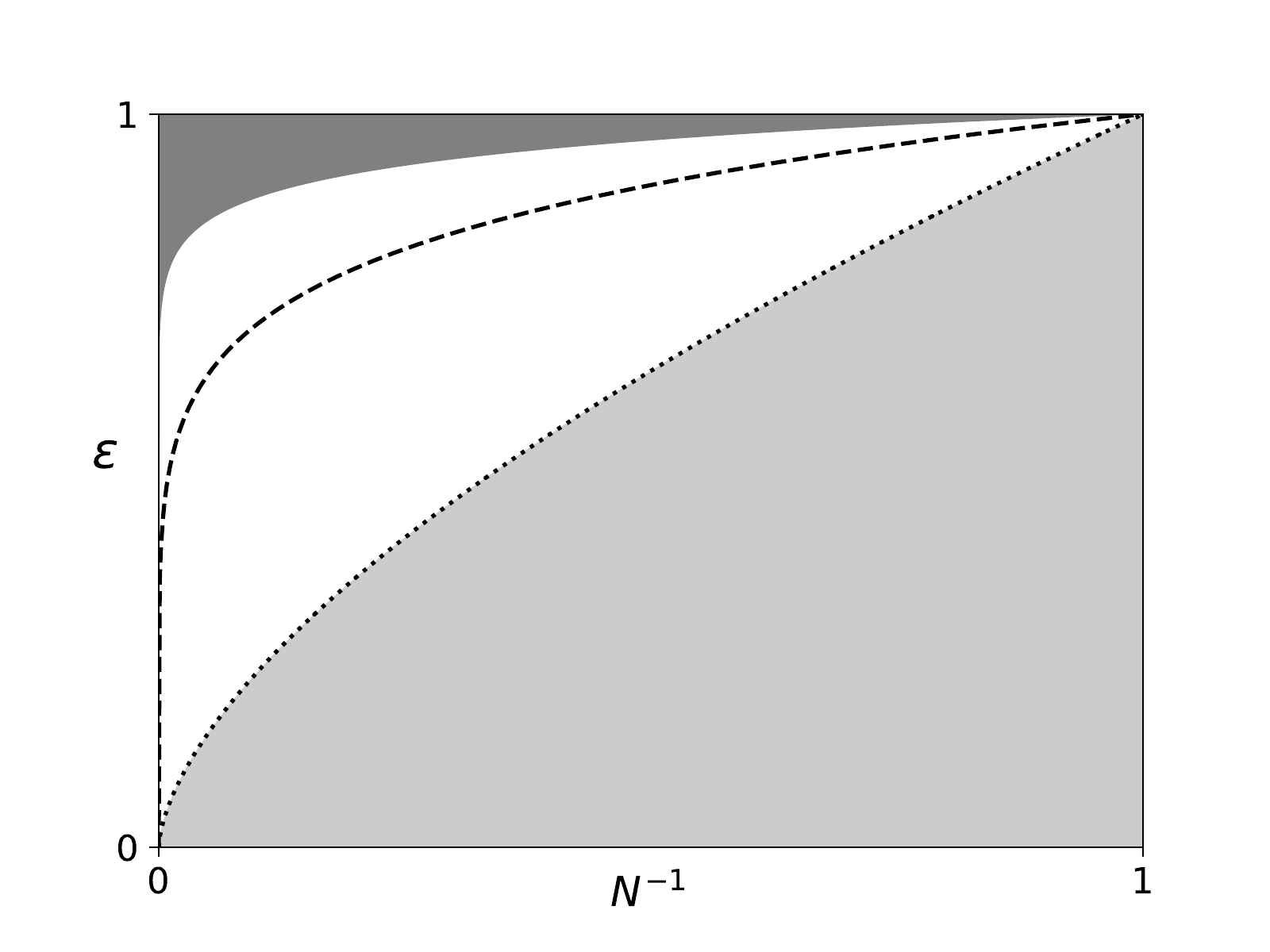}
    \caption{$d=2$, $\beta=\frac{23}{60}$}
  \end{subfigure}
  \caption{Coverage of the parameter space $\N\times[0,1]$ for some exemplary choices of $\beta\in(0,\frac25)$.
	In \cite{chen2013}, Chen and Holmer cover sequences within the dark grey region, while the white and light grey area
	are excluded. In comparison, Theorem \ref{thm} applies to all sequences enclosed between the black dashed line and the black dotted line, where the dashed line corresponds to the admissibility and the dotted line to the moderate confinement condition. Limiting sequences within the light grey region are expected to yield a free effective evolution equation. Plotted with Matplotlib \cite{matplotlib}.}
  \label{fig:CH}
\end{figure}

\section{Proof of the main result}\label{sec:proof}
The proof of Theorem~\ref{thm}, both for the NLS scaling $\beta\in(0,1)$ and the Gross--Pitaevskii case $\beta=1$, follows the approach developed by Pickl in~\cite{pickl2015}.
The main idea is to avoid a direct estimate of the differences in~\eqref{eqn:thm:1} and~\eqref{eqn:thm:2}, but instead to define a functional 
$$\alwb:\,\R\,\times\,L^2(\R^{3N})\,\times\, L^2(\R^3)\,\to\,\R_0^+,\qquad (t,\psiNe(t),\phe(t))\mapsto\alwb(t,\psiNe(t),\phe(t))$$
in such a way that 
$$\lim\limits_{(N,\varepsilon)\to(\infty,0)}\alwb(t,\psiNe(t),\phe(t))= 0\quad \Longleftrightarrow \quad \eqref{eqn:thm:1} \wedge \eqref{eqn:thm:2}.$$
Physically, the functional $\alwb$ provides a measure of the relative number of particles that remain outside the condensed phase $\phe(t)$, and is therefore also referred to as a counting functional.
The index $\wb$ indicates that the evolutions of $\psiNe(t)$ and $\phe(t)$ are generated by $\Hb(t)$ and $\hb(t)$, which depend, directly or indirectly, on the interaction $\wb$.
To define the functional $\alwb$, we recall the projectors onto the condensate wave function that were introduced in~\cite{pickl2008,keler2016}:
\begin{definition}\label{def:p}
Let $\phe(t)=\Phi(t)\chie$, where $\Phi(t)$ is the solution of the NLS equation \eqref{NLS} with initial datum $\Phi_0$ from \emph{A4} and with $\chie$ as in \eqref{eqn:chie}. Let
$$p:=\ket{\phe(t)}\bra{\phe(t)},$$
where we drop the $t$- and $\varepsilon$\,-dependence of $p$  in the notation. For $j\in\{1,\dots,N\}$, define the projection operators on $L^2(\R^{3N})$
$$p_j:=\underbrace{\mathbbm{1}\otimes\cdots\otimes\mathbbm{1}}_{j-1}\otimes\, p\otimes \underbrace{\mathbbm{1}\otimes\cdots\otimes\mathbbm{1}}_{N-j} \quad\text{and}\quad q_j:=\mathbbm{1}-p_j.$$
Further, define the orthogonal projections on $L^2(\R^3)$ 
\begin{align*}
\pp&:=\ket{\Phi(t)}\bra{\Phi(t)}\otimes\mathbbm{1}_{L^2(\R)},  &  \qp&:=\mathbbm{1}_{L^2(\R^3)}-\pp,\\
\pc&:=\mathbbm{1}_{L^2(\R^2)}\otimes\ket{\chie}\bra{\chie}, &  \qc&:=\mathbbm{1}_{L^2(\R^3)}-\pc,
\end{align*}
and define $\pp_j$, $\qp_j$, $\pc_j$ and $\qc_j$ on $L^2(\R^{3N})$ analogously to $p_j$ and $q_j$.
Finally, for $0\leq k\leq N$, define the many-body projections
$$ P_k=\big(q_1\cdots q_k\, p_{k+1}\cdots p_N\big)_\mathrm{sym}:=\sum\limits_{\substack{J\subseteq\{1,\dots,N\}\\|J|=k}}\prod\limits_{j\in J}q_j\prod\limits_{l\notin J}p_l $$
and $P_k=0$ for $k<0$ and $k>N$.
Further, for any function $f: \mathbb{N}_0\rightarrow\R_0^+$ and $d\in\mathbb{Z}$, define the operators $\hat{f},\hat{f}_d\in\mathcal{L}\left(L^2(\R^{3N})\right)$ by 
\begin{equation}\label{eqn:weighted:many:body:operators}
\hat{f}:= \sum\limits_{k=0}^N f(k)P_k\,,\qquad \hat{f}_d:=\sum\limits_{j=-d}^{N-d} f(j+d)P_j\,.
\end{equation}
\end{definition}
Clearly, $\sum_{k=0}^N P_k=\mathbbm{1}$. Besides, note the useful relations $p=\pp\pc$,  $\qp q=\qp$, $\qc q=\qc$  and  $q=\qc+\qp\pc=\qp+\pp \qc$.
In the sequel, we will make use of the following weight functions:
\begin{definition}\label{def:weights}
Define $$n:\N_0\to\R_0^+, \qquad k\mapsto n(k):=\sqrt{\tfrac{k}{N}}\,,$$
and, for some $\xi\in(0,\frac12)$,
$$ m:\N\to\R_0^+, \qquad
m(k):=\begin{cases}
	\displaystyle n(k) & \text{for } k\geq N^{1-2\xi}\,,\\[2pt]
	\displaystyle\tfrac12\left(N^{-1+\xi}k+N^{-\xi}\right) & \text{else}\,.
	\end{cases}$$ 
Further, define the weight functions $m^\sharp:\N_0\to\R^+_0$, $\sharp\in\{a,b,c,d,e,f\}$, by
\begin{equation}\label{eqn:weights}
\begin{array}{ll}
m^a(k)\,:=\,m(k)-m(k+1), &\quad  m^b(k)\,:=\,m(k)-m(k+2),\\
m^c(k)\,:=\,m^a(k)-m^a(k+1), \qquad &\quad m^d(k)\,:=\,m^a(k)-m^a(k+2),\\
m^e(k)\,:=\,m^b(k)-m^b(k+1), \qquad &\quad m^f(k)\,:=\,m^b(k)-m^b(k+2).\\
\end{array}
\end{equation}
The corresponding weighted many-body operators in the sense of \eqref{eqn:weighted:many:body:operators} are denoted by $\hat{m}^\sharp$.
Finally, define
$$\hat{r}:=\hat{m}^b p_1 p_2+\hat{m}^a(p_1 q_2+q_1p_2).$$
\end{definition}
Note that $m$ equals $n$ with a smooth, $\xi$-dependent cut-off. This modification of the weight $n$ is a technical trick that enables us to estimate expressions of the form $\onorm{\hat{f}-\hat{f}_d}$ for $\hat{f}_d$ as in \eqref{eqn:weighted:many:body:operators}, which appear at many points in the proof. The difference $\hat{f}-\hat{f}_d$ can be understood as operator that is weighted, in the sense of \eqref{eqn:weighted:many:body:operators}, with the derivative $\frac{\d f}{\d k}$. For the choice $f(k)=n(k)$, this derivative diverges as $k\to 0$, whereas the cut-off $\xi$ softens this singularity for small $k$ such that one finds $\onorm{\hat{m}-\hat{m}_d}\ls N^{-1+\xi}$ for the choice $f(k)=m(k)$ (Lemma~\ref{lem:l:2}).
\begin{definition}\label{def:alpha^<} For $\beta\in(0,1)$, define
\begin{equation*}
\alwb(t):=\alwb(t,\psi^{N,\varepsilon}(t),\phe(t)):=\llr{\psi^{N,\varepsilon}(t),\hat{m}\psi^{N,\varepsilon}(t)}+\big|\Eb^{\psi^{N,\varepsilon}(t)}(t)-\Ecal^{\Phi(t)}(t)\big|\,.
\end{equation*}
\end{definition}
The expression $\llr{\psiNe(t),\hat{m}\psiNe(t)}$ is a suitably weighted sum of the expectation values of $P_k\psiNe(t)$.
As $m(0)\approx0$ and $m$ is increasing, the parts of $\psiNe(t)$ with more particles outside $\phe(t)$ contribute more to $\alwb(t)$.
It is well known that $\llr{\psiNe(t),\hat{m}\psiNe(t)}\to0$ is equivalent to the convergence~\eqref{eqn:thm:1} of the one-particle reduced density matrix, 
hence  $\alwb(t)\to0$ is equivalent to~\eqref{eqn:thm:1} and~\eqref{eqn:thm:2}. The relation between the respective rates of convergence is stated in the following lemma, whose proof is given in~\cite[Lemma 3.6]{NLS}:
\begin{lem}\label{lem:equivalence}
For any  $t\in[0,\Tex)$, it holds that 
$$
\Tr\left|\gamma_{\psiNe(t)}^{(1)}-|\phe(t)\rangle\langle\phe(t)|\right|\leq \sqrt{8\alwb(t)}\,,$$
$$
\alwb(t)\leq \left|\Eb^{\psiNe(t)}(t)-\Ecal^{\Phi(t)}(t)\right|+\sqrt{\Tr\left|\gamma_{\psiNe(t)}^{(1)}-|\phe(t)\rangle\langle\phe(t)|\right|}+\tfrac12 N^{-\xi}\,.
$$
\end{lem}

\subsection{The NLS case $\beta\in(0,1)$}\label{subsec:NLS}
The strategy of our proof is to derive a bound for $|\tfrac{\d}{\d t}\alwb(t)|$, which leads to an estimate of $\alwb(t)$ by means of Grönwall's inequality. The first step is therefore to compute this derivative.
\begin{prop}\label{prop:alpha^<}
Assume \emph{A1 -- A4} for $\beta\in(0,1)$. Let
$$\wbot :=\wb(z_1-z_2) \quad \text{ and }\quad Z_\beta^{(12)}:=\wbot -\tfrac{\bb}{N-1}\left(|\Phi(t,x_1)|^2+|\Phi(t,x_2)|^2\right)$$
and define
\begin{equation}\label{eqn:mathcal:L}
\mathcal{L}:=\left\{N\hat{m}^a_{-1},\,N\hat{m}^b_{-2}\right\}
\end{equation}
for $\hat{m}^a_{-1}$ and $\hat{m}^b_{-2}$ as defined in \eqref{eqn:weighted:many:body:operators} and \eqref{eqn:weights}.
Then
\begin{eqnarray*}
\left|\tfrac{\d}{\d t}\alwb(t) \right|&\leq& \big|\gamma_{a,<}(t)\big|+|\gamma_{b,<}(t)| 
\end{eqnarray*}
for almost every $t\in[0,\Tex)$,
where
\begin{eqnarray}
	\hspace{-0.4cm}\gamma_{a,<}(t)&:=&\Big|\llr{\psi^{N,\varepsilon}(t),\dot{\Vp}(t,z_1)\psi^{N,\varepsilon}(t)}-\lr{\Phi(t),\dot{\Vp}(t,(x,0))\Phi(t)}_{L^2(\R^2)}\Big|\label{gamma_a:1}\\
	&&-2N\Im\llr{\psi^{N,\varepsilon}(t),\hat{m}^a_{-1}q_1\big(\Vp(t,z_1)-\Vp(t,(x_1,0))\big)p_1\psi^{N,\varepsilon}(t)}, \label{gamma_a:2}\\\nonumber\\
	\hspace{-0.4cm}\gamma_{b,<}(t)&:=&- N(N-1)\Im\llr{\psi^{N,\varepsilon}(t),Z_\beta^{(12)}\hat{m}\psi^{N,\varepsilon}(t)}
\label{gamma_b},\\
&\;\,=:&\gamma_{b,<}^{(1)}(t)+\gamma_{b,<}^{(2)}(t)+\gamma_{b,<}^{(3)}(t)+\gamma_{b,<}^{(4)}(t)\,,
\nonumber\\
\nonumber\\
&&\hspace{-3cm}\text{with}\nonumber\\\nonumber\\
	\hspace{-0.4cm}\big|\gamma_{b,<}^{(1)}(t)\big|&:=&N\max\limits_{\hat{l}\,\in \mathcal{L}}  \left|\llr{\psi^{N,\varepsilon}(t),\hat{l}\qp_1\pc_1p_2Z_\beta^{(12)}p_1p_2\psi^{N,\varepsilon}(t)}\right|,\label{gamma_b_1}\\
\nonumber\\\
	\hspace{-0.4cm}\big|\gamma_{b,<}^{(2)}(t)\big|&:=&N\max\limits_{\hat{l}\,\in \mathcal{L}}\;\max\limits_{t_2\in\{p_2,\,q_2,\,\qp_2\pc_2\}}\left|\llr{\psi^{N,\varepsilon}(t),
	\qc_1t_2\hat{l}\wbot  p_1p_2\psi^{N,\varepsilon}(t)} \right|
		\label{gamma_b_2:1}\\
		&&+N\max\limits_{\hat{l}\,\in \mathcal{L}}\left|\llr{\psi^{N,\varepsilon}(t),\qc_1q_2\hat{l}\wbot p_1\qc_2\psi^{N,\varepsilon}(t)}\right|\label{gamma_b_2:2}\\
		&&+N\max\limits_{\hat{l}\,\in \mathcal{L}}\left|\llr{\psi^{N,\varepsilon}(t),\qc_2\qp_1\pc_1\hat{l}\wbot p_1\qc_2\psi^{N,\varepsilon}(t)}\right|\label{gamma_b_2:3}\\
		&&+N\max\limits_{\hat{l}\,\in \mathcal{L}}\left|\llr{\psi^{N,\varepsilon}(t),\qc_1\qc_2\hat{l}\wbot p_1\pc_2\qp_2\psi^{N,\varepsilon}(t)}\right|\label{gamma_b_2:4}\,,\\
\nonumber\\\
	\hspace{-0.4cm}\big|\gamma_{b,<}^{(3)}(t)\big|&:=&N\max\limits_{\hat{l}\,\in \mathcal{L}}\left|\llr{\psi^{N,\varepsilon}(t),(\qc_2\qp_1\pc_1+\qc_1\qp_2\pc_2)\hat{l}\wbot p_1\pc_2\qp_2\psi^{N,\varepsilon}(t)}\right|\label{gamma_b_3:1}\\
		&&+N\max\limits_{\hat{l}\,\in \mathcal{L}}\left|\llr{\psi^{N,\varepsilon}(t),\qp_1\qp_2\pc_1\pc_2\hat{l}\wbot p_2\qc_1\psi^{N,\varepsilon}(t)}\right|\label{gamma_b_3:2}\,,\\
	\nonumber\\\
	\hspace{-0.4cm}\big|\gamma_{b,<}^{(4)}(t)\big|& := &N\max\limits_{\hat{l}\,\in \mathcal{L}}\left|\llr{\psi^{N,\varepsilon}(t),\qp_1\qp_2\hat{l}\pc_1\pc_2\wbot p_1p_2\psi^{N,\varepsilon}(t)} \right| \label{gamma_b_4:1}\\
		&&+N\max\limits_{\hat{l}\,\in \mathcal{L}}\left|\llr{\psi^{N,\varepsilon}(t),\qp_1\qp_2\hat{l}\pc_1\pc_2\wbot p_1\pc_2\qp_2\psi^{N,\varepsilon}(t)}\right|\label{gamma_b_4:2}\\
		&&+\bb\max\limits_{\hat{l}\,\in \mathcal{L}}\left|\llr{\psi^{N,\varepsilon}(t),q_1q_2\hat{l}|\Phi(t,x_1)|^2p_1q_2\psi^{N,\varepsilon}(t)}\right|.\label{gamma_b_4:3}
\end{eqnarray}
\end{prop}
The term $\gamma_{a,<}$ summarises all contributions from interactions between the particles and the external field $\Vp$, while $\gamma_{b,<}$ collects all contributions from the mutual interactions between the bosons. 
The latter can be subdivided into four  parts: \begin{itemize}
\item $\gamma_{b,<}^{(1)}$ and $\gamma_{b,<}^{(4)}$ contain the quasi two-dimensional interaction $\wbbar(x_1-x_2)$ resulting from integrating out the transverse degrees of freedom in $\wb$, which is given as
$$\pc_1\pc_2\wb(z_1-z_2)\pc_1\pc_2=:\wbbar(x_1-x_2)\pc_1\pc_2$$ 
(see Definition~\ref{def:bar}).
Hence, $\gamma_{b,<}^{(1)}$ and $\gamma_{b,<}^{(4)}$ can be understood as two-dimensional analogue of the corresponding expressions in the three-dimensional problem without confinement~\cite[Lemma A.4]{pickl2015}, and the estimates are inspired by~\cite{pickl2015}. 
Note that $\gamma_{b,<}^{(1)}$ contains the difference between the quasi two-dimensional interaction potential $\wbbar$ and the effective one-body potential $\bb|\Phi(t)|^2$, which means that it vanishes in the limit $(N,\varepsilon)\to(\infty,0)$ only if~\eqref{NLS} with coupling parameter $\bb$ is the correct effective equation.
The last line~\eqref{gamma_b_4:3} of $\gamma_{b,<}^{(4)}$ contains merely the effective interaction potential $\bb|\Phi(t)|^2$ instead of the pair interaction $\wb$, hence, it is easily controlled.
\item $\gamma_{b,<}^{(2)}$ and $\gamma_{b,<}^{(3)}$ are remainders from the replacement $\wb\to\wbbar$, hence they have no three-dimensional equivalent. They are comparable to the expression $\gamma_b^{(2)}$ in~\cite{NLS} from the analogous replacement of the originally three-dimensional interaction by its quasi one-dimensional counterpart.
\end{itemize}
The second step is to control $\gamma_{a,<}$ to $\gamma_{b,<}^{(4)}$ in terms of $\alwb(t)$ and by expressions that vanish in the limit $(N,\varepsilon)\to(\infty,0)$. To write the estimates in a more compact form, let us define the function $\efrak:[0,\Tex)\rightarrow [1,\infty)$ as
\begin{equation}\label{def:e}\begin{split}
\efrak^2(t):=\norm{\Phi(t)}^2_{H^4(\R^2)}+|\Eb^{\psi^{N,\varepsilon}_0}(0)|+|\Ecal^{\Phi_0}(0)|&+\int\limits_0^t \norm{\dot{\Vp}(s)}_{L^\infty(\R^3)}\d s\\
&+\sup\limits_{i,j\in\{0,1\}}\norm{\partial_t^i\partial_{y}^j\Vp(t)}_{L^\infty(\R^3)}\,,
\end{split}\end{equation}
where $\Phi(t)$ denotes the solution of~\eqref{NLS} with initial datum $\Phi_0$ from \emph{A4}.
Note that $\efrak(t)$ is bounded uniformly in $N$ and $\varepsilon$ because the only $(N,\varepsilon)$-dependent quantity $\Eb^{\psi^{N,\varepsilon}_0}(0)$ converges to $\Ecal^{\Phi_0}(0)$ as $(N,\varepsilon)\to(\infty,0)$ by \emph{A4}.
The function $\efrak$ is particularly useful since
$$
\big|\Eb^{\psi^{N,\varepsilon}(t)}(t)\big|\leq \efrak^2(t)-1 \quad\text{ and }\quad \big|\Ecal^{\Phi(t)}(t)\big|\leq \efrak^2(t)-1 
$$
for any $t\in[0,\Tex)$ by the fundamental theorem of calculus. 
Note that for a time-independent external field $\Vp$, $\efrak^2(t)\ls 1$ as a consequence of Remark~\ref{rem:existence:solutions}, hence
$\Eb^{\psi^{N,\varepsilon}(t)}(t)$ and $\Ecal^{\Phi(t)}(t)$ are in this case bounded uniformly in $t\in[0,\Tex)$.\\

Recall that by assumption A4, we consider sequences $(N,\varepsilon)$ that are $(\Theta,\Gamma)_\beta$-admissible with $\Gamma_\beta=1/\beta$ and $\Theta_\beta\in(1/\beta,3/\beta)$.
To make a clear distinction between the cases $\beta\in(0,1)$ and $\beta=1$, let us define 
$$\delta:=\beta\Theta_\beta\in(1,3)\,,$$
i.e., we consider sequences with 
$$(\Theta,\Gamma)_\beta=(\tfrac{\delta}{\beta},\tfrac{1}{\beta})\,.$$

\begin{prop}\label{prop:gamma^<}
Let $\beta\in(0,1)$ and assume \emph{A1 -- A4} with parameters $\beta$ and $\eta$ in \emph{A1} and $(\Theta,\Gamma)_\beta=(\frac{\delta}{\beta},\frac{1}{\beta})$ in~\emph{A4}.
Let 
$$0<\xi<\min\left\{\tfrac13\,,\,\tfrac{1-\beta}{2}\,,\,\beta\,,\,\tfrac{\beta(3-\delta)}{2(\delta-\beta)}\right\}\,,\quad 
0<\sigma<\min\left\{\tfrac{1-3\xi}{4}\,,\,\beta-\xi\right\}\,.
$$ 
Then, for sufficiently small $\mu$, the terms $\gamma_{a,<}$ to $\gamma_{b,<}^{(4)}$ from Proposition~\ref{prop:alpha^<} are bounded by
\begin{eqnarray*}
	\big|\gamma_{a,<}(t)\big| & \ls &\efrak^3(t)\,\varepsilon+\efrak(t)\llr{\psi^{N,\varepsilon}(t),\hat{n}\psi^{N,\varepsilon}(t)},\\[4pt]
	\big|\gamma_{b,<}^{(1)}(t)\big| & \ls&\efrak^2(t)\left(\tfrac{\mu^\beta}{\varepsilon}+N^{-1}+\mu^\eta\right) ,\\[2pt]
	\big|\gamma_{b,<}^{(2)}(t)\big| & \ls &\efrak^3(t)\left(\left(\tfrac{\varepsilon^\delta}{\mu^\beta}\right)^{\frac{\xi}{\beta}+\frac12}+\varepsilon^\frac{1-\beta}{2}\right),\\[2pt]
	\big|\gamma_{b,<}^{(3)}(t)\big| & \ls &\efrak^3(t)\left(
	\left(\tfrac{\delta}{\beta}\right)^{\frac12}\left(\tfrac{\varepsilon^\delta}{\mu^\beta}\right)^\frac{\xi}{\beta}
	+\left(\tfrac{1}{1-\beta}\right)^\frac12N^{-\frac{\beta}{2}}\right),\\[2pt]
	\big|\gamma_{b,<}^{(4)}(t)\big| & \ls & \efrak^3(t)\,\alwb(t)\,+\,\efrak^3(t)\left(\tfrac{\mu^\beta}{\varepsilon}+\left(\tfrac{\varepsilon^3}{\mu^\beta}\right)^\frac12+N^{-\sigma}+\mu^\eta+\mu^\frac{1-\beta}{2}\right)\,.
\end{eqnarray*}
\end{prop}

\begin{remark}
\remit{
	\item \label{rem:differences:NLS:1d:1}
	The estimates of $\gamma_{a,<}$, $\gamma_{b,<}^{(1)}$ and $\gamma_{b,<}^{(2)}$ work analogously to the corresponding bounds in~\cite{NLS} and are briefly summarised in Sections~\ref{subsec:gamma_a} and~\ref{subsec:gamma_b_2}.
	While $\gamma_{a,<}$ is easily bounded since it contains only one-body contributions, the key for the estimate of $\gamma_{b,<}^{(1)}$ is that for sufficiently large $N$ and small $\varepsilon$,
	\begin{equation*}\begin{split}
	N\int&\d y_2|\chie(y_2)|^2\int \d z_1|\phe(z_1)|^2\wb(z_1-z_2)\\
	&\approx N\left(\int\d y_2|\chie(y_2)|^4\right)\norm{\wb}_{L^1(\R^3)}|\Phi(x_2)|^2
	=\bNe|\Phi(x_2)|^2
	\end{split}
	\end{equation*}
	due to sufficient regularity of $\phe$ and since the support of $\wb$ shrinks as $\mu^\beta$. 
	For this argument, it is crucial that the sequence $(N,\varepsilon)$ is moderately confining.
	
	 The main idea to control $\gamma_{b,<}^{(2)}$ is an integration by parts, exploiting that the antiderivative of $\wb$ is less singular than $\wb$ and that $\nabla_j\psiNe(t)$ can be controlled in terms of the energy $\Eb^{\psiNe(t)}(t)$. 
	 To this end, we define the function $\he$ as the solution of the equation $\Delta\he=\wb$ on a three-dimensional ball with radius $\varepsilon$ and Dirichlet boundary conditions and integrate by parts on that ball. To prevent contributions from the boundary, we insert a smoothed step function whose derivative can be controlled (Definition~\ref{def:h:ball}). To make up for the factors $\varepsilon^{-1}$ from the derivative, one observes that all expressions in $\gamma_{b,<}^{(2)}$ contain at least one projection $\qc$. Since $\norm{\qc_1\psiNe(t)}=\mathcal{O}(\varepsilon)$ (Lemma~\ref{lem:a_priori:4}), which follows  since the spectral gap between ground state and excitation spectrum grows proportionally to $\varepsilon^{-2}$, the projections $\qc$ provide the missing factors $\varepsilon$. The second main ingredient is the admissibility condition, which allows us to cancel small powers of $N$ by powers of $\varepsilon$ gained from $\qc$.
	\item \label{rem:differences:NLS:1d:2}
	For $\gamma_{b,<}^{(3)}$, this strategy  of a three-dimensional integration by parts does not work: whereas $\qc$ cancels the factor $\varepsilon^{-1}$ from the derivative,  we do not gain sufficient powers of $\varepsilon$ to compensate for all positive powers of $N$. Note that this problem did not occur in~\cite{NLS}, where the ratio of $N$ and $\varepsilon$ was different.%
	\footnote{In the 3d $\to$ 1d case~\cite{NLS}, the range of the interaction scales as $\mu_\mathrm{1d}^\beta=(\varepsilon^2/N)^\beta$, besides $\chie_\mathrm{1d}(y)=\varepsilon^{-1}\chi_\mathrm{1d}(y/\varepsilon)$, and the admissibility condition reads $\varepsilon^2/\mu^\beta_\mathrm{1d}\to 0$. These slightly different formulas lead to the estimate $\onorm{(\nabla_1\he^\mathrm{1d}(z_1-z_2))p_1^\mathrm{1d}}\ls N^{-1+\frac{\beta}{2}}\varepsilon^{1-\beta} $, while we obtain in our case $\onorm{(\nabla_1\heot)p_1}\ls N^{-1+\frac{\beta}{2}}\varepsilon^\frac{1-\beta}{2}$ (Lemma~\ref{lem:h:ball}). 
	Following the same path as in $\gamma_{b,<}^{(2)}$, e.g., for~\eqref{gamma_b_3:1} (corresponding to (21) in~\cite{NLS}), we obtain in the 1d problem the estimate $\sim N^{\frac{\beta}{2}}\varepsilon^{1-\beta}=(\varepsilon^2/\mu^\beta_\mathrm{1d})^\frac12$, which can be controlled by the respective admissibility condition. As opposed to this, we compute in our case that $\eqref{gamma_b_3:1}\sim N^{\frac{\beta}{2}}\varepsilon^\frac{1-\beta}{2}=(\varepsilon/\mu^\beta)^\frac12$, which diverges due to moderate confinement.
}

To cope with $\gamma_{b,<}^{(3)}$, note that both~\eqref{gamma_b_3:1} and~\eqref{gamma_b_3:2} contain the expression $\pc_1\wbot\pc_1$, which, analogously to $\wbbar$, defines a function $\wbar(x_1-x_2,y_2)$ where one of the $y$-variables is integrated out (Definition~\ref{def:bar}). 
We integrate by parts only in the $x$-variable, which has the advantages that $\nabla_x$ does not generate factors $\varepsilon^{-1}$ and that the $x$-anti\-derivative of $\wbar(\cdot,y)$ diverges only logarithmically in $\mu^{-1}$ (Lemma~\ref{lem:bar:4}).
Due to admissibility and moderate confinement condition, this can be cancelled by any positive power of $\varepsilon$ or $N^{-1}$.
In distinction to $\gamma_{b,<}^{(2)}$, we do not integrate by parts on a ball with Dirichlet boundary conditions but instead add and subtract suitable counter-terms as in~\cite{pickl2015} and integrate over $\R^2$. 
Note that one would obtain the same result when  integrating by parts on a ball as in $\gamma_{b,<}^{(2)}$, but in this way the estimates are easily transferable to $\gamma_{b,<}^{(4)}$ (see below).

More precisely, we construct $\vbar(\cdot,y)$ such that $\norm{\wbar(\cdot,y)}_{L^1(\R^2)}=\norm{\vbar(\cdot,y)}_{L^1(\R^2)}$ and that $\supp\vbar(\cdot,y)$ scales as $\rho\in(\mu^\beta,1]$ (Definition~\ref{def:bar}).
As a consequence of Newton's theorem, the solution $\hr$ of $\Delta_x\hr=\wbar-\vbar$ is supported within a two-dimensional ball with radius $\rho$.
We then write $\wbar(\cdot,y)=\Delta_x\hr(\cdot,y)+\vbar(\cdot,y)$, integrate the first term by parts in $x$, and choose $\rho$ sufficiently large that the contributions from $\vbar$ can be controlled.
The full argument is given in Sections~\ref{subsec:p.I.} and~\ref{subsec:gamma_b_3}.
\item \label{rem:differences:NLS:1d:3}
Finally, to estimate $\gamma_{b,<}^{(4)}$ (Section~\ref{subsec:gamma_b_4}), we define $\wbbar$ as above and integrate by parts in $x$, using an auxiliary potential $\vbbar$ analogously to $\vbar$ (Definition~\ref{def:bar}). To cope with the logarithmic divergences from the two-dimensional Green's function, we integrate by parts twice, following an idea from~\cite{pickl2015}.
This is the reason why we defined $\hr$ and $\hbbar$ on $\R^2$ and not on a ball, which would require the use of a smoothed step function. While the results are the same when integrating by parts only once, it turns out that the additional factors $\rho^{-1}$ from a second derivative hitting the step function cannot be controlled sufficiently well.

For~\eqref{gamma_b_4:2}, the bound $\norm{\nabla_{x_1}\psiNe(t)}^2\ls 1$  from \textit{a priori} energy estimates is insufficient, comparable to the situation in~\cite{pickl2015} and~\cite{NLS}. Instead, we require an improved bound on the kinetic energy of the part of $\psiNe(t)$ with at least one particle orthogonal to $\Phi(t)$, given by $\norm{\nabla_{x_1}\qp_1\psiNe(t)}^2$.
Essentially, one shows that
\begin{eqnarray*}
\big|\Eb^{\psiNe(t)}-\Ecal^{\Phi(t)}(t)\big|
&\gs&\norm{\nabla_{x_1}\psiNe(t)}^2-\norm{\nabla_x\Phi(t)}^2-\smallO{1}\\
&\gs&\norm{\nabla_{x_1}\qp_1\psiNe(t)}^2+(\norm{\nabla_{x_1}\pp_1\psiNe(t)}^2-\norm{\nabla_x\Phi(t)}^2)-\smallO{1}\\
&\geq&\norm{\nabla_{x_1}\qp_1\psiNe(t)}^2-\norm{\nabla_x\Phi(t)}^2\llr{\psiNe(t),\hat{n}\psiNe(t)}-\smallO{1}\,,
\end{eqnarray*}
which implies 
\begin{equation*}
\norm{\nabla_{x_1}\qp_1\psiNe(t)}^2\ls \alwb(t)+\smallO{1}\,.
\end{equation*}
The rigorous proof of this bound (Lemma~\ref{lem:E_kin<}) is an adaptation of the corresponding Lemma~4.21 in~\cite{NLS} and requires the new strategies described above, as well as both moderate confinement and admissibility condition.
}\label{rem:differences:NLS:1d}
\end{remark}

\subsection{The Gross--Pitaevskii case $\beta=1$}\label{subsec:GP}
For an interaction $\wm$ in the Gross--Pitaevskii scaling regime, the previous strategy, i.e., deriving an estimate of the form $|\tfrac{\d}{\d t}\awm^<(t)|\ls \awm^<(t)+\smallO{1}$, cannot work. To understand this, let us analyse the term $\gamma_{b,<}^{(1)}$, which contains the difference between the quasi two-dimensional interaction $\wbbar$ and the effective potential $b_1|\Phi(t)|^2$. 
As pointed out in Remark~\ref{rem:differences:NLS:1d:1}, the basic idea here is to expand $|\phe(z_1-z_2)|^2$ around $z_2$, which can be made rigorous for sufficiently regular $\phe$ and yields
\begin{equation}\label{eqn:large:difference}
N\int\d y_2|\chie(y_2)|^2\int \d z_1|\phe(z_1)|^2\wm(z_1-z_2)\approx N\left(\int\d y|\chie(y)|^4\right)\norm{\wm}_{L^1(\R^3)}|\Phi(x_2)|^2\,.
\end{equation}
Whereas this equals (at least asymptotically) the coupling parameter $\bb$ for $\beta\in(0,1)$, the situation is now different since  $b_1=8\pi a\int|\chi(y)|^4\d y$. In order to see that~\eqref{eqn:large:difference} and $b_1$ are not asymptotically equal, but actually differ by an error of $\mathcal{O}(1)$, let us briefly recall the definition of the scattering length and its scaling properties.

The three-dimensional zero energy scattering equation for the interaction $  \wm = \mu^{-2} w(\cdot/\mu)$ is 
\begin{equation}\label{eqn:scat}
\begin{cases}
	\left(-\Delta+\tfrac12 \wm(z)\right)j_\mu(z)=0 	& \text{ for } |z|<\infty\,,\\
	\;j_\mu(z)\rightarrow 1														& \text{ as } |z|\rightarrow\infty\,.
\end{cases}
\end{equation}
By \cite[Theorems C.1 and C.2]{LSSY}, the unique solution $j_\mu\in\mathcal{C}^1(\R^3)$ of \eqref{eqn:scat} is spherically symmetric, non-negative and non-decreasing in $|z|$, and satisfies
\begin{equation}\label{eqn:j}
\begin{cases}
	j_\mu(z)=1-\frac{a_\mu}{|z|} & \text{ for }|z|>\mu,\vphantom{\bigg(}\\
	j_\mu(z)\geq 1-\frac{a_\mu}{|z|} & \text{ else},
\end{cases}
\end{equation} 
where $a_\mu\in\R$ is called the scattering length of $\wm$.  Equivalently,
\begin{equation}\label{eqn:integral:scat}
8\pi a_\mu:=\int\limits_{\R^3}\wm(z)j_\mu( z)\d z\,.
\end{equation}
From the scaling behaviour of \eqref{eqn:scat}, it is obvious that $j_\mu(z) = j_{\mu=1}(  z/\mu)$ and that
\begin{equation}\label{eqn:a^N,eps}
a_\mu=\mu a\,,
\end{equation}
where $a$ denotes the scattering length of the unscaled interaction $w=w_{\mu=1}$, i.e.,
\begin{equation}\label{eqn:def:a}
8\pi a :=\int\limits_{\R^3} w(z)j_{\mu=1}(z)\d z\,.
\end{equation}
Returning to the original question, this implies that
$$b_1=8\pi a\int\limits_{\R}|\chi(y)|^4\d y
=N\int\limits_{\R}|\chie(y)|^4\d y\int\limits_{\R^3}\wm(z)j_\mu( z)\d z\,,
$$
and consequently
\begin{eqnarray*}
\eqref{eqn:large:difference}-b_1|\Phi(x_2)|^2
&=&N|\Phi(x_2)|^2\int\limits_\R|\chie(y)|^4\d y\int\limits_{\R^3}w_\mu(z)(1-j_\mu(z))\\
&\geq& \mu^{-1}|\Phi(x_2)|^2\int\limits_{\R}|\chi(y)|^4\d y\left(1-j_\mu(\mu)\right)\norm{\wm}_{L^1(\R^3)}
\;=\;\mathcal{O}(1)\,,
\end{eqnarray*}
where we have used that $\norm{\wm}_{L^1(\R^3)}=\mu\norm{w}_{L^1(\R^3)}$ and that $j_\mu(z)$ is continuous and non-decreasing, hence $j_\mu(z)\leq j_\mu(\mu)$ for $z\in\supp\wm$ and
$1-j_\mu(\mu)\approx a$.
In conclusion, the contribution from $\gamma_{b,<}^{(1)}$ does not vanish if $b_1$ is the coupling parameter in~\cite{NLS}.
Naturally, one could amend this by taking $\int|\chi(y)|^4\d y\norm{w}_{L^1(\R^3)}$ instead of $b_1$ as parameter in the non-linear equation. However, for this choice, the contributions from $\gamma_{b,<}^{(2)}$ to $\gamma_{b,<}^{(4)}$ would not vanish in the limit $(N,\varepsilon)\to(\infty,0)$, as can easily be seen by setting $\beta=1$ in Proposition~\ref{prop:gamma^<}.

The physical reason why the Gross--Pitaevskii scaling is fundamentally different --- and why it requires a different strategy of proof --- is the fact that the length scale $a_\mu$ of the inter-particle correlations is of the same order as  the range $\mu$ of the interaction.
In contrast, for $\beta\in(0,1)$, the relation $a_{\mu,\beta}\ll\mu^\beta$  implies that $j_{\mu,\beta}\approx 1$ on the support of $\wb$, hence the first order Born approximation $8\pi a_{\mu,\beta}\approx\norm{\wb}_{L^1(\R^3)}$ applies in this case. 

Before explaining the strategy of proof for the Gross--Pitaevskii scaling, let us introduce the auxiliary function $\fb\in\mathcal{C}^1(\R^3)$. 
This function will be defined in such a way that it asymptotically coincides with $j_\mu$ on $\supp\wm$ but, in contrast to $j_\mu$, satisfies $\fb(z)=1$ for sufficiently large $|z|$, which has the benefit of $1-\fb$ and $\nabla\fb$ being compactly supported.
To construct $\fb$, we define the potential $\Ubt$ such that the scattering length of $\wm-\Ubt$ equals zero, and define $\fb$ as the solution of the corresponding zero energy scattering equation:

\begin{definition} \label{def:U}
Let $\bt\in(\tfrac13,1)$. Define
$$\Ubt(z):=\begin{cases}
	\mu^{1-3\bt}a 	& \text{ for } \mu^\bt<|z|<\Rbt,\\
	0				& \text{ else,}\end{cases}$$
where $\Rbt$ is the minimal value in $(\mu^\bt,\infty]$ such that the scattering length of $\wm-\Ubt$ equals zero. 
Further, let $\fb\in\mathcal{C}^1(\R^3)$ be the solution of
\begin{equation}\label{eqn:scat:f}
\begin{cases}
\Big(-\Delta+\frac{1}{2}\left(\wm(z)-\Ubt(z)\right)\Big)\fb(z)=0 \;& \text{for }|z|< \Rbt,\\[5pt]
\;\fb(z)=1 & \text{for }|z|\geq \Rbt\,,
\end{cases}
\end{equation}
and define
$$\gb:=1-\fb.$$
\end{definition}
In the sequel, we will abbreviate
$$\Ubt^{(ij)}:=\Ubt(z_i-z_j)\,,\quad  \gbij:=\gb(z_i-z_j)\,, \quad \text{ and } \quad \fbij:=\fb(z_i-z_j).$$
In~\cite[Lemma 4.9]{GP}, it is shown by explicit construction that a suitable $\Rbt$ exists and that it is of order $\mu^\bt$. 
Note that Definition \ref{def:U} implies in particular that
\begin{equation}\label{eqn:scat(w-U)=0}
\int\limits_{\R^3}\left(\wm(z)-\Ubt(z)\right)\fb(z)\d z=0\,,
\end{equation}
which is an equivalent way of expressing that the scattering length of $\wm-\Ubt$ equals zero.
Let us remark that a comparable construction was used in \cite{brennecke2017} and in the series of papers \cite{boccato2017,boccato2018,boccato2018_2}\footnote{
Translated to our setting, the authors consider the ground state $f_\ell$ of the rescaled Neumann problem $\left(-\Delta+\tfrac12 w_\mu(z)\right)f_\ell(z)=\mu^{-2}\lambda_\ell f_\ell(z)$ on the ball $\lbrace|z|\leq \ell\rbrace$ for some $\ell\sim 1$ and extend it by $f_\ell(z)=1$ outside the ball.
The lowest Neumann eigenvalue scales as $\lambda_\ell\sim(\mu/\ell)^3$, hence one can re-write the equation in the form
$\left(-\Delta+\tfrac12\left(w_\mu(z)-U(z)\right)\right)f_\ell(z)=0$, where $U(z)=\mu  C \mathbbm{1}_{|z|\leq \ell}$ for some constant $C$. This is comparable to  \eqref{eqn:scat(w-U)=0} for the choice $\bt=0$. Note that in contrast, we require $\bt>\max\lbrace\frac{\gamma+1}{2\gamma},\frac{5}{6}\rbrace$ (Proposition~\ref{prop:gamma:GP}).
}.\\

Heuristically, one may think of the condensed $N$-body state as a product state that is overlaid with a microscopic structure described by $\fb$, i.e.,
\begin{equation}\label{psi_cor}
\psi_\mathrm{cor}(t,z_1\mydots z_N):=\prod\limits_{k=1}^N\phe(t,z_k)\prod\limits_{1\leq l<m\leq N}\fb(z_l-z_m)\,,
\end{equation}
as was first proposed by Jastrow in \cite{jastrow}.
For $\beta\in(0,1)$, it holds that $\fb\approx 1$, i.e., the condensate is approximately described by the product $\left(\phe\right)^{\otimes N}$ --- which is precisely the state onto which the operator $P_0=p_1\mycdots p_N$  projects. 
For the Gross--Pitaevskii scaling, however, $\fb$ is not approximately constant, and the product state is no appropriate description of the condensed $N$-body wave function.
The idea in~\cite{pickl2015} is to account for this in the counting functional by replacing the projection $P_0$ onto the product state by the projection onto the correlated state $\psi_\mathrm{cor}$.
In this spirit, one substitutes the expression $\llr{\psi,\hat{m}\psi}$ in $\alwb(t)$ by
$$\llr{\psi,\prod\limits_{k<l}\fblk\hat{m}\prod\limits_{r<s}\fbrs\psi}\approx\llr{\psi,\hat{m}\psi}-N(N-1)\Re\llr{\psi,\gbot\hat{m}\psi}\,,
$$
where we expanded $\fb=1-\gb$ and kept only the terms which are at most linear in $\gb$.
This leads to the following definition:
\begin{definition}\label{def:alpha:GP}
$$\awm(t):=\alwm(t)-N(N-1)\Re\llr{\psi^{N,\varepsilon}(t),\gbot\,\hat{r}\,\psi^{N,\varepsilon}(t)}.$$
\end{definition}
Since the convergence of $\alwm(t)$ is equivalent to~\eqref{eqn:thm:1} and~\eqref{eqn:thm:2}, an estimate of $\awm(t)$ is only meaningful if the correction to $\alwm(t)$ in Definition \ref{def:alpha:GP} converges to zero as $(N,\varepsilon)\to(\infty,0)$. This is the reason why we defined it using the operator $\hat{r}$  (Definition~\ref{def:weights}) instead of $\hat{m}$: as $\hat{r}$ contains additional projections $p_1$ and $p_2$, we can use the estimate $\onorm{\gbot p_1}\ls \varepsilon^{-\frac12}\mu^{1+\frac{\bt}{2}}$ instead of $\norm{\gb}_\infty\ls 1$ (Lemma~\ref{lem:g}).
In the following proposition, it is shown that this suffices for the correction term to vanish in the limit.
\begin{prop}\label{prop:correction}
Assume \emph{A1 -- A4}. Then
$$\left|N(N-1)\Re\llr{\psi^{N,\varepsilon}(t),\gbot\,\hat{r}\,\psi^{N,\varepsilon}(t)}\right|\;\ls\;\varepsilon $$
for all $t\in[0,\Tex)$.
\end{prop}
By adding the correction term to $\alwm(t)$, we effectively replace $\wm$ by $\Ubt\fb$ in the time derivative of $\alwm(t)$. To explain what is meant by this statement, let us analyse the contributions to the time derivative of $\awm(t)$, which are collected in the following proposition:
\begin{prop}\label{prop:dt_alpha:GP}
Assume \emph{A1 -- A4} for $\beta=1$.  Then
$$\left|\tfrac{\d}{\d t}\awm(t)\right|\leq \big|\gamma^<(t)\big|+\big|\gamma_a(t)\big|+|\gamma_b(t)|+|\gamma_c(t)|+|\gamma_d(t)|+|\gamma_e(t)|+|\gamma_f(t)|$$
for almost every $t\in[0,\Tex)$, 
where
\begin{eqnarray}
\gamma^<(t)&:=&\left|\llr{\psi^{N,\varepsilon}(t),\dot{\Vp}(t,z_1)\psi^{N,\varepsilon}(t)}-\lr{\Phi(t),\dot{\Vp}(t,(x,0))\Phi(t)}_{L^2(\R^2)}\right| \label{gamma:GP:a<:1}\\
&&-2N\Im\llr{\psi^{N,\varepsilon}(t),q_1\hat{m}^a_{-1}\big(\Vp(t,z_1)-\Vp(t,(x_1,0))\big)p_1\psi^{N,\varepsilon}(t)} \label{gamma:GP:a<:2}\\
&&-N(N-1)\Im\Big\llangle\psi^{N,\varepsilon}(t),\tilde{Z}^{(12)}\hat{m}\psi^{N,\varepsilon}(t)\Big\rrangle\,,\label{eqn:gamma:GP:b:1:3}\\\nonumber\\
\gamma_a(t)&:=&N^2(N-1)\Im\llr{\psi^{N,\varepsilon}(t),\gbot\left[\Vp(t,z_1)-\Vp(t,(x_1,0)),\hat{r}\right]\psi^{N,\varepsilon}(t)},\label{gamma:GP:a}\\\nonumber\\
\gamma_b(t)&:=&-N\Im\llr{\psi^{N,\varepsilon}(t),b_1(|\Phi(t,x_1)|^2+|\Phi(t,x_2)|^2)\gbot\,\hat{r}\,\psi^{N,\varepsilon}(t)}\label{gamma:GP:b:1:1}\\
&&-N\Im\Big\llangle\psi^{N,\varepsilon}(t),(b_\bt-b_1)(|\Phi(t,x_1)|^2+|\Phi(t,x_2)|^2)\,\hat{r}\,\psi^{N,\varepsilon}(t)\Big\rrangle\label{gamma:GP:b:1:2}\\
&&-N(N-1)\Im\llr{\psi^{N,\varepsilon}(t),\gbot\,\hat{r}\,Z^{(12)}\psi^{N,\varepsilon}(t)},\label{gamma:GP:b:2}\\
\nonumber\\
\gamma_c(t)&:=&-4N(N-1)\Im\llr{\psi^{N,\varepsilon}(t),(\nabla_1\gbot)\cdot\nabla_1\hat{r}\,\psi^{N,\varepsilon}(t)}\label{gamma:GP:c},\\\nonumber\\
\gamma_d(t)&:=&-N(N-1)(N-2)\Im\llr{\psi^{N,\varepsilon}(t),\gbot\left[b_1|\Phi(t,x_3)|^2,\hat{r}\,\right]\psi^{N,\varepsilon}(t)}\label{gamma:GP:d:1}\\
&&+2N(N-1)(N-2)\Im\llr{\psi^{N,\varepsilon}(t),\gbot\big[\wm^{(13)},\hat{r}\,\big]\psi^{N,\varepsilon}(t)}\label{gamma:GP:d:2},\\\nonumber\\
\gamma_e(t)&:=&\tfrac12 N(N-1)(N-2)(N-3)\Im\llr{\psi^{N,\varepsilon}(t),\gbot\big[\wm^{(34)},\hat{r}\,\big]\psi^{N,\varepsilon}(t)}\label{gamma:GP:e},\\\nonumber\\
\gamma_f(t)&:=&-2N(N-2)\Im\llr{\psi^{N,\varepsilon}(t),\gbot\left[b_1|\Phi(t,x_1)|^2,\hat{r}\,\right]\psi^{N,\varepsilon}(t)}\label{gamma:GP:f}.
\end{eqnarray}
Here, we have used the abbreviations
\begin{eqnarray*}
Z^{(ij)}&:=&\wm^{(ij)}-\tfrac{b_1}{N-1}\left(|\Phi(t,x_i)|^2+|\Phi(t,x_j)|^2\right),\\[2pt]
\tilde{Z}^{(ij)}&:=&\Ubt^{(ij)}\fbij-\tfrac{b_\bt}{N-1}(|\Phi(t,x_i)|^2+|\Phi(t,x_j)|^2),
\end{eqnarray*}
where $$
b_\bt\,:=\lim\limits_{(N,\varepsilon)\rightarrow(\infty,0)}\mu^{-1}\int\limits_{\R^3}\Ubt(z)\fb(z)\d z\int\limits_{\R^2}|\chi(y)|^4\d y.$$
\end{prop}
The proof of this proposition is given in Section~\ref{subsec:GP:prop:dt_alpha:GP}. Note that the contributions to the derivative $\tfrac{\d}{\d t}\awm(t)$ fall into two categories:
\begin{itemize}
\item The terms~\eqref{gamma:GP:a<:1}--\eqref{gamma:GP:a<:2} in $\gamma^<$ equal $\gamma_{a,<}$ from Proposition~\ref{prop:alpha^<}, and~\eqref{eqn:gamma:GP:b:1:3} is exactly $\gamma_{b,<}$ with interaction potential $\Ubt\fb$.
Hence, estimating $\gamma^<$ is equivalent to estimating the functional $\alpha^<_{\Ubt\fb}(t)$, which arises from $\alwm(t)$ by replacing the interaction $\wm$ by $\Ubt\fb$. 
Since $\Ubt\fb\in\Wbt$ for any $\eta\in(0,1-\bt)$ (Lemma~\ref{lem:Uf:in:W}), this is an interaction in the NLS scaling regime, which was covered in the previous section.
The physical idea here is that  a sufficiently distant test particle with very low energy cannot resolve the difference  between $\wb$ and $\Ubt\fb\approx\Ubt$ since the scattering length of this difference is approximately zero by construction~\eqref{eqn:scat(w-U)=0}.
\item $\gamma_a$ to $\gamma_f$ can be understood as remainders from this substitution. $\gamma_a$ collects the contributions coming from the fact that the $N$-body wave function interacts with a three-dimensional external trap $\Vp$, while only $\Vp$ evaluated on the plane $y=0$ enters in the effective equation~\eqref{NLS}. Since this is an effect of the strong confinement, it has no equivalent in the three-dimensional problem~\cite{pickl2015}, but the same contribution occurs in the situation of a cigar-shaped confinement~\cite{GP}. The terms $\gamma_b$ to $\gamma_f$ are analogous to the corresponding expressions in~\cite{pickl2015} and~\cite{GP}.
\end{itemize}

By assumption \emph{A4}, our analysis covers sequences $(N,\varepsilon)$ that are $(\Theta,\Gamma)_1$-admissible with $1<\Gamma<\Theta\leq 3$. To emphasize the distinction from the case $\beta\in(0,1)$, let us call $\Theta_1=:\vartheta$ and $\Gamma_1=:\gamma$, i.e., we consider
$$(\Theta,\Gamma)_1=(\vartheta,\gamma)\,.$$

\begin{prop}\label{prop:gamma:GP}
Let $\beta=1$  and assume \emph{A1 -- A4} with parameters $(\Theta,\Gamma)_1=(\vartheta,\gamma)$ in \emph{A4}. 
Let $t\in[0,\Tex)$ and let
$$\max\left\{\tfrac{\gamma+1}{2\gamma},\tfrac56\right\}<d<\bt<1,\qquad 
0<\xi<\min\left\{\tfrac{1-\bt}{2}\,,\,\tfrac{3-\vartheta\bt}{2(\vartheta-1)}\right\}\,.$$
Then, for sufficiently small $\mu$,
\begin{eqnarray*}
\big|\gamma^<(t)\big|&\ls& \efrako^3(t)\alwm+\efrako^4(t)\left(\left(\tfrac{\varepsilon^\vartheta}{\mu}\right)^{\frac{\bt}{2}}+\left(\tfrac{\mu}{\varepsilon^\gamma}\right)^\frac{1}{\bt\gamma^2}+\varepsilon^\frac{1-\bt}{2}+N^{-d+\frac56} \right)\,, \\
\big|\gamma_a(t)\big|&\ls&\efrako^4(t)\left(\tfrac{\varepsilon^\vartheta}{\mu}\right)^{1+\xi-\frac{\bt}{2}},\\
\big|\gamma_b(t)\big|&\ls& \efrako^3(t)\,\varepsilon^\frac{1+\bt}{2},\\
\big|\gamma_c(t)\big|&\ls& \efrako^3(t)\left(\varepsilon^\frac{1+\bt}{2}+\left(\tfrac{\mu}{\varepsilon^\gamma}\right)^{\frac{\bt}{2}-\xi}\right), \\
\big|\gamma_d(t)\big|&\ls& \efrako^3(t)\left(\left(\tfrac{\varepsilon^\vartheta}{\mu}\right)^{1+\xi-\bt}+\varepsilon^\frac{1+\bt}{2}\right),\\
\big|\gamma_e(t)\big|&\ls&\efrako^3(t)\,\varepsilon^\frac{1+\bt}{2}\,, \\
\big|\gamma_f(t)\big|&\ls&\efrako^3(t)\,\varepsilon^\frac{1+\bt}{2}.
\end{eqnarray*}
\end{prop}

\begin{remark}
\remit{
\item To estimate $\gamma^<$, observe first that we have chosen $\bt$ such that $\Ubt\fb\in\Wbt$ for some $\eta$, and such that assumption~\emph{A4} with parameters $(\Theta,\Gamma)_1=(\vartheta,\gamma)$ makes the sequence $(N,\varepsilon)$ at the same time admissible/moderately confining with parameters $(\Theta,\Gamma)_\bt=(\delta/\bt,1/\bt)$ for some $\delta\in(1,3)$ (see Section \ref{subsec:GP:gamma<}). 
Consequently, Proposition~\ref{prop:gamma^<} yields 
\begin{equation}\label{eqn:wrong:gamma<}
|\gamma^<(t)|\;\ls\; \alpha^<_{\Ubt\fb}(t)+\smallO{1}\;=\;
\llr{\psiNe,\hat{m}\psiNe}+\big|E_{\Ubt\fb}^{\psiNe(t)}(t)-\mathcal{E}_{b_{\bt}}^{\Phi(t)}(t)\big|+\smallO{1}\,.
\end{equation}
However, this does not yet complete the estimate for $\gamma^<$ since we need to bound all expressions in Proposition~\ref{prop:dt_alpha:GP} in terms of $\alwm=\llr{\psiNe,\hat{m}\psiNe}+\big|E_{\wm}^{\psiNe(t)}(t)-\mathcal{E}_{b_1}^{\Phi(t)}(t)\big|$, up to contributions $\smallO{1}$.
By construction of $\fb$, it follows that $b_\bt=b_1$ (see \eqref{b=b} in Lemma~\ref{lem:Uf:in:W}), hence $\mathcal{E}_{b_{\bt}}^{\Phi(t)}(t)=\mathcal{E}_{b_1}^{\Phi(t)}(t)$. 
On the other hand, heuristic arguments\footnote{
	See~\cite[pp.\ 1019--1020]{GP}. Essentially, when evaluated on the trial function $\psi_\mathrm{cor}$ from~\eqref{psi_cor}, the energy 
	difference is to leading order given by 
	$N\llr{\psi_\mathrm{cor}(t),(\wm^{(12)}-(\Ubt\fb)^{(12)})\psi_\mathrm{cor}(t)}
	\sim N\int\d z_1|\phe(t,z_1)|^2\int\d z |\fb(z)|^2(\wm(z)-\Ubt(z))
	\sim\mu^{-1}\int\d z \gb(z)\wm(z)\fb(z)\geq \mu^{-1}\gb(\mu)\int\d z\wm(z)\fb(z)\sim 8\pi a^2$.
}
indicate that $E^{\psiNe(t)}_{\Ubt\fb}(t)$ and $E^{\psiNe(t)}_{\wm}(t)$ differ by an error of order $\mathcal{O}(1)$, which implies
that the right hand side of~\eqref{eqn:wrong:gamma<} is different from $\alwm(t)$ by $\mathcal{O}(1)$.  

By Remark~\ref{rem:differences:NLS:1d:3}, this energy difference enters only in the estimate of~\eqref{gamma_b_4:2} in $\gamma_{b,<}^{(4)}$ via  $\norm{\nabla_{x_1}\qp_1\psiNe(t)}^2\ls \alpha^<_{\Ubt\fb}(t)+\smallO{1}$. 
For the Gross--Pitaevskii scaling of the interaction, $\norm{\nabla_{x_1}\qp_1\psiNe(t)}^2$ is not asymptotically zero because the microscopic structure described by $\fb$ lives on the same length scale as the interaction and thus contributes a kinetic energy of $\mathcal{O}(1)$.
However, as this kinetic energy is concentrated around the scattering centres, one can show a similar bound for the kinetic energy on a subset $\mathcal{A}_1$ of $\R^{3N}$, where appropriate holes around these centres are cut out (Definition~\ref{def:cutoffs}).
This is done in Section~\ref{subsec:GP:E_kin}, where we show in Lemma~\ref{lem:E_kin:GP} that 
$$\norm{\charAo\nabla_{x_1}\qp_1\psiNe(t)}^2\ls\alwm(t)+\smallO{1}.$$
The proof of this lemma is similar to the corresponding proof in~\cite[Lemma 4.12]{GP}, which, in turn, adjusts ideas from~\cite{pickl2015} to the problem with dimensional reduction. However, since one key tool for the estimate is the  Gagliardo--Nirenberg--Sobolev inequality in the $x$-coordinates, the estimates depend in a non-trivial way on the dimension of $x$. As one consequence, our estimate requires the moderate confinement condition with parameter $\gamma>1$, where no such restriction was needed in~\cite{GP}.

Finally, we adapt the estimate of~\eqref{gamma_b_4:2}. In distinction to the corresponding proof in~\cite[Section 4.5.1]{GP}, we need to integrate by parts in two steps to be able to control the logarithmic divergences that are due to the two-dimensional Green's function. Inspired by an idea in~\cite{pickl2015}, we introduce two auxiliary potentials $\vbbarbz$ and $\nbbaro$ such that $\norm{\Ufbbar}_{L^1(\R^2)}=\norm{\vbbarbz}_{L^1(\R^2)}=\norm{\nbbaro}_{L^1(\R^2)}$, define $\hbbarbz$ and $\hbbaro$ as the solutions of 
 $\Delta_x\hbbarbz=\Ufbbar-\vbbarbz$ and $\Delta_x\hbbaro=\vbbarbz-\nbbaro$, and write $\Ufbbar=\Delta_x\hbbarbz+\Delta_x\hbbaro+\nbbaro$.
The expressions depending on $\nbbaro$ can be controlled immediately, while we integrate the remainders by parts in $x$, making use of different properties of $\hbbarbz$ and $\hbbaro$ (Lemma~\ref{lem:bar:4}).  
Subsequently, we insert identities $\mathbbm{1}=\charAo+\charAbaro$, where $\Abaro$ denotes the complement of $\Ao$.
On the one hand, this yields $\norm{\charAo\nabla_{x_1}\qp_1\psiNe(t)}$, which can be controlled by the new energy lemma (Lemma~\ref{lem:E_kin:GP}). On the other hand, we obtain terms containing $\charAbaro$, which we estimate by exploiting the smallness of $\Abaro$. The full argument is given in Section~\ref{subsec:GP:gamma<}.
\item The remainders $\gamma_a$ to $\gamma_f$ are estimated in Sections~\ref{subsec:GP:remainders}, and work, for the most part, analogously to the corresponding proofs in~\cite[Sections 4.5.2 -- 4.5.7]{GP}.
The only exception is $\gamma_c$, where  the strategy from~\cite{GP} produces too many factors $\varepsilon^{-1}$. Instead, we estimate the $x$- and $y$-contributions to the scalar product $(\nabla\gb)\cdot\nabla\hat{r}=(\nabla_x\gb)\cdot\nabla_x\hat{r}+(\partial_y\gb)\partial_y\hat{r}$ separately. To control the $y$-part, we integrate by parts in $y$ and use the moderate confinement condition with $\gamma>1$. Again, this is different from the situation in~\cite{GP}, where the corresponding term $\gamma_c$ could be estimated without any restriction on the sequence $(N,\varepsilon)$.
}\label{rem:differences:GP:1d}
\end{remark}

\subsection{Proof of Theorem~\ref{thm}}
\label{subsec:proof}
Let $0\leq T<\Tex$.
For $\beta\in(0,1)$, Proposition~\ref{prop:gamma^<} implies that
$$\left|\tfrac{\d}{\d t}\alwb(t)\right|\ls\efrak^3(t)\alwm(t)+\efrak^3(t)R_{\eta,\beta,\delta,\sigma,\xi}(N,\varepsilon)
$$
for almost every $t\in[0,T]$ and  sufficiently small $\mu$, where 
$$R^<_{\eta,\beta,\delta,\sigma,\xi}(N,\varepsilon):=\left(\tfrac{\varepsilon^\delta}{\mu^\beta}\right)^{\frac{\xi}{\beta}}+\left(\tfrac{\varepsilon^3}{\mu^\beta}\right)^\frac12
+\tfrac{\mu^\beta}{\varepsilon}+\mu^\eta+\varepsilon^\frac{1-\beta}{2}+N^{-\sigma}+N^{-\frac{\beta}{2}}
$$
with $0<\sigma<\min\{\tfrac{1-3\xi}{4},\beta-\xi\}$.
Since $t\mapsto\alwb(t)$ is  non-negative and absolutely continuous on $[0,T]$, the differential version of 
Grönwall's inequality (see e.g.~\cite[Appendix B.2.j]{evans}) yields
\begin{equation}\label{eqn:rate:of:conv:1}
 \alwb(t)\ls\e^{\int_0^t\efrak^3(s)\d s}\left(\alwb(0)+R^<_{\eta,\beta,\delta,\sigma,\xi}(N,\varepsilon)\int_0^t\efrak^3(s)\d s\right) 
 \end{equation}
for all $t\in[0,T]$.
Since $\efrak(t)$ is  bounded uniformly in $N$ and $\varepsilon$ by~\eqref{A4:2} and with $R^<_{\eta,\beta,\delta,\sigma,\xi}(N,\varepsilon)\to 0$ as $(N,\varepsilon)\to(\infty,0)$, this implies~\eqref{eqn:thm:1} and~\eqref{eqn:thm:2} by Lemma~\ref{lem:equivalence}.\\

\noindent For $\beta=1$, observe first that Proposition~\ref{prop:correction} implies that the correction term in $\awm(t)$ is bounded by $\varepsilon$ uniformly in $t\in[0,T]$, provided $\mu$ is sufficiently small. Hence, $t\mapsto\awm(t)+\varepsilon$ is non-negative and absolutely continuous and 
$$
\alwm(t)\ls\awm(t)+\varepsilon<\awm(t)+ R_{\gamma,\vartheta,\xi}(N,\varepsilon)
$$
for
$$R_{\gamma,\vartheta,\xi}(N,\varepsilon)
=\left(\tfrac{\varepsilon^\vartheta}{\mu}\right)^\frac{\bt}{2}
+\left(\tfrac{\mu}{\varepsilon^\gamma}\right)^\frac{1}{\gamma^2}
+\left(\tfrac{\mu}{\varepsilon^\gamma}\right)^{\frac{\bt}{2}-\xi}
+\varepsilon^\frac{1-\bt}{2}+N^{-d+\frac56}$$
with $\max\{\frac{\gamma+1}{2\gamma},\tfrac{5}{6}\}<d<\bt<\tfrac{3}{\vartheta}$.
Consequently, Proposition~\ref{prop:gamma:GP} yields
\begin{equation}\label{eqn:rate:of:conv:2}
\left|\tfrac{\d}{\d t}(\awm(t)+\varepsilon)\right|\ls \efrako^4(t)\left(\awm(t)+R_{\gamma,\vartheta,\xi}(N,\varepsilon)\right)
\end{equation}
for almost every $t\in[0,T]$ and  sufficiently small $\mu$, which, as before, implies the statement of the theorem because both $\varepsilon$ and $R_{\gamma,\vartheta,\xi}(N,\varepsilon)$ converge to zero as $(N,\varepsilon)\to\infty$.

\section{Preliminaries}
\label{sec:preliminaries}
We will from now on always assume that assumptions \emph{A1 -- A4} with parameters $(\Theta,\Gamma)_\beta=(\delta/\beta,1/\beta)$ for $\beta\in(0,1)$ and $(\Theta,\Gamma)_1=(\vartheta,\gamma)$ for $\beta=1$  are satisfied.

\begin{definition}\label{def:H_M}
Let $\mathcal{M}\subseteq\{1,\dots,N\}$. Define $\mathcal{H}_\mathcal{M}\subseteq L^2(\R^{3N})$ as the subspace of functions which are symmetric in all variables in $\mathcal{M}$, i.e.~for $\psi\in\mathcal{H}_\mathcal{M}$,
$$\psi(z_1\mydots z_j\mydots z_k\mydots z_N)=\psi(z_1\mydots z_k\mydots z_j\mydots z_N) \qquad\forall\,j,k\in\mathcal{M}.$$
\end{definition}

\begin{lem}\label{lem:fqq}
Let $f:\mathbb{N}_0\rightarrow\mathbb{R}_0^+$, $d\in\mathbb{Z}$, $\rho\in\{a,b\}$ and $\nu\in\{c,d,e,f\}$. Further, let $\mathcal{M}_1,\mathcal{M}_{1,2}\subseteq\{1,2 \mydots N\}$ with $1\in\mathcal{M}_1$ and $1,2\in\mathcal{M}_{1,2}$. Then 
\lemit{
	\item 	\label{lem:l:1} 
			$\onorm{\hat{f}}=\onorm{\hat{f}_d}=\onorm{\hat{f}^\frac{1}{2}}^2=\sup\limits_{0\leq k\leq N}f(k)$,
	\item 	\label{lem:l:2}		
			$\onorm{\hat{m}^\rho}\leq N^{-1+\xi}$, \; $\onorm{\hat{m}^\nu}\ls N^{-2+3\xi}$ \, and \, $\onorm{\hat{r}}\ls N^{-1+\xi},$ 
	\item 	\label{lem:fqq:1}
			$\hat{n}^2=\frac{1}{N}\sum\limits_{j=1}^N q_j,$			
	\item	\label{lem:fqq:2}
			$\norm{\hat{f}q_1\psi}^2\leq\frac{N}{|\mathcal{M}_1|}\norm{\hat{f}\hat{n}\psi}^2$\; for\; $\psi\in\mathcal{H}_{\mathcal{M}_1},$
			\\[0.2cm]
			$\norm{\hat{f}q_1q_2\psi}^2\leq\frac{N^2}{|\mathcal{M}_{1,2}|\left(|\mathcal{M}_{1,2}|-1\right)}\norm{\hat{f}\,\hat{n}^2\psi}^2$
			\; for \;$ \psi\in\mathcal{H}_{\mathcal{M}_{1,2}},$
			\\[0.2cm]
			$\norm{\hat{m}^\rho_d\, q_1\psi^{N,\varepsilon}(t)}\ls N^{-1}$,
	\item	\label{cor:fqq:1}
			$\norm{\nabla_1\hat{f}q_1\psi}\ls\onorm{\hat{f}}\norm{\nabla_1q_1\psi}$\; for\; $\psi\in L^2(\R^{3N})$, \\[0.2cm]
			$\norm{\nabla_{x_1}\hat{f}\qp_1\psi}\ls\onorm{\hat{f}}\norm{\nabla_{x_1}\qp_1\psi}$\; for\; $\psi\in L^2(\R^{3N}),$
	\item	\label{cor:fqq:2}
			$\norm{\nabla_2\hat{f}q_1q_2\psi}\leq\frac{N}{|\mathcal{M}_1|-1}\onorm{\hat{f}\,\hat{n}}\norm{\nabla_2q_2\psi}$\; for\; $\psi\in\mathcal{H}_{\mathcal{M}_{1}}$,
			\\[0.2cm]
			$\norm{\nabla_{x_2}\hat{f}\qp_1\qp_2\psi}\leq\frac{N}{|\mathcal{M}_1|-1}\onorm{\hat{f}\,\hat{n}}\norm{\nabla_{x_2}\qp_2\psi}$
			\; for\; $\psi\in\mathcal{H}_{\mathcal{M}_{1}}$.
}
\end{lem}
\begin{proof}
\cite{NLS}, Lemmas~4.1 and~4.5 and Corollary~4.6 and~\cite{GP}, Lemma~4.1.
\end{proof}

\begin{lem}\label{lem:commutators}
Let $f,g:\mathbb{N}_0\rightarrow\mathbb{R}_0^+$ be any weights and $i,j\in\{1,\dots,N\}$. 
\lemit{
	\item	\label{lem:commutators:1} For $k\in\{0,\dots,N\}$,
			$$\hat{f}\,\hat{g}=\hat{fg}=\hat{g}\hat{f},\qquad\hat{f}p_j=p_j\hat{f},
			\qquad\hat{f}q_j=q_j\hat{f}, \qquad\hat{f}P_k=P_k\hat{f}.$$
	\item	\label{lem:commutators:2}
			Define $Q_0:=p_j$, $Q_1:=q_j$, $\tilde{Q}_0:=p_ip_j$, $\tilde{Q}_1\in\{p_iq_j,q_ip_j\}$ and $\tilde{Q}_2:=q_iq_j$. 
			Let $S_j$ be an operator acting non-trivially only on coordinate $j$ and $T_{ij}$ only on coordinates $i$ and $j$.
			Then for $\mu,\nu\in\{0,1,2\}$	
			$$ Q_\mu\hat{f}S_jQ_\nu=Q_\mu S_j\hat{f}_{\mu-\nu}Q_\nu \quad \text{ and } \quad
			\tilde{Q}_\mu\hat{f}T_{ij}\tilde{Q}_\nu=\tilde{Q}_\mu T_{ij}\hat{f}_{\mu-\nu}\tilde{Q}_\nu.$$
	\item	\label{lem:commutators:5}
			\begin{equation*}
			[T_{ij},\hat{f}]=[T_{ij},p_ip_j(\hat{f}-\hat{f}_2)+(p_iq_j+q_ip_j)(\hat{f}-\hat{f}_1)].
			\end{equation*}
}
\end{lem}

\begin{proof}
\cite{NLS}, Lemma~4.2.
\end{proof}

\begin{lem}\label{lem:derivative_m}
Let $f:\mathbb{N}_0\rightarrow\R^+_0$.
\lemit{
	\item	\label{lem:derivative_m:1}
			The operators $P_k$ and $\hat{f}$ are continuously differentiable as functions of time, 
			i.e.,
			$$P_k,\;\hat{f}\in \mathcal{C}^1\big(\R,\mathcal{L}\left(L^2(\R^{3N})\right)\big)$$ 
			for $0\leq k\leq N$. Moreover,
			$$\tfrac{\d}{\d t}\hat{f}=\i\Big[\hat{f},\sum\limits_{j=1}^N \hb^{(j)}(t)\Big],$$
			where $\hb^{(j)}(t)$ denotes the one-particle operator corresponding to $\hb(t)$ from \eqref{NLS} acting on the $j$\textsuperscript{th} coordinate.
	\item	\label{lem:derivative_m:2}
			$\left[-\partial_{y_j}^2+\tfrac{1}{\varepsilon^2}V^\perp(\tfrac{y_j}{\varepsilon}),\hat{f}\,\right]=0$ for $1\leq j \leq N$\,.				
}
\end{lem}
\begin{proof}
\cite{NLS}, Lemma~4.3.
\end{proof}

\begin{lem}\label{lem:psi-Phi}
Let $\psi\in L_+^2(\R^{3N})$ be normalised and $f\in L^\infty(\R^2)$. Then
\begin{equation*}
\left|\llr{\psi,f(x_1)\psi}-\lr{\Phi(t),f\Phi(t)}_{L^2(\R^2)}\right|\ls\norm{f}_{L^\infty(\R^2)}\llr{\psi,\hat{n}\psi}.
\end{equation*}
\end{lem}
\begin{proof}
\cite{NLS}, Lemma~4.7.
\end{proof}

\begin{lem}\label{lem:Gamma:Lambda}
Let $\Gamma,\Lambda\in L^2(\R^{3N})\in\mathcal{H}_\mathcal{M}$ such that $j\notin\mathcal{M}$ and $k,l\in\mathcal{M}$ with $j\neq k\neq l\neq j$.
Let $O_{j,k}$ be an operator acting non-trivially only on coordinates $j$ and $k$, 
denote by $r_k$ and $s_k$ operators acting only on the $k$\textsuperscript{th} coordinate, and let $F:\R^3\times\R^3\to\R^d$ for $d\in\mathbb{N}$. Then
\lemit{
	\item	\label{lem:Gamma:Lambda:1}
		$		|\llr{\Gamma,O_{j,k}\Lambda}|\leq\norm{\Gamma}
		\Big(|\llr{O_{j,k}\Lambda,O_{j,l}\Lambda}|+|\mathcal{M}|^{-1}\norm{O_{j,k}\Lambda}^2\Big)^\frac12.
		$
\item	\label{lem:Gamma:Lambda:2}
		$
		|\llr{r_k F(z_j,z_k)s_k\Gamma,r_lF(z_j,z_l) s_l \Gamma }|\leq\norm{s_kF(z_j,z_k)r_k\Gamma}^2.
	$
	\item \label{lem:Gamma:Lambda:3}
	$\left|\llr{\Gamma,r_kF(z_j,z_k)s_k\Lambda}\right|\leq\norm{\Gamma}\left(\norm{s_kF(z_j,z_k)r_k\Lambda}^2+|\mathcal{M}|^{-1}\norm{r_kF(z_j,z_k)s_k\Lambda}^2\right)^\frac12.$
}
\end{lem}
\begin{proof}
\cite{NLS}, Lemma~4.8 and~\cite{GP}, Lemma~4.4.
\end{proof}

\begin{lem}\label{lem:Phi}
Let $t\in[0,\Tex)$. Then for sufficiently small $\varepsilon$,
\lemit{
	\item 	\label{lem:Phi:1}
			$\norm{\Phi(t)}_{L^2(\R^2)}=1,$ \\[0.2cm]
			$\norm{\Phi(t)}_{L^\infty(\R^2)}\ls \norm{\Phi(t)}_{H^2(\R^2)}\leq\efrak(t),$ \\[0.2cm]
			$\norm{\nabla_x\Phi(t)}_{L^\infty(\R^2)}\ls\norm{\Phi(t)}_{H^3(\R^2)}\leq\efrak(t),$\\[0.2cm]
			$\norm{\Delta_x\Phi(t)}_{L^\infty(\R^2)}\ls\norm{\Phi(t)}_{H^4(\R^2)}\leq\efrak(t),$
			
	\item 	$\norm{\chie}_{L^2(\R)}=1$, \quad 
			$\norm{\frac{\d}{\d y}\chie}_{L^2(\R)}\ls\varepsilon^{-1}$, \\[0.2cm]
			$\int\limits_\R|\chie(y)|^4\d y=\varepsilon^{-1}\int\limits_\R|\chi(y)|^4\d y$, \\[0.2cm]
			$\norm{\chie}_{L^\infty(\R)}\ls\varepsilon^{-\frac12}$, \quad
			$\norm{\frac{\d}{\d y}\chie}_{L^\infty(\R)}\ls\varepsilon^{-\frac32}$,		
			
	\item	\label{cor:varphi:1} 
			$\norm{\phe(t)}_{L^\infty(\R^3)}\ls\efrak(t)\varepsilon^{-\frac12},$ \\[0.2cm]
			$\norm{\nabla\phe(t)}_{L^\infty(\R^3)}\ls \efrak(t)\varepsilon^{-\frac32},$\\[0.2cm]
			$\norm{\nabla|\phe(t)|^2}_{L^2(\R^3)}\ls \efrak(t)\varepsilon^{-\frac32}$.

}
\end{lem}
\begin{proof}
Part (a) follows from the Sobolev embedding theorem~\cite[Theorem 4.12, Part I\,A]{adams} and by definition of $\efrak$. Part (b) is an immediate consequence of~\eqref{eqn:chie}, and part (c) is implied by (a) and (b).
\end{proof}

\begin{lem}\label{lem:pfp}
Fix $t\in[0,\Tex)$ and let $j,k\in\{1\mydots N\}$. Let $g:\R^3\times\R^3\rightarrow\R$, $h:\R^2\times\R^2\to\R$ 
be measurable functions such that $|g(z_j,z_k)|\leq G(z_k-z_j)$ and $|h(x_j,x_k)|\leq H(x_k-x_j)$ 
almost everywhere for some $G:\R^3\rightarrow\R$, $H:\R^2\to\R$. 
Let $t_j\in\{p_j,\nabla_{x_j} p_j\}$ and $t_j^\Phi\in\{\pp_j,\nabla_{x_j}\pp_j\}$. Then
\lemit{
	\item	\label{lem:pfp:1} 	
			$\onorm{(t_j)^\dagger\,g(z_j,z_k)t_j}\ls \efrak^2(t) \varepsilon^{-1}\norm{G}_{L^1(\R^3)}$\; for $G\in L^1(\R^3)$, 
	\item	\label{lem:pfp:2}
			$\onorm{g(z_j,z_k)t_j}=\onorm{t_j^\dagger\,g(z_j,z_k)}\ls \efrak(t)\varepsilon^{-\frac12}\norm{G}_{L^2(\R^3)}$\; for $G\in L^2\cap L^\infty(\R^3)$, 
	\item	\label{lem:pfp:3}
			$\onorm{g(z_j,z_k)\nabla_j p_j}\ls \efrak(t)\varepsilon^{-\frac32}\norm{G}_{L^2(\R^3)}$\; for $G\in L^2(\R^3)$, 		
	\item 	\label{lem:pfp:4}	
			$\onorm{ h(x_j,x_k) t^\Phi_j}=\onorm{(t^\Phi_j)^\dagger h(x_j,x_k)}\leq \efrak(t)\norm{H}_{L^2(\R^2)}$ for $H\in L^2\cap L^\infty(\R^2)$.
}
\end{lem}

\begin{proof}
Analogously to \cite{NLS}, Lemma~4.10.
\end{proof}

\begin{lem}\label{lem:a_priori:w12}
Let $\varepsilon$ be sufficiently small and fix $t\in[0,\Tex)$. Then for $\beta\in(0,1]$
\lemit{
	\item	\label{lem:a_priori:4}
			\hspace{-0.2cm}\begin{tabular}{p{4cm}p{5cm}p{5cm}}
			$\onorm{\nabla_{x_1}\pp_1}\leq\efrak(t),$ 
			& $\onorm{\Delta_{x_1}\pp_1}\leq\efrak(t)$,&
 			\end{tabular}
 		\item[]
 			\hspace{-0.2cm}\begin{tabular}{p{4cm}p{4cm}p{5cm}}
			$\onorm{\partial_{y_1}\pc_1}\ls\varepsilon^{-1},$ &
			$\onorm{\partial^2_{y_1}\pc_1}\ls\varepsilon^{-2}$,&
			\end{tabular}	
 		\item[]
 			\hspace{-0.2cm}\begin{tabular}{p{4cm}p{4cm}p{5cm}}
 			$\norm{\qc_1\psi^{N,\varepsilon}(t)}\leq \efrak(t)\varepsilon$, &
 			$\norm{\nabla_{x_1}\qp_1\psi}\ls\efrak(t)$,&
			$\norm{\partial_{y_1}\qc_1\psi^{N,\varepsilon}(t)}\ls\efrak(t),$ 
			\end{tabular}	
 		\item[]
 			\hspace{-0.2cm}\begin{tabular}{p{4cm}p{4cm}p{5cm}}
 			$\norm{\nabla_{x_1}\psi^{N,\varepsilon}(t)}\leq\efrak(t)$, &
			$\norm{\partial_{y_1}\psi^{N,\varepsilon}(t)}\ls\varepsilon^{-1}$, &
			$\norm{\nabla_1\psi^{N,\varepsilon}(t)}\ls\varepsilon^{-1},$
			\end{tabular}
	\item $\left\lVert{\sqrt{\wb^{(12)}}\psi^{N,\varepsilon}(t)}\right\rVert\ls \efrak(t)N^{-\frac12},\qquad $\label{lem:w12:1}
	\item $\norm{\wb^{(12)}\psi^{N,\varepsilon}(t)}\ls \efrak(t)N^{-\frac12}\mu^{\frac12-\frac{3\beta}{2}},$\label{lem:w12:2}
	\item $\onorm{p_1\mathbbm{1}_{\supp{\wb}}(z_1-z_2)}=\onorm{\mathbbm{1}_{\supp{\wb}}(z_1-z_2)p_1}
	\ls \efrak(t)\mu^{\frac{3\beta}{2}}\varepsilon^{-\frac12},$\label{lem:w12:3}
	\item $\norm{p_1 \wb^{(12)}\psi^{N,\varepsilon}(t)}\ls \efrak^2(t)N^{-1}\,.$\label{lem:w12:4}
}
\end{lem}
\begin{proof}
Analogously to \cite{NLS}, Lemma~4.11 and \cite{GP}, Lemma~4.7.
For parts (c) and (e), note that for $\beta\in(0,1)$,
\begin{equation}\label{eqn:L1norm:wb}
\norm{\wb}_{L^1(\R^3)}\sim \mu\, \bNe\leq \mu|\bNe-\bb|+\mu\,\bb\ls \mu
\end{equation}
since $\wb\in\Wb$ for some $\eta>0$. For $\beta=1$, $\norm{\wm}_{L^1(\R^3)}=\mu\norm{w}_{L^1(\R^3)}\ls \mu$ by scaling.
\end{proof}

\begin{lem}\label{lem:taylor}
Let $f:\R\times\R^3\rightarrow\R$ such that $f(t,\cdot)\in\mathcal{C}^1(\R^3)$ and $\partial_y f(t,\cdot)\in L^\infty(\R^3)$ for any $t\in[0,\Tex)$. Then 
\lemit{
	\item $\norm{(f(t,z_1)-f(t,(x_1,0))\pc_1\psi^{N,\varepsilon}(t)}\leq\varepsilon\norm{\partial_yf(t)}_{L^\infty(\R^3)},$
	\item $\norm{(f(t,z_1)-f(t,(x_1,0))\psi^{N,\varepsilon}(t)}\leq\varepsilon\left(\efrak(t)\norm{f(t)}_{L^\infty(\R^3)}+\norm{\partial_yf(t)}_{L^\infty(\R^3)}\right).$
}
\end{lem}
\begin{proof}
Analogously to \cite{NLS}, Lemma~4.12.
\end{proof}

\begin{lem}\label{lem:log}
Let $c\in\R$. Then 
\lemit{
	\item 	$ N^{-c}\ln N< N^{-c^-}$, \qquad
			$\varepsilon^c\ln\varepsilon^{-1}< \varepsilon^{c^-},$ \qquad
			$\mu^c\ln\mu^{-1}<\mu^{c^-},$
	\item $\varepsilon^c\ln N<(\Theta-1)\varepsilon^{c^-}\ls
		\begin{cases}
			\tfrac{\delta-\beta}{\beta}\varepsilon^{c^-} & \beta\in(0,1)\,,\\[1mm]
			\varepsilon^{c^-} & \beta=1\,,
		\end{cases}$ 
	\item[] $N^{-c}\ln\varepsilon^{-1} <\tfrac{1}{\Gamma-1}N^{-c^-}=
		\begin{cases}
			\tfrac{\beta}{1-\beta}N^{-c^-} & \beta\in(0,1)\,,\\[1mm]
			\tfrac{1}{\gamma-1}N^{-c^-} & \beta=1\,,
		\end{cases}$
	\item $N^{-c}\ln\mu^{-1}<\tfrac{\Gamma}{\Gamma-1}N^{-c^-}=
		\begin{cases}
			\tfrac{1}{1-\beta}N^{-c^-} & \beta\in(0,1)\,,\\[1mm]
			\tfrac{\gamma}{\gamma-1}N^{-c^-} & \beta=1\,,
		\end{cases}$
	\item[] $\varepsilon^c\ln\mu^{-1}<\Theta\,\varepsilon^{c^-}\ls
	\begin{cases}
			\tfrac{\delta}{\beta}\,\varepsilon^{c^-} & \beta\in(0,1)\,,\\[1mm]
			\varepsilon^{c^-} & \beta=1\,.
		\end{cases}$
}
\end{lem}
\begin{proof} 
Observe that $N<\varepsilon^{-\Theta+1}$ and $\varepsilon^{-1}<N^\frac{1}{\Gamma-1}$ due to admissibility and moderate confinement, hence $\ln N<(\Theta-1)\ln \varepsilon^{-1}$ and $\ln\varepsilon^{-1}<\tfrac{1}{\Gamma-1}\ln N$.
\end{proof}

\section{Proofs for $\beta\in(0,1)$}
\label{sec:NLS}

\subsection{Proof of Proposition \ref{prop:alpha^<}}
The proof works analogously to the proof of Proposition 3.7 in \cite{NLS} and we provide only the main steps for convenience of the reader. From now on, we will drop the time dependence of $\Phi$, $\phe$ and $\psi^{N,\varepsilon}$ in the notation and abbreviate $\psi^{N,\varepsilon}\equiv\psi$. 
The time derivative of $\alwb(t)$ is bounded by
\begin{equation}\label{eqn:dt_alpha:0}
	\left|\tfrac{\d}{\d t}\alwb(t)\right|\leq\left|\tfrac{\d}{\d t}\llr{\psi,\hat{m}\psi}\right| +\left|\tfrac{\d}{\d t}\big|\Eb^{\psi}(t)-		
	\Ecal^{\Phi}(t)\big|\right|.
\end{equation}
For the second term in~\eqref{eqn:dt_alpha:0}, note that 
$$\left|\tfrac{\d}{\d t}\big|\Eb^{\psi}(t)-\Ecal^{\Phi}(t)\big|\right| = \left|\tfrac{\d}{\d t}\big(\Eb^{\psi}(t)-\Ecal^{\Phi}(t)\big)\right|=\left|\llr{\psi,\dot{\Vp}(t,z_1)\psi}-\lr{\Phi,\dot{\Vp}\left(t,(x,0)\right)\Phi}\right|$$
for almost every $t\in[0,\Tex)$ by \cite[Theorem~6.17]{lieb_loss} because $t\mapsto\tfrac{\d}{\d t}\big(\Eb^{\psi}(t)-\Ecal^{\Phi}(t)\big)$ is continuous due to assumption \emph{A3}. 
The first term in~\eqref{eqn:dt_alpha:0} yields
\begin{eqnarray}
\tfrac{\d}{\d t}\llr{\psi,\hat{m}\psi}
&=&-2 N\Im\llr{\psi,q_1\hat{m}^a_{-1}\left(\Vp(t,z_1)-\Vp\left(t,(x_1,0)\right)\right)p_1\psi}\label{eqn:dt_alpha:1}\\
&&-2N(N-1)\Im\llr{\psi,q_1p_2\hat{m}^a_{-1}Z_\beta^{(12)}p_1p_2\psi}\label{eqn:dt_alpha:2}\\
&&-N(N-1)\Im\llr{\psi,q_1q_2\hat{m}^b_{-2}\wbot p_1p_2\psi}\label{eqn:dt_alpha:3}\\
&&-2N(N-1)\Im\llr{\psi,q_1q_2\hat{m}^a_{-1}Z_\beta^{(12)}p_1q_2\psi}\label{eqn:dt_alpha:4},
\end{eqnarray}
which follows from Lemmas~\ref{lem:commutators} and~\ref{lem:derivative_m}.
Expanding $q=\qc+\pc\qp$ in~\eqref{eqn:dt_alpha:2} to~\eqref{eqn:dt_alpha:4} and subsequently estimating $N\hat{m}^a_{-1}\leq \hat{l}$ and $N\hat{m}^b_{-2}\leq \hat{l}$ for $\hat{l}\in\mathcal{L}$ from \eqref{eqn:mathcal:L} concludes the proof. 
\qed

\subsection{Proof of Proposition~\ref{prop:gamma^<}}
In this section, we will again drop the time dependence of $\psiNe(t)$, $\phe(t)$ and $\Phi(t)$ and abbreviate $\psiNe\equiv\psi$. Besides, we will always take $\hat{l}\in\mathcal{L}$ from~\eqref{eqn:mathcal:L}, hence Lemma~\ref{lem:fqq} implies the bounds
$$\onorm{\hat{l}}\ls N^\xi\,, \qquad \norm{\hat{l}_d\,q_1\psi}\ls 1
$$
for $d\in\mathbb{Z}$.

\subsubsection{Estimate of $\gamma_{a,<}(t)$ and $\gamma_{b,<}^{(1)}(t)$}\label{subsec:gamma_a}
The bounds of $\gamma_{a,<}(t)$ and $\gamma_{b,<}^{(1)}(t)$ are established analogously to \cite{NLS}, Sections 4.4.1 and 4.4.2, and we summarise the main steps of the argument for convenience of the reader. 
With Lemmas~\ref{lem:psi-Phi}, \ref{lem:taylor} and \ref{lem:fqq:2}, we obtain
$$|\gamma_{a,<}(t)|\ls \efrak^3(t)\varepsilon+\efrak(t)\llr{\psi,\hat{n}\psi}.$$
By Lemmas~\ref{lem:Phi} and \ref{lem:fqq:2} and since $\wb\in\Wb$, $\gamma_{b,<}^{(1)}(t)$ can be estimated as
\begin{eqnarray*}
	|\eqref{gamma_b_1}|&\leq &\left|\llr{\hat{l}\qp_1\psi,\pc_1p_2(N\wbot -\bNe|\Phi(x_1)|^2)p_1p_2\psi}\right|\\
	&&+\left|\llr{\hat{l}\qp_1\psi,\pc_1p_2\left(\bNe-\tfrac{N}{N-1}\bb\right)|\Phi(x_1)|^2p_1p_2\psi}\right|\\
	&\ls& \left|\llr{\,\hat{l}\qp_1\psi,\pc_1p_2 \mathcal{G}(x_1)\pp_1\psi}\right|
	+\efrak^2(t)\left(N^{-1}+\mu^\eta\right),
\end{eqnarray*}
where
\begin{equation}\label{eqn:Gamma}
\mathcal{G}(x_1):=N\int\limits_\R|\chie(y_1)|^2\d y_1\left(\;\int\limits_{\R^3}|\phe(z_1-z)|^2\wb(z)\d z-|\phe(z_1)|^2\norm{\wb}_{L^1(\R^3)}\right).
\end{equation}
Note that for any $g\in\mathcal{C}^\infty_0(\R^3)$, 
$\int_{\R^3} g(z_1-z)\wb(z)\d z=g(z_1)\norm{\wb}_{L^1(\R^3)}+R(z_1)$
with
$$|R(z_1)|:=\bigg|\int\limits_{\R^3} \d z\int\limits_0^1\nabla g(z_1-sz)\cdot z\wb(z)\d s\,\bigg|
\leq\sup\limits_{\substack{s\in[0,1]\\z\in\R^3}}|\nabla g(z_1-sz)|\int\limits_{\R^3}\d z|z|\wb(z).$$
Since $|z|\ls\mu^\beta$ for $z\in\supp\wb$ and by~\eqref{eqn:L1norm:wb}, this implies
$\norm{R}^2_{L^2(\R^3)}\ls\mu^{2\beta+2}\norm{\nabla g}^2_{L^2(\R^3)},$
which, by density, extends to $g=|\phe|^2\in H^1(\R^3)$. Hence,
$$\norm{\mathcal{G}}_{L^2(\R^2)}\ls N\norm{|\chie|^2}_{L^2(\R)}\mu^{\beta+1}\norm{\nabla|\phe|^2}\ls\tfrac{\mu^\beta}{\varepsilon}\efrak(t)$$
by Hölder's inequality and Lemma~\ref{lem:Phi}. Using Lemmas~\ref{lem:pfp:4} and \ref{lem:fqq:2}, we obtain
$$|\eqref{gamma_b_1}|\ls\efrak^2(t)\left(\tfrac{\mu^\beta}{\varepsilon}+N^{-1}+\mu^\eta\right).$$

\subsubsection{Estimate of $\gamma_{b,<}^{(2)}(t)$}\label{subsec:gamma_b_2}
The key idea for the estimate $\gamma_{b,<}^{(2)}(t)$ is to integrate by parts on a ball with radius $\varepsilon$, using a smooth cut-off function to prevent contributions from the boundary.
\begin{definition}\label{def:h:ball}
Define $\he: \R^3\to\R$, $z\mapsto\he(z), $
by
$$
\he(z):=\begin{cases}
\displaystyle \frac{1}{4\pi}\left(\;\int\limits_{\R^3}\frac{\wb(\zeta)}{|z-\zeta|}\d\zeta-\int\limits_{\R^3}\frac{\varepsilon}{|\zeta|}\frac{\wb(\zeta)}{|\zeta^*-z|}\d\zeta\right) & \mbox{for } |z|<\varepsilon,\\
 0 & \mbox{else,}
 \end{cases}
$$
where
$\zeta^*:=\frac{\varepsilon^2}{|\zeta|^2}\zeta.$
Further, define $\te:\R^3\rightarrow [0,1]$, $z\mapsto\te(z)$, by
\begin{equation*}
\te(z):=\begin{cases}
	1 & \text{for }|z|\leq \varrho_\beta,\\
	\mathfrak{h}_\varepsilon(|z|) & \text{for } \varrho_\beta<|z|<\varepsilon,\\
	0 &  \text{for }|z|\geq \varepsilon,
\end{cases}
\end{equation*}
where $\mathfrak{h}_\varepsilon:(\varrho_\beta,\varepsilon)\rightarrow (0,1)$, $r\mapsto\mathfrak{h}_\varepsilon(r)$, is a smooth, decreasing  function as in~\cite[Definition 4.15]{NLS} with $\lim_{r\to\varrho_\beta}\mathfrak{h}_\varepsilon(r)=1$ and $\lim_{r\to\varepsilon}\mathfrak{h}_\varepsilon(r)=0$.
We will abbreviate $$\heij:=\he(z_i-z_j), \qquad \teij:=\te(z_i-z_j).$$
\end{definition}

\begin{lem}\label{lem:h:ball}
Let $\mu\ll\varepsilon$. Then
\lemit{
	\item 	\label{lem:h:ball:1}
			$\he$ solves the problem $\Delta\he=\wb$ with boundary condition $\he\big|_{|z|=\varepsilon}=0$ in the sense of distributions,
			
	\item	\label{lem:h:ball:2}
			$\norm{\nabla\he}_{L^2(\R^3)}\ls\mu^{1-\frac{\beta}{2}}$, 
			
	\item	\label{lem:h:ball:4}
			$\norm{\te}_{L^\infty(\R^3)}\ls1$, \quad	$\norm{\te}_{L^2(\R^3)}\ls\varepsilon^\frac32$,\quad
		$\norm{\nabla\te}_{L^\infty(\R^3)}\ls\varepsilon^{-1}$, \quad	$\norm{\nabla\te}_{L^2(\R^3)}\ls\varepsilon^\frac12$.
}
\end{lem}
\begin{proof}
The proof of Lemma~\ref{lem:h:ball} works analogously to Lemmas~4.12 and 4.13 in \cite{NLS} and we briefly recall the argument
for part (b) for convenience of the reader. First, we define
$\he^{(1)}(z):=\int_{\R^3}\frac{\wb(\zeta)}{|z-\zeta|}\d\zeta$ and $\he^{(2)}(z):=\int_{\R^3}\frac{\varepsilon}{|\zeta|}\frac{\wb(\zeta)}{|\zeta^*-z|}\d\zeta.$
To estimate $|\nabla\he^{(1)}|$, note that $|\zeta|\leq \varrho_\beta\ls\mu^\beta$ for $\zeta\in\supp\wb$. For $|z|\leq 2\varrho_\beta$, this implies $|z-\zeta|\leq 3\varrho_\beta\ls\mu^\beta$, hence
$|\nabla\he^{(1)}(z)|\ls\mu^{1-2\beta}$.
For $2\varrho_\beta\leq|z|\leq \varepsilon$, we find $|z-\zeta|\geq\tfrac12|z|$, hence 
$|\nabla\he^{(1)}(z)|\ls\mu|z|^{-2} $. 

For $|\he^{(2)}|$, observe that $\zeta\in\supp\wb$ implies $|\zeta^*|\geq\varepsilon^2\varrho_\beta^{-1}$, hence, for $\mu$ small enough that $\varepsilon \varrho_\beta^{-1}>2$, we obtain $|z|\leq\varepsilon<\tfrac12\varepsilon^2\varrho_\beta^{-1}\leq\tfrac12|\zeta^*|$. Consequently, $|\zeta^*-z|\geq|\tfrac12\varepsilon^2|\zeta|^{-1}$, which yields 
$|\nabla\he^{(2)}|\ls\varepsilon^{-3}\norm{\wb}_{L^\infty(\R^3)}\int_{\supp\wb}|\zeta|^3\d|\zeta|\ls\varepsilon^{-3}\mu^{1+\beta}$. Part (b) follows from this by integration over the finite range of $\supp\he$. 
Part (c) is obvious.
\end{proof}

We now use this lemma to estimate $\gamma_{b,<}^{(2)}$.
Let $t_2\in\{p_2,q_2,\qp_2\pc_2\}$.
As $\te(z_1-z_2)=1$ for $z_1-z_2\in\supp\wb$ and besides $\supp\te=\overline{B_\varepsilon(0)}$, Lemma~\ref{lem:h:ball:1} implies
\begin{eqnarray*}
|\eqref{gamma_b_2:1}|&=&N\left|\llr{\hat{l}t_2\qc_1\psi,\teot\Delta_1\heot p_1p_2\psi}\right|\\
&\leq& N\left|\llr{\hat{l}\qc_1\psi,t_2\teot(\nabla_1\heot)\cdot p_2\nabla_1p_1\psi}\right|\\
&&+N\left|\llr{\hat{l}\qc_1\psi,t_2(\nabla_1\teot)\cdot(\nabla_1\heot)p_2p_1\psi}\right|\\
&&+N\left|\llr{\nabla_1\hat{l}\qc_1\psi,t_2\teot(\nabla_1\heot)p_2p_1\psi}\right|\\
&\ls& N\norm{\hat{l}\qc_1\psi}\Big(\onorm{p_2\teot}^2\onorm{(\nabla_1\heot)\cdot\nabla_1p_1}^2+N^{-1}\onorm{(\nabla_1\heot)
		\nabla_1p_1}^2\Big)^\frac{1}{2}\\
	&&+N\norm{\hat{l}\qc_1\psi}\Big(\onorm{p_2(\nabla_1\heot)}^2\onorm{(\nabla_1\teot)p_1}^2\\
	&&\qquad\qquad\qquad+N^{-1}\norm{\nabla\te}^2_{L^\infty(\R^3)}
	\onorm{(\nabla_1\heot)p_2}^2\Big)^\frac12\\
	&&+N\norm{\nabla_1\hat{l}\qc_1\psi}\Big(\onorm{p_2\teot}^2\onorm{(\nabla_1\heot)p_1}^2
	+N^{-1}\onorm{(\nabla_1\heot)p_2}^2\Big)^\frac12\\
	&\ls& \efrak^3(t)\left(N^{\xi+\frac{\beta}{2}}\varepsilon^\frac{3-\beta}{2}+N^\xi\mu^\frac{1-\beta}{2}\right)\,,
\end{eqnarray*}
where the boundary terms upon integration by parts vanish because $\te(|z|)=0$ for $|z|=\varepsilon$, and where we have used  Lemmas~\ref{lem:Gamma:Lambda}, \ref{lem:fqq},~\ref{lem:pfp}, \ref{lem:a_priori:4} and \ref{lem:h:ball}.
Similarly, one computes
\begin{eqnarray*}
|\eqref{gamma_b_2:2}|
	&\ls& \efrak^3(t)N^{\xi+\frac{\beta}{2}}\varepsilon^\frac{3-\beta}{2}\,,\\
|\eqref{gamma_b_2:3}|
	&\ls&\efrak^3(t)N^{\xi+\frac{\beta}{2}}\varepsilon^\frac{3-\beta}{2}\,,\\
|\eqref{gamma_b_2:4}|
	&\ls &\efrak^3(t)\left(N^{\xi+\frac{\beta}{2}}\varepsilon^\frac{3-\beta}{2}+\mu^\frac{1-\beta}{2}\right).
\end{eqnarray*}
The bound for $\gamma_{b,<}^{(2)}$ follows from this because 
$N^\xi\mu^\frac{1-\beta}{2}=N^\frac{-1+\beta+2\xi}{2}\varepsilon^\frac{1-\beta}{2}\leq \varepsilon^\frac{1-\beta}{2}$ for $\xi\leq\frac{1-\beta}{2}$ and since the admissibility condition implies for $\xi\leq\frac{3-\delta}{2}\cdot\frac{\beta}{\delta-\beta}$ that
$$N^{\xi+\frac{\beta}{2}}\varepsilon^\frac{3-\beta}{2}=\left(\tfrac{\varepsilon^\delta}{\mu^\beta}\right)^{\frac{\xi}{\beta}+\frac12}  \varepsilon^{\frac{3-\delta}{2}-\frac{\delta-\beta}{\beta}\xi}
\leq\left(\tfrac{\varepsilon^\delta}{\mu^\beta}\right)^{\frac{\xi}{\beta}+\frac12}\,.$$

\subsubsection{Preliminary estimates for the integration by parts}\label{subsec:p.I.}
To control $\gamma_{b,<}^{(3)}(t)$ and $\gamma_{b,<}^{(4)}(t)$, we define the quasi two-dimensional interaction potentials  $\wbar(x_1-x_2,y_1)$ and $\wbbar(x_1-x_2)$, which result from integrating out one or both transverse variables of the three-dimensional pair interaction $\wb(z_1-z_2)$, and integrate by parts in $x$. In this section, we provide the required lemmas and definitions in a somewhat generalised form, which allows us to directly apply the results in Sections~\ref{subsec:gamma_b_3},~\ref{subsec:gamma_b_4},~\ref{subsec:E_kin} and~\ref{subsec:GP:gamma<}.

\begin{definition}\label{def:Wbar}
Let $\sigma\in(0,1]$ and define $\Vbar_\sigma$ as the set containing all functions $$\ombar_\sigma:\R^2\times\R\to\R, \qquad (x,y)\mapsto\ombar_\sigma(x,y)$$
such that
\begin{equation*}
\begin{cases}
(a) & \supp \ombar_\sigma(\cdot,y)\subseteq \{x\in\R^2: |x|\leq \sigma\}\text{ for all }y\in\R\,,\\[3pt]
(b) & \norm{\ombar_\sigma}_{L^\infty(\R^2\times\R)}\ls N^{-1}\sigma^{-2}\,,\\[2pt]
(c) & \sup\limits_{y\in\R}\norm{\ombar_\sigma(\cdot,y)}_{L^1(\R^2)}\ls N^{-1}\,,\\[2pt]
(d) &\sup\limits_{y\in\R}\norm{\ombar_\sigma(\cdot,y)}_{L^2(\R^2)}\ls N^{-1}\sigma^{-1}\,.
\end{cases}
\end{equation*}
Further, define the set 
$$\Vbbar_\sigma:=\left\{ \ombbar_\sigma:\R^2\to\R^2: \exists\; \ombar_\sigma\in \Vbar_\sigma \,\text{ s.t. }\, \ombbar_\sigma(x)=\int\limits_{\R}\d y\,|\chie(y)|^2\ombar_\sigma(x,y)\right\}\,.
$$
\end{definition}
Note that $\supp\ombbar_\sigma\subseteq \{x\in\R^2: |x|\leq \sigma\}$ and, since $\chie$ is normalised, the estimates for the norms of $\ombbar_\sigma$ coincide with the respective estimates for $\ombar_\sigma$.
Next, we define the quasi two-dimensional interaction potentials $\wbar$ and $\wbbar$ as well as the auxiliary potentials needed for the integration by parts, and show that they are contained in the sets $\Vbar_\sigma$ and $\Vbbar_\sigma$, respectively, for suitable choices of $\sigma$.
\begin{definition}\label{def:bar}
Let $\wb\in\Wb$ for some $\eta>0$ and define 
\begin{eqnarray}
\wbar:\R^2\times\R\to\R,\qquad (x,y)&\mapsto&\wbar(x,y):=\int\limits_{\R}\d \tilde{y}\,|\chie(\tilde{y})|^2\wb(x,y-\tilde{y})\,,\label{def:wbar}\\
\wbbar:\R^2\to\R,\qquad\qquad x&\mapsto&\wbbar(x):=\int\limits_{\R}\d y\,|\chie(y)|^2\,\wbar(x,y)\,.\label{def:wbbar}
\end{eqnarray}
For $\rho\in(\varrho_\beta,1]$, define 
\begin{eqnarray}\label{def:vbar}
\vbar:\R^2\times\R\to\R\,,\quad (x,y)&\mapsto&\vbar(x,y):=\begin{cases}
	\frac{1}{\pi}\rho^{-2}\norm{\wbar(\cdot,y)}_{L^1(\R^2)} & \mbox{ for } |x|<\rho,\\[2mm]
	0 & \mbox{ else,}
\end{cases}\\
\label{def:vbbar}
\vbbar:\R^2\to\R\,,\quad \qquad x&\mapsto&\vbbar(x):=\int\limits_{\R}\d y\,|\chie(y)|^2\vbar(x,y)\,.
\end{eqnarray}
\end{definition}
It can easily be verified that $\wbbar$ and $\vbbar$ can equivalently be written as
\begin{eqnarray*}
\wbbar(x)&=&\int\limits_\R\d y_1|\chie(y_1)|^2\int_\R\d y_2|\chie(y_2)|^2\wb(x,y_1-y_1)\,,\\
\vbbar(x)&=&\begin{cases}
	\frac{1}{\pi}\rho^{-2}\norm{\wbbar}_{L^1(\R^2)} & \mbox{ for } |x|<\rho,\\[2mm]
	0 & \mbox{ else}.
\end{cases}\\
\end{eqnarray*}
Besides, note that
\begin{eqnarray*}
\pc_2\wbot\pc_2&=&\wbar(x_1-x_2,y_1)\pc_2\,,\\
\pc_1\pc_2\wbot\pc_1\pc_2&=&\wbbar(x_1-x_2)\pc_1\pc_2\,.
\end{eqnarray*}

\begin{lem}\label{lem:wb:in:Vb}
For $\wbar$, $\wbbar$, $\vbar$ and $\vbbar $ from Definition~\ref{def:bar}, it holds that
\lemit{
\item $\wbar\in\Vbar_{\varrho_\beta},\qquad \wbbar\in\Vbbar_{\varrho_\beta},\qquad \vbar\in\Vbar_\rho, \qquad \vbbar\in\Vbbar_\rho\,,$
\item $\norm{\wbar(\cdot,y)}_{L^1(\R^2)}=\norm{\vbar(\cdot,y)}_{L^1(\R^2)}$ for any $y\in\R$,
\item[]  $\norm{\wbbar}_{L^1(\R^2)}=\norm{\vbbar}_{L^1(\R^2)}$.
}
\end{lem}
\begin{proof}
It suffices to derive the respective estimates for $\wbar(\cdot,y)$ and $\vbar(\cdot,y)$ uniformly in $y\in\R$. For instance, Lemma \ref{lem:Phi} and~\eqref{eqn:L1norm:wb} yield
\begin{eqnarray*}
|\wbar(x,y)|&\leq&\norm{\chie}^2_{L^\infty(\R)}\int\limits_{y-\varrho_\beta}^{y+\varrho_\beta}\d y_1\,\mathbbm{1}_{|y-y_1|\leq\varrho_\beta}\wb(x,y-y_1)
\ls\varepsilon^{-1}\mu^{1-2\beta}\sim N^{-1}\varrho_\beta^{-2}\,,\\
\norm{\vbar(\cdot,y)}_{L^1(\R^2)}&=&\frac{1}{\rho^2\pi}\norm{\wbar(\cdot,y)}_{L^1(\R^2)}\int\limits_{\R^2} \mathbbm{1}_{|x|\leq\rho}\d x=\norm{\wbar(\cdot,y)}_{L^1(\R^2)}\ls N^{-1}\,,
\end{eqnarray*}
and the remaining parts are verified analogously.
\end{proof}
In analogy to electrostatics, let us now define the ``potentials'' $\hbarss$ and $\hbbarss$ corresponding to the ``charge distributions'' $\ombar_{\sigma_1}-\ombar_{\sigma_2}$ and $\ombbar_{\sigma_1}-\ombbar_{\sigma_2}$, respectively.
\begin{lem}\label{lem:bar}
Let $0<\sigma_1<\sigma_2\leq 1$, $\ombar_{\sigma_1}\in\Vbar_{\sigma_1}$ and $\ombar_{\sigma_2}\in\Vbar_{\sigma_2}$
such that for any $y\in\R$
$$\norm{\ombar_{\sigma_1}(\cdot,y)}_{L^1(\R^2)}=\norm{\ombar_{\sigma_2}(\cdot,y)}_{L^1(\R^2)}\,.$$ 
Define
\begin{eqnarray}\label{def:hr}
\hbarss:\R^2\times\R&\to&\R\nonumber\\
 (x,y)&\mapsto&\hbarss(x,y):=\frac{1}{2\pi}\int\limits_{\R^2}\d\xi\ln|x-\xi|\Big(\ombar_{\sigma_1}(\xi,y)-\ombar_{\sigma_2}(\xi,y)\Big)\\
&&\hspace{-4.8cm}\text{and}\nonumber\\
\label{def:hbbar}
\hbbarss:\R^2&\to&\R\nonumber\\
x&\mapsto&\hbbarss(x):=\int\limits_{\R}\d y\,|\chie(y)|^2\hbarss(x,y)\,.
\end{eqnarray}
Let $y\in\R$ and
$\big(h_{\sigma_1,\sigma_2},\omega_{\sigma_1},\omega_{\sigma_2}\big)\in
\left\{\Big(\,\hbarss(\cdot,y),\ombar_{\sigma_1}(\cdot,y),\ombar_{\sigma_2}(\cdot,y)\Big)\,,\, \left(\,\hbbarss,\ombbar_{\sigma_1},\ombbar_{\sigma_2}\right) \right\}$. 
\lemit{
	\item	\label{lem:bar:3}
			$h_{\sigma_1,\sigma_2}$ satisfies
			$$ 
				\Delta_x h_{\sigma_1,\sigma_2}=\omega_{\sigma_1}-\omega_{\sigma_2}
			$$
			in the sense of distributions, and
			$$\supp h_{\sigma_1,\sigma_2}\subseteq \left\{x\in\R^2:|x|\leq \sigma_2\right\}\,,$$ 
			
	\item	\label{lem:bar:4}
			$\norm{h_{\sigma_1,\sigma_2}}_{L^2(\R^2)}\ls N^{-1}\sigma_2\left(1+\ln\sigma_2^{-1}\right)  $, 
			\item[]$\norm{\nabla_x h_{\sigma_1,\sigma_2}}_{L^2(\R^2)}\ls N^{-1}\left(\ln \sigma_1^{-1}\right)^\frac12\,
			$.
}
\end{lem}
\begin{proof}
The first part of (a) follows immediately from~\cite[Theorem~6.21]{lieb_loss}.
For the second part, Newton's theorem~\cite[Theorem~9.7]{lieb_loss} states that for $|x|\geq\sigma_2$,
$$\hbarss(x,y)=\frac{1}{2\pi}\ln|x|\int\limits_{\R^2}(\ombar_{\sigma_1}(\xi,y)-\ombar_{\sigma_2}(\xi,y))\d\xi=0$$
as $\norm{\ombar_{\sigma_1}(\cdot,y)}_{L^1(\R^2)}=\norm{\ombar_{\sigma_2}(\cdot,y)}_{L^1(\R^2)}$.
Besides,~\cite[Theorem~9.7]{lieb_loss} yields the estimate
$$\left|\hbarss(x,y)\right|\leq\frac{1}{2\pi}\big|\ln|x|\big|\int\limits_{\R^2}(\ombar_{\sigma_1}(\xi,y)+\ombar_{\sigma_2}(\xi,y))\d\xi
\ls N^{-1}\big|\ln|x|\big|$$
by definition of $\ombar$. Hence, 
$$\norm{\hbarss(\cdot,y)}^2_{L^2(\R^2)}\;\ls \;N^{-2}\int\limits_0^{\sigma_2} r(\ln r)^2\d r
\;\ls\; N^{-2}\sigma_2^2(1+\ln\sigma_2^{-1})^2\,.
$$
To derive the second part of (b), let us define the abbreviations 
\begin{equation*}
\hso(x,y):=\int_{\R^2}  \d\xi\ln|x-\xi|\,\wso(\xi,y), \qquad  
\hst(x,y):=\int_{\R^2}  \d\xi\ln|x-\xi|\,\wst(\xi,y)\,.
\end{equation*}
To estimate $\nabla_x\hso$, let $y\in\R$ and consider $\xi\in\supp\wso(\cdot,y)$, hence $|\xi|\leq\sigma_1$.
If $|x|\leq 2\sigma_1$, we have $|x-\xi|\leq |x|+|\xi|\leq3\sigma_1$, hence
$$|\nabla_x\hso(x,y)|\ls\norm{\wso}_{L^\infty(\R^2\times\R)}\int\limits_0^{3\sigma_1}\d r\ls N^{-1}\sigma_1^{-1}\,. $$
If $2\sigma_1<|x|\leq\sigma_2$, this implies $|x-\xi|\geq |x|-|\xi|\geq|x|-\sigma_1\geq\frac12|x|$, and one concludes
$$|\nabla_x\hso(x,y)|\leq \tfrac{2}{|x|}\int\limits_{\R^2}\wso(\xi,y)\d\xi\ls N^{-1}\tfrac{1}{|x|}\,.
$$
To estimate $\nabla_x\hst$, note that $|x-\xi|\leq |x|+|\xi|\leq 2\sigma_2$ for $x\in\supp\hbarss(\cdot,y)$ and $\xi\in\supp\wst$, hence
$$|\nabla_x\hst(x,y)|\leq\sup\norm{\wst}_{L^\infty(\R^2\times\R)}\int\limits_{|\xi'\leq2\sigma_2}\d|\xi'|\;\ls\; N^{-1}\sigma_2^{-1}\,.$$
Part (b) follows from integrating over $|x|\leq\sigma_2$.
\end{proof}

\subsubsection{Estimate of $\gamma_{b,<}^{(3)}(t)$}\label{subsec:gamma_b_3}
To derive a bound for $\gamma_{b,<}^{(3)}$, observe first that both terms~\eqref{gamma_b_3:1} and~\eqref{gamma_b_3:2} contain the interaction $\wbar$. We add and subtract $\vbar$ from Definition~\ref{def:bar} for suitable choices of $\rho$, i.e.,
\begin{eqnarray*}
\wbar(x_1-x_2,y_1)&=&\wbar(x_1-x_2,y_1)-\vbar(x_1-x_2,y_1)+\vbar(x_1-x_2,y_1)\\
&=&\Delta_{x_1}\hr(x_1-x_2,y_1)+\vbar(x_1-x_2,y_1) 
\end{eqnarray*}
by Lemma~\ref{lem:bar}, which is applicable by Lemma~\ref{lem:wb:in:Vb}.\\

\noindent\emph{Estimate of~\eqref{gamma_b_3:1}.}
Due to the symmetry of $\psi$, \eqref{gamma_b_3:1} can be written as
$$\eqref{gamma_b_3:1}=N\left|\llr{\qc_1\psi,\qp_2\hat{l}\pc_2\wbot\pc_2\pc_1\pp_2\qp_1\psi}+
\llr{\qc_1\psi,\qp_2\hat{l}\pc_2\wbot\pc_2\pc_1\pp_1\qp_2\psi}\right|,
$$
hence with $(s^\Phi_1,t^\Phi_2)\in\{(\pp_1,\qp_2),(\qp_1,\pp_2)\}$ and for some $\rho\in(\varrho_\beta,1]$,
\begin{eqnarray}
|\eqref{gamma_b_3:1}|
&\leq& N\left|\llr{\qc_1\psi,\qp_2\pc_2\left(\Delta_{x_2}\hr(x_1-x_2,y_1)\right)\pc_1\,\hat{l}_1s_1^\Phi t_2^\Phi\psi}\right|\label{eqn:gamma:b:1:1}\\
&&+N\left|\llr{\qc_1\psi,\qp_2\pc_2\vbar(x_1-x_2,y_1)\pc_1\,\hat{l}_1s_1^\Phi t_2^\Phi\psi}\right|.\label{eqn:gamma:b:1:2}
\end{eqnarray}
Since $s^\Phi_1 t^\Phi_2$ contains in both cases a projector $\pp$ and a projector $\qp$, the second term is easily estimated as
$$\eqref{eqn:gamma:b:1:2} \leq N\norm{\qc_1\psi}\norm{\hat{l}_1\qp_1\psi}\onorm{\pp_1\vbar(x_1-x_2,y_1)}
\ls \efrak^2(t)\varepsilon\rho^{-1}$$
by Lemmas~\ref{lem:pfp:4} and \ref{lem:fqq:2}.
For~\eqref{eqn:gamma:b:1:1}, note first that for $(s^\Phi_1,t^\Phi_2)=(\qp_1,\pp_2)$,
\begin{eqnarray*}
\norm{(\nabla_{x_2}\hr(x_1-x_2,y_1))\nabla_{x_2}\pp_2\qp_1\pc_1\hat{l}_1 \psi}
&\leq&\onorm{(\nabla_{x_2}\hr(x_1-x_2,y_1))\nabla_{x_2}\pp_2}^2\norm{\hat{l}_1 \qp_1\psi}\\
&\ls&\efrak(t)N^{-1}(\ln\mu^{-1})^\frac12
\end{eqnarray*}
and for $(s^\Phi_1,t^\Phi_2)=(\pp_1,\qp_2)$,
\begin{eqnarray*}
\norm{(\nabla_{x_2}\hr(x_1-x_2,y_1))\pp_1\nabla_{x_2}\qp_2\pc_1\hat{l}_1 \psi}
&\leq &\onorm{(\nabla_{x_2}\hr(x_1-x_2,y_1))\pp_1}^2\norm{\nabla_{x_2}\qp_2\hat{l}_1 \psi}\\
&\ls&\efrak^2(t)N^{-1+\xi}(\ln\mu^{-1})^\frac12\,,
\end{eqnarray*}
where we have used that $\varrho_\beta\sim\mu^\beta$.
Hence, integration by parts in $x_2$  yields with Lemma~\ref{lem:Gamma:Lambda}
\begin{eqnarray*}
|\eqref{eqn:gamma:b:1:1}|
&\leq&N\left|\llr{\qc_1\psi,\qp_2\pc_2(\nabla_{x_2}\hr(x_1-x_2,y_1))\pc_1\nabla_{x_2}t^\Phi_2\hat{l}_1s^\Phi_1 \psi}\right|\\
&&+N\left|\llr{\nabla_{x_2}\qp_2\pc_2\psi,\qc_1(\nabla_{x_2}\hr(x_1-x_2,y_1))t^\Phi_2\pc_1\hat{l}_1s^\Phi_1 \psi}\right|\\
&\ls& N\norm{\qc_1\psi}\norm{(\nabla_{x_2}\hr(x_1-x_2,y_1))\nabla_{x_2}t^\Phi_2s^\Phi_1\pc_1\hat{l}_1 \psi}\\
&&+N\norm{\nabla_{x_2}\qp_2\pc_2\psi}\Big(\norm{\pc_1s^\Phi_1(\nabla_{x_2}\hr(x_1-x_2,y_1))t^\Phi_2\hat{l}_1\qc_1\psi}^2\\
&&+N^{-1}\norm{(\nabla_{x_2}\hr(x_1-x_2,y_1))t^\Phi_2s^\Phi_1\pc_1\hat{l}_1 \psi}^2\Big)^\frac12\\
&\ls&\efrak^3(t)(N^\xi\varepsilon+N^{-\frac12})(\ln\mu^{-1})^\frac12\,.
\end{eqnarray*}

\noindent\emph{Estimate of~\eqref{gamma_b_3:2}.}
For this term, we choose $\rho=1$
and integrate by parts in $x_2$. This yields
\begin{eqnarray*}
|\eqref{gamma_b_3:2}|
&\leq&N\left|\llr{\hat{l}\qp_1\qp_2\psi,\pc_1\pc_2\vbaro(x_1-x_2,y_1)\pp_2\qc_1\psi}\right|\\
&&+N\left|\llr{\hat{l}\qp_1\qp_2\psi,\pc_1\pc_2\left(\nabla_{x_2}\hro(x_1-x_2,y_1)\right)\cdot\nabla_{x_2}\pp_2\qc_1\psi}\right|\\
&&+N\left|\llr{\nabla_{x_2}\hat{l}\qp_1\qp_2\psi,\pc_1\pc_2\left(\nabla_{x_2}\hro(x_1-x_2,y_1)\right)\pp_2\qc_1\psi}\right|\\
&\leq& N\norm{\qc_1\psi}\norm{\hat{l}\qp_1\qp_2\psi}\left(\onorm{\pp_2\vbaro(x_1-x_2,y_1)}
+\onorm{\left(\nabla_{x_2}\hro(x_1-x_2,y_1)\right)\nabla_{x_2}\pp_2}\right)\\
&&+N\norm{\qc_1\psi}\onorm{\pp_2\left(\nabla_{x_2}\hro(x_1-x_2,y_1)\right)}\norm{\nabla_{x_2}\hat{l}\qp_1\qp_2\psi}\\
&\ls&\efrak^3(t)\varepsilon(\ln\mu^{-1})^\frac12
\end{eqnarray*}
by Lemmas~\ref{lem:fqq}, \ref{lem:a_priori:4}, \ref{lem:pfp:4} and \ref{lem:bar}.
Together, the estimates for~\eqref{gamma_b_3:1} and~\eqref{gamma_b_3:2} yield
$$|\gamma^{(3)}_{b,<}(t)|\ls\efrak^3(t)\left(N^\xi\varepsilon+N^{-\frac12}\right)(\ln\mu^{-1})^\frac12
\ls\efrak^3(t)\left(\tfrac{1}{1-\beta}N^{-1^-}+\tfrac{\delta}{\beta}N^{2\xi}\varepsilon^{2^-}\right)^\frac12
$$
by Lemma~\ref{lem:log}.
Since $\beta\in(0,1)$ and $3-\delta\in(0,2)$ as $\delta\in(1,3)$, this implies
$$|\gamma^{(3)}_{b,<}(t)|\ls\efrak^3(t)\left(\tfrac{1}{1-\beta}N^{-\beta}+\tfrac{\delta}{\beta}N^{2\xi}\varepsilon^{3-\delta}\right)^\frac12\,,
$$
which yields the final bound for $\gamma_{b,<}^{(3)}$ because, by admissibility and since $\xi\leq\frac{3-\delta}{2}\frac{\beta}{\delta-\beta}$,
$$N^\xi\varepsilon^\frac{3-\delta}{2}=\left(\tfrac{\varepsilon^\delta}{\mu^\beta}\right)^\frac{\xi}{\beta}\varepsilon^{\frac{3-\delta}{2}-\frac{\delta-\beta}{\beta}\xi}
\leq\left(\tfrac{\varepsilon^\delta}{\mu^\beta}\right)^\frac{\xi}{\beta}\,.$$

\subsubsection{Estimate of $\gamma_{b,<}^{(4)}(t)$}\label{subsec:gamma_b_4}
First, observe that
\begin{eqnarray*}
|\eqref{gamma_b_4:3}|&\ls&\norm{\hat{l}q_1q_2\psi}\norm{q_2\psi}\norm{\Phi}^2_{L^\infty(\R^2)}
\;\ls\;\efrak^2(t)\llr{\psi,\hat{n}\psi}.
\end{eqnarray*}
Since both terms~\eqref{gamma_b_4:1} and~\eqref{gamma_b_4:2} contain the quasi two-dimensional interaction $\wbbar$, we integrate by parts in $x$ as before, using that
$$\wbbar(x_1-x_1)=\Delta_{x_1}\hbbar(x_1-x_2)+\vbbar(x_1-x_2)$$
and choose $\rho=N^{-\bo}$ for  $\bo=\min\left\{\frac{1+\xi}{4},\beta\right\}$ in~\eqref{gamma_b_4:1} and $\rho=1$ in~\eqref{gamma_b_4:2}.
In the sequel, we abbreviate
$$\wbbarot:=\wbbar(x_1-x_2)\,,\qquad\vbbar^{(12)}:=\vbbar(x_1-x_2)\,,\qquad \hbbarot:=\hbbar(x_1-x_2).$$\\

\noindent\emph{Estimate of~\eqref{gamma_b_4:1}.} 
Integration by parts in $x_1$ yields with Lemma~\ref{lem:commutators:2}
\begin{eqnarray}
|\eqref{gamma_b_4:1}|
&\leq&N\left|\llr{\,\hat{l}^\frac12\qp_1\qp_2\psi,\vbbar^{(12)} p_1p_2\,\hat{l}_2^\frac12\psi}\right|\label{eqn:gamma:b:4:1}\\
&&+N\left|\llr{\nabla_{x_1}\hat{l}\qp_1\qp_2\psi,(\nabla_{x_1}\hbbarot)p_1p_2\psi}\right|\label{eqn:gamma:b:4:2}\\
&&+N\left|\llr{\hat{l}\qp_1\psi,\qp_2(\nabla_{x_1}\hbbarot)\cdot\nabla_{x_1}p_1p_2\psi}\right|\label{eqn:gamma:b:4:3}.
\end{eqnarray}
For the first term, we obtain with Lemmas~\ref{lem:Gamma:Lambda:3}, \ref{lem:pfp:4} and for $\rho=N^{-\beta_1}$
\begin{eqnarray*}
|\eqref{eqn:gamma:b:4:1}|
&\ls&N\norm{\hat{l}^\frac12\qp_1\psi}\left(\norm{p_2\vbbar^{(12)} p_1\hat{l}_2^\frac12\qp_1\psi}^2+N^{-1}\onorm{\vbbar^{(12)}\pp_2}^2\norm{\hat{l}^\frac12_2\psi}^2\right)^\frac12\\
&\ls&\efrak^2(t)\left(\llr{\psi,\hat{n}\psi}+ N^{-\frac12+\frac{\xi}{2}+\bo}\right)\,,
\end{eqnarray*}
where we used that $\vbbar=\sqrt{\vbbar}\sqrt{\vbbar}$ since $\vbbar\geq0$ and consequently
\begin{eqnarray}
\onorm{p_2\vbbar^{(12)} p_1}^2
&\leq \onorm{\pp_2\sqrt{\vbbar^{(12)}}}^2\onorm{\sqrt{\vbbar^{(12)}}\pp_1}^2\ls \efrak^4(t)\norm{\vbbar}^2_{L^1(\R^2)}\ls\;\efrak^4(t)N^{-2}\,.\quad\label{eqn:p_1Up_2}
\end{eqnarray}
To estimate~\eqref{eqn:gamma:b:4:2} and~\eqref{eqn:gamma:b:4:3}, observe first that for any operator $s_1$ acting only on the first coordinate,
\begin{eqnarray}
&&\hspace{-1cm}\llr{\qp_2(\nabla_{x_2}\hbbarot) s_1p_2\tilde{\psi},\qp_3(\nabla_{x_3}\hbbaroth )s_1p_3\tilde{\psi}}\nonumber\\
&=&-\llr{\hbbarot s_1\nabla_{x_2}p_2\qp_3\tilde{\psi},(\nabla_{x_3}\hbbaroth)s_1p_3\qp_2\tilde{\psi}}
-\llr{\hbbarot s_1p_2\qp_3\tilde{\psi},(\nabla_{x_3}\hbbaroth)s_1p_3\nabla_{x_2}\qp_2\tilde{\psi}}\nonumber\\
&=&\llr{\hbbarot s_1\nabla_{x_2}p_2\nabla_{x_3}\qp_3\tilde{\psi},\hbbaroth s_1p_3\qp_2\tilde{\psi}}
+\llr{\hbbarot s_1\nabla_{x_2}p_2\qp_3\tilde{\psi},\hbbaroth s_1\nabla_{x_3}p_3\qp_2\tilde{\psi}}\nonumber\\
&&+\llr{\hbbarot s_1p_2\nabla_{x_3}\qp_3\tilde{\psi},\hbbaroth s_1p_3\nabla_{x_2}\qp_2\tilde{\psi}}
+\llr{\hbbarot s_1p_2\qp_3\tilde{\psi},\hbbaroth s_1\nabla_{x_3}p_3\nabla_{x_2}\qp_2\tilde{\psi}}\nonumber\\
&\ls& \efrak^2(t)\norm{\hbbar}_{L^2(\R^2)}^2\left(\norm{s_1\qp_2\tilde{\psi}}^2+\norm{s_1\nabla_{x_2}\qp_2\tilde{\psi}}^2\right)\label{eqn:q23trick}
\end{eqnarray}
by Lemmas~\ref{cor:fqq:1} and \ref{lem:a_priori:4}. 
With Lemmas~\ref{lem:Gamma:Lambda}, \ref{lem:fqq} and~\ref{lem:bar:4}, we thus obtain for $\rho=N^{-\bo}$
\begin{eqnarray*}
|\eqref{eqn:gamma:b:4:2}|
&\ls&N\norm{\nabla_{x_1}\hat{l}\qp_1\psi}\bigg(\llr{\qp_2(\nabla_{x_2}\hbbarot) p_1p_2\psi,\qp_3(\nabla_{x_3}\hbbaroth) p_1p_3\psi}\\
&&\qquad\qquad\qquad+N^{-1}\onorm{(\nabla_{x_1}\hbbarot)\pp_1}^2\bigg)^\frac12\\
&\ls&\efrak^3(t)\left(N^{-\bo+\xi}\ln N+ N^{-\frac12+\xi}(\ln\mu^{-1})^\frac12\right)\,,\\
|\eqref{eqn:gamma:b:4:3}|
&\ls& N\norm{\hat{l}\qp_1\psi}\bigg(\llr{\qp_2(\nabla_{x_2}\hbbarot) \nabla_{x_1}p_1p_2\psi,\qp_3(\nabla_{x_3}\hbbaroth) \nabla_{x_1}p_1p_3\psi}\\
&&+N^{-1}\onorm{(\nabla_{x_1}\hbbarot)\pp_2}^2\onorm{\nabla_{x_1}\pp_1}^2\bigg)^\frac12\\
&\ls&\efrak^2(t) \left(N^{-\bo}\ln N+N^{-\frac12}(\ln\mu^{-1})^\frac12\right).
\end{eqnarray*}
Together, this yields with Lemma~\ref{lem:log}
\begin{eqnarray*}
|\eqref{gamma_b_4:1}|
\ls \efrak^2(t)\llr{\psi,\hat{n}\psi}
+\efrak^3(t)\left(N^{-\frac12+\frac{\xi}{2}+\bo}+(\tfrac{1}{1-\beta})^\frac12N^{-(\frac12-\xi)^-}+N^{-(\beta_1-\xi)^-}\right)
\end{eqnarray*}
Note that for $\bo=\min\{\frac{1+\xi}{4},\beta\}$ and since $\xi<\frac13$, it holds that $N^{-\bo+\xi}>N^{-\frac12+\xi}$ and that $-\frac12+\frac{\xi}{2}+\bo<-\bo+\xi$. Hence, 
\begin{eqnarray*}
|\eqref{gamma_b_4:1}|
\ls \efrak^2(t)\llr{\psi,\hat{n}\psi}+\efrak^3(t)N^{-(\bo-\xi)^-}\,.
\end{eqnarray*}

\noindent\emph{Estimate of~\eqref{gamma_b_4:2}.} 
Observe first that for $j\in\{0,1\}$,
\begin{eqnarray}
&&\hspace{-1.5cm}\norm{\pp_1(\nabla_{x_2}\hbbarot)\hat{l}_j\qp_1\qp_2\psi}^2\\
&=&\norm{|\Phi(x_1)\rangle\langle\nabla_{x_1}\Phi(x_1)|\hbbarot\hat{l}_j\qp_1\qp_2\psi}^2+\norm{\pp_1\hbbarot\nabla_{x_1}\hat{l}_j\qp_1\qp_2\psi}^2\nonumber\\
&&+\left(\llr{|\Phi(x_1)\rangle\langle\nabla_{x_1}\Phi(x_1)|\hbbarot\hat{l}_j\qp_1\qp_2\psi,\pp_1\hbbarot\nabla_{x_1}\hat{l}_j\qp_1\qp_2\psi}+\mathrm{h.c.}\right)\nonumber\\
&\ls&\onorm{\hbbarot\nabla_{x_1}\pp_1}^2\norm{\hat{l}\qp_1\qp_2\psi}^2+\onorm{\hbbarot\pp_1}^2\norm{\nabla_{x_1}\hat{l}\qp_1\qp_2\psi}^2\nonumber\\
&\ls&\efrak^2(t)\norm{\hbbar}^2_{L^2(\R^2)}\left(\llr{\psi,\hat{n}\psi}+\norm{\nabla_{x_1}\qp_1\psi}^2\right)\,.
\label{eqn:p23trick}
\end{eqnarray}
Integration by parts in $x_2$ with $\rho=1$ yields with Lemmas~\ref{lem:commutators:2}, \ref{lem:bar}, \ref{lem:a_priori:4} and \ref{lem:log}
\begin{eqnarray*}
|\eqref{gamma_b_4:2}|&\leq&
N\left|\llr{\hat{l}\qp_1\qp_2\psi,\pc_1\pc_2\vbbaro^{(12)}\qp_2\pp_1\psi}\right|
+N\left|\llr{\hat{l}\qp_1\qp_2\psi,\pc_1\pc_2(\nabla_{x_2}\hbbaroo^{(12)})\pp_1\nabla_{x_2}\qp_2\psi}\right|\\
&&+N\left|\llr{\nabla_{x_2}\qp_1\qp_2\psi,\pc_1\pc_2(\nabla_{x_2}\hbbaroo^{(12)})\pp_1\hat{l}_1\qp_2\psi}\right|\\
&\ls&N\norm{\hat{l}\qp_1\qp_2\psi}\onorm{\vbbaro^{(12)}\pp_1}\norm{\qp_2\psi}
+N\norm{\nabla_{x_2}\qp_2\psi}\norm{\pp_1(\nabla_{x_1}\hbbaroo^{(12)})\hat{l}\qp_1\qp_2\psi}\\
&&+N\norm{\nabla_{x_2}\qp_2\psi}\left(\norm{\pp_1(\nabla_{x_2}\hbbaroo^{(12)})\hat{l}_1\qp_1\qp_2\psi}^2
+N^{-1}\onorm{(\nabla_{x_2}\hbbaroo^{(12)})\pp_1}^2\norm{\hat{l}_1\qp_2\psi}^2\right)^\frac12\\
&\ls&\efrak(t)\left(\llr{\psi,\hat{n}\psi}+\norm{\nabla_{x_1}\qp_1\psi}^2+\tfrac{1}{1-\beta}N^{-1^-}\right)\,.
\end{eqnarray*}
With Lemma~\ref{lem:E_kin<} below, we obtain
\begin{eqnarray*}
|\eqref{gamma_b_4:2}|&\ls&\efrak^3(t)\alwb(t)+\efrak^4(t)
\left(\tfrac{\mu^\beta}{\varepsilon}+\left(\tfrac{\varepsilon^3}{\mu^\beta}\right)^\frac12
+N^{-\beta_2^-}+\mu^\eta+\mu^\frac{1-\beta}{2} \right)
\end{eqnarray*}
for $\beta_2=\min\left\{\beta,\frac14\right\}$.
Together, the estimates of~\eqref{gamma_b_4:1} and~\eqref{gamma_b_4:2} yield 
\begin{eqnarray*}
|\gamma_{b,<}^{(4)}(t)|&\ls&\efrak^3(t)\,\alwb(t)+\efrak^4(t)\left(\tfrac{\mu^\beta}{\varepsilon}+\left(\tfrac{\varepsilon^3}{\mu^\beta}\right)^\frac12+N^{-(\bo-\xi)^-}
+\mu^\eta+\mu^\frac{1-\beta}{2}\right)\,.
\end{eqnarray*}

\subsection{Estimate of the kinetic energy for $\beta\in(0,1)$}
\label{subsec:E_kin}
\begin{lem}\label{lem:E_kin<}
For $\beta_2=\min\left\{\frac14,\beta\right\}$ and sufficiently small $\mu$,
\begin{eqnarray*}
\norm{\nabla_{x_1}\qp_1\psi}^2&\ls& \efrak^2(t)\alwb(t)+\efrak^3(t)\left(\tfrac{\mu^\beta}{\varepsilon}+\left(\tfrac{\varepsilon^3}{\mu^\beta}\right)^\frac12+N^{-\beta_2^-}
+\mu^\eta+\mu^\frac{1-\beta}{2}\right).
\end{eqnarray*}
\end{lem}

\begin{proof}
Analogously to the proof of Lemma 4.21 in~\cite{NLS}, we expand
\begin{eqnarray}
\Eb(&&\hspace{-1cm}\Psi)-\Ecal(\Phi)\nonumber\\
&\gs& \norm{\nabla_{x_1}\qp_1\psi}^2+N\norm{\sqrt{\wbot}(1-p_1p_2)\psi}^2\label{eqn:E:kin:1}\\
&&-\left|\norm{\nabla_{x_1}\pp_1\psi}^2-\norm{\nabla_x\Phi}^2_{L^2(\R^2)}\right|-\left|\llr{\hat{n}^{-\frac12}\qp_1\psi,\Delta_{x_1}\pp_1\left(\qc_1\hat{n}^\frac12+\pc_1\hat{n}_1^\frac12\right)\psi}\right|\qquad\label{eqn:E:kin:2}\\
&&-\left|\llr{\psi,p_1p_2\left(N\wbot-\bb|\Phi(x_1)|^2\right)p_1p_2\psi}\right|-\norm{\sqrt{\wbot}p_1p_2\psi}^2\label{eqn:E:kin:3}\\
&&-N\left|\llr{\hat{n}^{-\frac12}q_1\psi,p_2\wbot p_1p_2\hat{n}^\frac12_2\psi}\right|\label{eqn:E:kin:4}\\
&&-N\left|\llr{\psi,q_1q_2\wbot p_1p_2\psi}\right|\label{eqn:E:kin:5}\\
&&-\left|\llr{\psi,(1-p_1p_2)|\Phi(x_1)|^2p_1p_2\psi}\right|-\left|\llr{\psi,(1-p_1p_2)|\Phi(x_1)|^2(1-p_1p_2)\psi}\right|\label{eqn:E:kin:6}\\
&&-\left|\llr{\psi,|\Phi(x_1)|^2\psi}-\lr{\Phi,|\Phi(x_1)|^2\Phi}\right|\label{eqn:E:kin:7}\\
&&-\left|\llr{\psi,\Vp(t,z_1)\psi}-\lr{\Phi,\Vp(t,(x_1,0))\Phi}\right|.\label{eqn:E:kin:8}
\end{eqnarray}
Note that the second term in~\eqref{eqn:E:kin:1} is non-negative. For~\eqref{eqn:E:kin:2}, we observe that
$$\norm{\nabla_{x_1}\pp_1\psi}^2-\norm{\nabla_x\Phi}^2_{L^2(\R^2)}=-\norm{\nabla_{x}\Phi}^2_{L^2(\R^2)}\norm{\qp_1\psi}^2\ls\efrak^2(t)\llr{\psi,\hat{n}\psi}$$
and
$\llr{\hat{n}^{-\frac12}\qp_1\psi,\Delta_{x_1}\pp_1\hat{n}^\frac12\psi}\ls \efrak^2(t)\llr{\psi,\hat{n}\psi}.$
Making use of $\mathcal{G}(x)$ from~\eqref{eqn:Gamma} and Lemma~\ref{lem:pfp}, we find
$|\eqref{eqn:E:kin:3}|
\ls\efrak^2(t)\left(\tfrac{\mu^\beta}{\varepsilon}+N^{-1}+\mu^\eta\right)$
and $|\eqref{eqn:E:kin:4}|\ls\efrak(t)\llr{\psi,\hat{n}\psi}$.  
Insertion of $\hat{n}^\frac12\hat{n}^{-\frac12}$ yields
$|\eqref{eqn:E:kin:6}|\ls\efrak^2(t)\llr{\psi,\hat{n}\psi}$. As a consequence of Lemmas~\ref{lem:psi-Phi} and \ref{lem:taylor},
$|\eqref{eqn:E:kin:7}|+|\eqref{eqn:E:kin:8}|\ls\efrak^2(t)\llr{\psi,\hat{n}\psi}+\efrak^3(t)\varepsilon$. Finally, we decompose $|\eqref{eqn:E:kin:5}|$ as
\begin{eqnarray*}
|\eqref{eqn:E:kin:5}|
&\ls&N\left|\llr{\psi,\qc_1q_2\wbot p_1p_2\psi}\right|
+N\left|\llr{\qc_2\psi,\qp_1\pc_1\wbot p_1p_2\psi}\right|\\
&&+N\left|\llr{\qp_1\qp_2\psi,\pc_1\pc_2\wbot p_1p_2\psi}\right|\,.
\end{eqnarray*}
Analogously to the bound of~\eqref{gamma_b_2:1} (Section~\ref{subsec:gamma_b_2}), the first line is bounded by 
$$\efrak^3(t)(\varepsilon^\frac{3}{2}\mu^{-\frac{\beta}{2}}+\mu^{\frac{1-\beta}{2}}),$$
and the second line yields 
$$\efrak^2(t)\llr{\psi,\hat{n}\psi}+\efrak^3(t)N^{-\bz^-}$$ 
for $\bz=\min\left\{\beta,\frac14\right\}$ as in the estimate of~\eqref{gamma_b_4:1} (Section~\ref{subsec:gamma_b_4}). 
\end{proof}

\section{Proofs for $\beta=1$}
\label{sec:GP}

\subsection{Microscopic structure}
This section collects properties of the scattering solution $\fb$ and its complement $\gb$.
\begin{lem}\label{lem:scat} 
Let $\fb$ and $\Rbt$ as in Definition~\ref{def:U} and $j_\mu$ as in \eqref{eqn:scat}. Then
\lemit{
	\item $\fb$ is a non-negative, non-decreasing function of $|z|$,\label{lem:scat:1}
	\item $\fb(z)\geq j_\mu(z)$ for all $z\in\R^3$ and there exists $\kappa_\bt\in \big(1,\frac{\mu^\bt}{\mu^\bt-\mu a}\big)$ such that $$\fb(z)=\kappa_\bt j_\mu(z)$$ for $|z|\leq\mu^\bt$, \label{lem:scat:2}
	\item $\Rbt\sim\mu^\bt$,\label{lem:scat:3}
	\item $\norm{\mathbbm{1}_{|z_1-z_2|<\Rbt}\nabla_1\psi}^2+\tfrac12\llr{\psi, (\wm^{(12)}-\Ubt^{(12)})\psi}\geq 0$ for any $\psi\in \mathcal{D}(\nabla_1)$.\label{lem:scat:4}
}
\end{lem}
\begin{proof}
Parts (a) to (c) are proven in \cite[Lemma 4.9]{GP}. For part (d), see~\cite[Lemma 5.1(3)]{pickl2015}.
\end{proof}

\begin{lem}\label{lem:g}
For $\gb$ as in Definition~\ref{def:U} and sufficiently small $\varepsilon$,
\lemit{
	\item $|\gb(z)|\ls\frac{\mu}{|z|}$ ,\label{lem:g:1}
	\item 
		  $\norm{\gb}_{L^2(\R^3)}\ls \mu^{1+\frac{\bt}{2}},$\label{lem:g:2}
	\item $\norm{\nabla\gb}_{L^2(\R^3)}\ls \mu^\frac12,$\label{lem:nabla:g:1}
	\item $\norm{\gbot\psi^{N,\varepsilon}(t)}
	\ls N^{-1}$,\label{lem:g:3}
	\item $\norm{\mathbbm{1}_{\supp{\gb}}(z_1-z_2)\psi^{N,\varepsilon}(t)}\ls\efrako(t)\mu^{\bt}\varepsilon^{-\frac13}=\efrako(t)N^{-\bt}\varepsilon^{\bt-\frac13}$,\label{lem:g:5}
	\item $\norm{\mathbbm{1}_{\supp{\gb(\,\cdot\,,\,y_1-y_2)}}(x_1-x_2)\psi^{N,\varepsilon}(t)}\ls\efrako(t)\mu^{\frac{p-1}{p}\bt}$ for any fixed $p\in[1,\infty)$.\label{lem:g:6} 
}
\end{lem}
\begin{proof}
Parts (a) to (c) are proven in \cite[Lemmas 4.10 and 4.11]{GP}.
Assertion (d) works analogously as~\cite[Lemma~4.10c]{GP}. For (e), we obtain similarly to~\cite[Lemma 4.10e]{GP}
\begin{eqnarray*}
\norm{\mathbbm{1}_{\supp{\gb}}(z_1-z_2)\psi}^2\ls \mu^{2\bt}\int\d z_N\mycdots \d z_2\left(\int\d z_1|\psi(z_1\mydots z_N)|^6\right)^\frac26\,,
\end{eqnarray*}
where we have used Hölder's inequality in the $\d z_1$ integration. Now we substitute $z_1\mapsto \tilde{z}_1=(x_1,\frac{y_1}{\varepsilon})$ and use Sobolev's inequality in the $\d\tilde{z}_1$-integration, noting that $\nabla_{\tilde{z}_1}=(\nabla_{x_1},\varepsilon\partial_{y_1})$ and $\d\tilde{z}_1=\varepsilon\d z_1$.  This yields
\begin{eqnarray*}
\left(\int\d z_1|\psi(z_1\mydots z_N)|^6\right)^\frac26
&=&\left(\varepsilon\int\d \tilde{z}_1|\psi((x_1,\varepsilon\tilde{y}_1),z_2\mydots z_N)|^6\right)^\frac26\\
&\ls&\varepsilon^\frac13\int\d\tilde{z}_1|\nabla_{\tilde{z}_1}\psi((x_1,\varepsilon\tilde{y}_1),z_2\mydots z_N)|^2\\
&=&\varepsilon^{-\frac23}\int\d z_1\left(|\nabla_{x_1}\psi(z_1\mydots z_N)|^2+\varepsilon^2|\partial_{y_1}\psi(z_1\mydots z_N)|^2\right)\,.
\end{eqnarray*} 
The statement then follows with Lemma~\ref{lem:a_priori:4}.
For part (f), recall the two-dimensional Gagliardo--Nirenberg--Sobolev inequality: for $2< q<\infty$ and  $f\in H^1(\R^2)$, 
\begin{equation}\label{eqn:2d:Sobolev}
\norm{\nabla f}_{L^2(\R^2)}^\frac{q-2}{q}\norm{f}_{L^2(\R^2)}^\frac{2}{q}\geq S_q \norm{f}_{L^q(\R^2)}\,,
\end{equation}
where $S_q$ is a positive constant which is finite for $2<q<\infty$
(e.g.~\cite[Equation~(2.2)]{nirenberg1959} and~\cite[Equation~(2.2.5)]{lieb_stability}).
Consequently, $\norm{f}_{L^q(\R^2)}\ls\norm{f}_{L^2(\R^2)}^\frac{2}{q}\norm{\nabla f}_{L^2(\R^2)}^\frac{q-2}{q}$ for each fixed $q\in(2,\infty)$.
Hence, for any fixed $p\in(1,\infty)$ and $\psi\in L^2(\R^{3N})\cap\mathcal{D}(\nabla_{x_1})$,
\begin{eqnarray*}
&&\hspace{-1.3cm}\norm{\mathbbm{1}_{\supp{\gb}}(x_1-x_2)\psi}^2\\
&\leq&\int\d z_N\mycdots \d y_1\left(\;\int\limits_{\R^2}\mathbbm{1}_{|x|\leq\Rbt}\d x\right)^\frac{p-1}{p}\left(\;\int\limits_{\R^2}\d x_1|\psi(z_1\mydots z_N)|^{2p}\right)^\frac{2}{2p}\\
&\ls&\mu^{\frac{2\bt(p-1)}{p}}\int\d z_N\mycdots \d y_1\left(\;\int\limits_{\R^2}\d x_1|\psi(z_1\mydots z_N)|^2\right)^\frac{1}{p}\left(\;\int\limits_{\R^2}\d x_1|\nabla_{x_1}\psi(z_1\mydots z_N)|^2\right)^\frac{p-1}{p}\\
&\leq &\mu^{\frac{2\bt(p-1)}{p}}\norm{\psi}^\frac{2}{p}\norm{\nabla_{x_1}\psi}^\frac{2(p-1)}{p}\,,
\end{eqnarray*}
where we have used Hölder's inequality in the $\d x_1$ integration, applied~\eqref{eqn:2d:Sobolev}, and finally used again Hölder in the $\d z_N\mycdots \d y_1$ integration.
\end{proof}

\subsection{Characterisation of the auxiliary potential $\Ubt$}\label{subsec:GP:auxiliary}
In this section, we show that both $\Ubt\fb$ and $\Ubt$ from Definition~\ref{def:U}  are contained in the set $\Wbt$ from Definition~\ref{def:W}, which  admits the transfer of results obtained in Section~\ref{sec:NLS} to these interaction potentials.

\begin{lem}\label{lem:U:in:W}
The family $\Ubt$ is contained in $\mathcal{W}_{\bt,\eta}$ for any $\eta>0$. 
\end{lem}
\begin{proof}
Note that $\mu^{-1}\int_{\R^3}\Ubt(z)\d z=\tfrac{4\pi}{3}a(\Rbt^3\,\mu^{-3\bt}-1)=\tfrac{4\pi}{3}ac$ for some $c>0$ by Lemma~\ref{lem:scat:3}, hence $\btNe(\Ubt)=\lim_{(N,\varepsilon)\to(\infty,0)}\btNe(\Ubt)$. The remaining requirements are easily verified.
\end{proof}

\begin{lem}\label{lem:Uf:in:W}
Let $0<\eta<1-\bt$. Then the family $\Ubt\fb$ is contained in $\mathcal{W}_{\bt,\eta}$.
\end{lem}
\begin{proof}
As before, it only remains to show that $\Ubt\fb$ satisfies part (d) of Definition~\ref{def:W}.
To see this, observe that
$$
\mu^{-1}\int\limits_{\R^3}\Ubt(z)\fb(z)\d z
\overset{\text{\eqref{eqn:scat(w-U)=0}}}{=}\mu^{-1}\int\limits_{B_\mu(0)} \wm(z)\fb(z)
\overset{\text{\ref{lem:scat:2}}}{=}\mu^{-1}\kbt\int\limits_{B_\mu(0)}\wm(z)j_\mu(z)
\overset{\text{\eqref{eqn:integral:scat}}}{=}\kbt 8\pi a\,,
$$
hence $\btNe(\Ubt\fb)=\kbt 8\pi a\int_{\R}|\chi(y)|^4\d y$. By Lemma~\ref{lem:scat:2}, this implies
\begin{equation}\label{b=b}
\lim\limits_{(N,\varepsilon)\to(\infty,0)}\btNe(\Ubt\fb)=8\pi a\int\limits_{\R}|\chi(y)|^4\d y=b_1
\end{equation}
and
$$|\btNe(\Ubt\fb)-b_1|=8\pi a (\kbt-1)\int\limits_{\R}|\chi(y)|^4\d y\overset{\text{\ref{lem:scat:2}}}{\ls}  \frac{\mu a}{\mu^\bt-\mu a}\ls\mu^{1-\bt}. \quad \qedhere$$
\end{proof}

\subsection{Estimate of the kinetic energy for $\beta=1$}
The main goal of this section is to provide a bound for the kinetic energy of the part of $\psiNe(t)$ with at least one particle orthogonal to $\Phi(t)$.
Since the predominant part of the kinetic energy is caused by the microscopic structure and thus concentrated in neighbourhoods of the scattering  centres, we will consider the part of the kinetic energy originating from the complement of these neighbourhoods and prove that it is subleading.
The first step is to define the appropriate neighbourhoods $\overline{\mathcal{C}}_j$ as well as sufficiently large balls $\Abar_j\supset \overline{\mathcal{C}}_j$ around them.
\label{subsec:GP:E_kin}
\begin{definition}\label{def:cutoffs}
Let $\max\left\{\frac{\gamma+1}{2\gamma},\frac56\right\}< d<\bt$, $j,k\in\{1\mydots N\}$, and define the subsets of $\R^{3N}$
\begin{eqnarray*}
a_{j,k}&: = &\left\lbrace (z_1\mydots z_N): |z_j-z_k|<\mu^d\right\rbrace,\\
c_{j,k}&: = &\left\lbrace (z_1\mydots z_N): |z_j-z_k|<\Rbt\right\rbrace,\\
a_{j,k}^x&: = &\left\lbrace ((x_1,y_1)\mydots (x_N,y_N)): |x_j-x_k|<\mu^d\right\rbrace
\end{eqnarray*}
with $(x_j,y_j)\in\R^{2+1}$ as usual.
Then the subsets $\Abar_j$,  $\Bbar_j$, $\overline{\mathcal{C}}_j$ and $\Abar^x_j$
of $\R^{3N}$ are defined as
$$\Abar_j\; :=\; \bigcup\limits_{k\neq j}a_{j,k}, \qquad
\Bbar_j\; :=\; \bigcup\limits_{k,l\neq j}a_{k,l}, \qquad
\overline{\mathcal{C}}_j\; :=\; \bigcup\limits_{k\neq j}c_{j,k}, \qquad
\Abar^x_j\; :=\; \bigcup\limits_{k\neq j}a^x_{j,k}
$$
and their complements are denoted by $\A_j$, $\B_j$, $\mathcal{C}_j$ and $\A^x_j$, e.g.,~$\A_j:=\R^{3N}\setminus\Abar_j$.
\end{definition}
The sets $\Abar_j$ and $\Abar^x_j$  contain all $N$-particle configurations where at least one other particle is sufficiently close to particle $j$ or where the projections in the $x$-direction are close, respectively.
The sets $\mathcal{B}_j$ consist of all $N$-particle configurations where particles can interact with particle $j$ but are mutually too distant to interact among each other.

Note that the characteristic functions $\charAox$ and $\charAbarox$ do not depend on any $y$-coordinate, and
$\charBo$ and $\charBbaro$ are independent of  $z_1$. Hence, the multiplication operators corresponding to these functions commute with all operators that act non-trivially only on the $y$-coordinates or on $z_1$, respectively.
Some useful properties of these cut-off functions are collected in the following lemma.
\begin{lem}\label{lem:cutoffs}
Let $\Abaro$, $\Abar^x_1$ and $\Bbaro$ as in Definition~\ref{def:cutoffs}. Then
\lemit{
	\item $\onorm{\charAbaro p_1}\ls\efrako(t)\mu^\frac{3d-1}{2}$, \qquad  
	$\onorm{\charAbaro \nabla_{x_1}p_1}\ls\efrako(t)\,\mu^{\frac{3d-1}{2}}$,\label{lem:cutoffs:1}
	\item $\norm{\charAbaro\psi}\ls\mu^{d-\frac13}\left(\norm{\nabla_{x_1}\psi}+\varepsilon\norm{\partial_{y_1}\psi}\right)$\; for any $\psi\in L^2(\R^{3N})\cap \mathcal{D}(\nabla_1)$,\label{lem:cutoffs:2}
	\item $\norm{\charAbaro\partial_{y_1}\pc_1\psi^{N,\varepsilon}(t)}\ls\efrako(t)\varepsilon^{-1}\mu^{d-\frac13}
	,$\label{lem:cutoffs:3}
	\item $\norm{\charBbaro\psi}\ls\mu^{d-\frac13}\left(\sum\limits_{k=2}^N(\norm{\nabla_{x_k}\psi}^2+\varepsilon^2\norm{\partial_{y_k}\psi}^2)\right)^\frac12$\; for any $\psi\in L^2(\R^{3N})\cap \mathcal{D}(\nabla_1)$,\label{lem:cutoffs:4}
	\item $\norm{\charBbaro\psi^{N,\varepsilon}(t)}\ls \efrako(t)N^\frac12\mu^{d-\frac13}=\efrako(t)N^{-d+\frac56}\varepsilon^{d-\frac13}$,
	\label{lem:cutoffs:5}
	\item $\norm{\charAbarox\psi}\ls(N\mu^{2d})^\frac{p-1}{2p}\norm{\psi}^\frac{1}{p}\norm{\nabla_{x_1}\psi}^\frac{p-1}{p}$
	for any fixed $p\in(1,\infty)$, $\psi\in L^2(\R^{3N})\cap\mathcal{D}(\nabla_{x_1})$,
	\label{lem:cutoffs:6}
	\item $\norm{\charAbarox\qc_1\psi^{N,\varepsilon}(t)}\ls\efrako(t)\varepsilon^\frac{1}{p}(N\mu^{2d})^\frac{p-1}{2p}$ for any fixed $p\in(1,\infty)$.
	\label{lem:cutoffs:7}
}
\end{lem}
\begin{proof}
The proof of parts (a) to (e) works analogously to the proof of~\cite[Lemma 4.13]{GP}: one first observes that in the sense of operators,
$\charAbaro\leq\sum_{k=2}^N\mathbbm{1}_{a_{1,k}}$ and $\charBbaro\leq\sum_{k=2}^N\mathbbm{1}_{\Abar_k}$, concludes that $\int_{\R^3}\mathbbm{1}_{a_{1,k}}(z_1,z_k)\d z_1 \ls \mu^{3d}$, and proceeds as in the proof of Lemma~\ref{lem:g:5}.
The proofs of (f) and (g) work analogously to the proof of Lemma~\ref{lem:g:6}, where one uses the estimate $\int_{\R^2}\charAbarox(x_1\mydots x_N)\d x_1\ls N\mu^{2d}$.
\end{proof}

\begin{lem}\label{lem:E_kin:GP}
Let $1>\bt>d>\max\left\{\frac{\gamma+1}{2\gamma},\frac56\right\}$. Then, for sufficiently small $\mu$, 
\begin{eqnarray*}
\norm{\charAo\nabla_{x_1}\qp_1\psiNe(t)}^2&\ls &
\efrako^2(t)\alwm(t)
+\efrako^3(t)\left(\left(\tfrac{\varepsilon^\vartheta}{\mu}\right)^\frac{\bt}{2}+\left(\tfrac{\mu}{\varepsilon^\gamma}\right)^\frac{1}{\bt\gamma^2}+\mu^\frac{1-\bt}{2}+N^{-d+\frac56}\right)\,.
\end{eqnarray*}
\end{lem}

\begin{proof}

We will in the following abbreviate $\psi^{N,\varepsilon}(t)\equiv \psi$ and $\Phi(t)\equiv\Phi$. 
Analogously to~\cite[Lemma 4.12]{GP}, we decompose the energy difference as 
\begin{eqnarray}
&\hspace{-1cm}E^\psi_{\wm}&\hspace{-0.7cm}(t)-\mathcal{E}^\Phi_{b_1}(t)\nonumber\\
&\geq&\norm{\charAo\nabla_{x_1}\qp_1\psi}^2-\left|\llr{\nabla_{x_1}\qp_1\psi,\charAo\nabla_{x_1}\pp_1\qc_1\psi}\right|\label{eqn:E_kin:GP:0}
\\
&&+\norm{\charAbaro\charBo\nabla_{x_1}\psi}^2+\llr{\psi,(-\partial^2_{y_1}+\tfrac{1}{\varepsilon^2}V^\perp(\tfrac{y_1}{\varepsilon})-\tfrac{E_0}{\varepsilon^2})\psi}\nonumber\\
&&\hspace{7cm}+\tfrac{N-1}{2}\llr{\psi,\charBo \left(\wm^{(12)}-\Ubt^{(12)}\right)\psi}\label{eqn:E_kin:GP:1}\\
&&+2\Re\llr{\nabla_{x_1}p_1\psi,\charAo\nabla_{x_1}q_1\psi}\label{eqn:E_kin:GP:2}\\
&&+\norm{\charAo\nabla_{x_1}p_1\psi}^2-\norm{\nabla_x\Phi}^2_{L^2(\R)}\label{eqn:E_kin:GP:3}\\
&&+\tfrac{b_1}{2}\left(\llr{\psi,|\Phi(x_1)|^2\psi}-\lr{\Phi,|\Phi|^2\Phi}\right)
+\llr{\psi,\Vp(t,z_1)\psi}-\lr{\Phi,\Vp(t,(x,0))\Phi}\label{eqn:E_kin:GP:4}\\
&&+\tfrac{N-1}{2}\llr{\psi,\charBo p_1p_2\Ubt^{(12)}p_1p_2\charBo\psi}-\tfrac{b_1}{2}\llr{\psi,|\Phi(x_1)|^2\psi}\label{eqn:E_kin:GP:5}\\
&&+(N-1)\Re\llr{\psi,\charBo(p_1q_2+q_1p_2)\Ubt^{(12)}p_1p_2\charBo\psi}\label{eqn:E_kin:GP:6}\\
&&+(N-1)\Re\llr{\psi,\charBo q_1q_2\Ubt^{(12)}p_1p_2\charBo\psi}.\label{eqn:E_kin:GP:7}
\end{eqnarray}
The first line is easily controlled as
$$\eqref{eqn:E_kin:GP:0}\gs\norm{\charAo\nabla_{x_1}\qp_1\psi}^2-\efrako^3(t)\varepsilon\,.
$$
To estimate~\eqref{eqn:E_kin:GP:1}, note that $(c_{1,k}\cap\Bo)\cap(c_{1,l}\cap\Bo)=\emptyset$ by Definition~\ref{def:cutoffs} and since $d<\bt$ implies $\Rbt<2\Rbt<\mu^d$. Consequently,
$\charAbaro\charBo\geq\charCbaro\charBo=\charBo\sum\limits_{k=2}^N\mathbbm{1}_{c_{1,k}}=\charBo\sum\limits_{k=2}^N\mathbbm{1}_{|z_1-z_k|\leq\Rbt}$,
which yields with Lemma~\ref{lem:scat:4}
$$\norm{\charAbaro\charBo\nabla_{1}\psi}^2+\tfrac{N-1}{2}\llr{\charBo\psi,\left(\wm^{(12)}-\Ubt^{(12)}\right)\charBo\psi}\geq 0\,.$$
To use this for~\eqref{eqn:E_kin:GP:1}, we must extract a contribution $\norm{\charAbaro\charBo\partial_{y_1}\psi}^2$ from the remaining expression $\llr{\psi, (-\partial^2_{y_1}+\tfrac{1}{\varepsilon^2}V^\perp(\tfrac{y_1}{\varepsilon})-\tfrac{E_0}{\varepsilon^2})\psi}$. 
To this end, recall that $\chie$ is the ground state of $\partial^2_{y_1}+\tfrac{1}{\varepsilon^2}V^\perp(\tfrac{y_1}{\varepsilon})$ corresponding to the eigenvalue $\frac{E_0}{\varepsilon^2}$, 
hence $O_{y_1}:=-\partial^2_{y_1}+\tfrac{1}{\varepsilon^2}V^\perp(\tfrac{y_1}{\varepsilon})-\tfrac{E_0}{\varepsilon^2}$ is a positive operator and $O_{y_1}\psi=O_{y_1}\qc_1\psi$.
Since $\charAbarox$ and $\charBbaro$ and their complements commute with any operator that acts non-trivially only on $y_1$  and since $\charAbarox\charBbaro\psi$ and $\charAox\psi$ are contained in the domain of $O_{y_1}$ if this holds for $\psi$, we find
\begin{eqnarray*}
\llr{\psi,O_{y_1}\psi}
&=&\llr{\charAbarox\charBo\qc_1\psi, O_{y_1}\charAbarox\charBo\qc_1\psi}\\
&&+\llr{(\charAbarox\charBbaro+\charAox)\psi, O_{y_1}(\charAbarox\charBbaro+\charAox)\psi}\\
&\geq&\norm{\charAbarox\charBo\partial_{y_1}\qc_1\psi}^2-\varepsilon^{-2}\norm{(V^\perp-E_0)_-}_{L^\infty(\R)}\norm{\charAbarox\qc_1\psi}^2\\
&\gs&\norm{\charAbarox\charBo\partial_{y_1}\psi}^2-2\big|\llr{\charBo\partial_{y_1}\qc_1\psi,\partial_{y_1}\pc_1\charAbarox\psi}\big|-\varepsilon^{-2}\norm{\charAbarox\qc_1\psi}^2\\
&&-\norm{\charAbarox\partial_{y_1}\pc_1\psi}^2\\
&\gs&\norm{\charAbaro\charBo\partial_{y_1}\psi}^2-\efrako^2(t)\left(\varepsilon^{-1}(N\mu^{2d})^\frac{p-1}{2p}-\varepsilon^{-2+\frac{2}{p}}(N\mu^{2d})^\frac{p-1}{p}\right)
\end{eqnarray*}
for any fixed $p\in(1,\infty)$ by Lemma~\ref{lem:cutoffs}.
Note that we have used in the last line the fact that $\charAbarox\geq\charAbaro$ in the sense of operators as $\Abar_1\subseteq\Abar^x_1$.
Now choose $p=1+\frac{2}{\gamma(2d-1)-1}$, which is contained in $(1,\infty)$ as $2d-1>\frac{1}{\gamma}$ because $d>\frac12+\frac{1}{2\gamma}$.
This yields 
$$\varepsilon^{-1}(N\mu^{2d})^\frac{p-1}{2p}=\left(N^{-1}\varepsilon^{1-\gamma}\right)^{\frac{p-1}{2p}(2d-1)}\varepsilon^{\frac{1}{2p}(\gamma(2d-1)(p-1)-p-1)}=\left(\tfrac{\mu}{\varepsilon^\gamma}\right)^{\frac{p-1}{p}(2d-1)}<\left(\tfrac{\mu}{\varepsilon^\gamma}\right)^\frac{1}{\bt\gamma^2}
$$
because, since $\gamma>1$ and $d<\bt$,
$$\tfrac{p-1}{p}(2d-1)=\tfrac{2(2d-1)}{\gamma(2d-1)+1}>\tfrac{2d-1}{d\gamma}>\tfrac{1}{\bt\gamma^2}\,.$$
For the second expression in the brackets, recall that $d>\frac{1}{2\gamma}+\frac12$ by Definition~\ref{def:cutoffs}, hence
$$\varepsilon^{-2+\frac{2}{p}}(N\mu^{2d})^\frac{p-1}{p}
=\left(N^{-1}\varepsilon^{1-\gamma}\right)^{\frac{p-1}{p}(2d-1)}\varepsilon^{\frac{p-1}{p}((\gamma-1)(2d-1)-2+2d)}
<\left(\tfrac{\mu}{\varepsilon^\gamma}\right)^{\frac{p-1}{p}(2d-1)}
<\left(\tfrac{\mu}{\varepsilon^\gamma}\right)^\frac{1}{\bt\gamma^2}\,.
$$
Consequently,
$$\eqref{eqn:E_kin:GP:1}\gs-\efrako^2(t)\left(\tfrac{\mu}{\varepsilon^\gamma}\right)^\frac{1}{\bt\gamma^2}\,.$$
Analogously to the estimates of~(48) to~(50) in~\cite[Lemma 4.12]{GP}, we obtain
\begin{eqnarray*}
|\eqref{eqn:E_kin:GP:2}|&\ls&\efrako^2(t)\left(\llr{\psi,\hat{n}\psi}+\mu^{\frac{3d-1}{2}}\right)\,,\\
|\eqref{eqn:E_kin:GP:3}|&\ls& \efrako^2(t)\left(\mu^{3d-1}+\llr{\psi,\hat{n}\psi}\right)\,,\\
|\eqref{eqn:E_kin:GP:4}|&\ls&\efrako^2(t)\llr{\psi,\hat{n}\psi}+\efrako^3(t)\varepsilon\,,
\end{eqnarray*}
where we have decomposed $\charAo=\mathbbm{1}-\charAbaro$ and used that $\norm{\nabla_{x_1}p_1\psi}^2=\norm{\nabla_x\Phi}_{L^2(\R^2)}^2\norm{p_1\psi}^2$ as well as Lemmas~\ref{lem:commutators:2}, \ref{lem:psi-Phi}, \ref{lem:Phi:1}, \ref{lem:a_priori:4},  \ref{lem:taylor} and \ref{lem:cutoffs:1}.
Analogously to the corresponding terms (51) and (52) in~\cite[Lemma~4.12]{GP}, we write~\eqref{eqn:E_kin:GP:5} as
$$\tfrac{N-1}{2}\llr{(\mathbbm{1}-\charBbaro)\psi,p_1p_2\left((\Ubt\fb)^{(12)}+(\Ubt\gb)^{(12)}\right)p_1p_2(\mathbbm{1}-\charBbaro)\psi}
-\llr{\psi,b_1|\Phi(x_1)|^2\psi}\,$$
and control the contribution with $\Ubt\fb$ and without $\charBbaro$ by means of $\mathcal{G}(x)$ as in~\eqref{eqn:Gamma}, using the respective estimates from Section~\ref{subsec:gamma_a} since $\Ubt\fb\in\Wbt$ for $\eta\in(0,1-\bt)$. 
For the remainders of~\eqref{eqn:E_kin:GP:5}, note that $\norm{\Ubt}_{L^1(\R^3)}\ls \mu$ and that
$$\norm{\Ubt\gb}_{L^1(\R^3)}=a\mu^{1-3\bt}\int_{\supp{\Ubt}}\d z |\gb(z)|\leq a\mu^{1-3\bt}\gb(\mu^\bt)\int_{\supp{\Ubt}}\d z
\ls\mu^{2-\bt}\,.$$
For~\eqref{eqn:E_kin:GP:6}, we decompose $\charBo=\mathbbm{1}-\charBbaro$ and insert $\hat{n}^\frac12\hat{n}^{-\frac12}$ into the term with identities on both sides. 
This leads to the bounds
\begin{eqnarray*}
|\eqref{eqn:E_kin:GP:5}|
&\ls&\efrako^3(t)\left(\tfrac{\mu^\bt}{\varepsilon}+\mu^{1-\bt}+N^{-1}+N^{-d+\frac56}\varepsilon^{d-\frac13}\right)\,,\\
|\eqref{eqn:E_kin:GP:6}|&\ls& \efrako^3(t)\left(N^{-d+\frac56}\varepsilon^{d-\frac13}+\llr{\psi,\hat{n}\psi}\right)\,.
\end{eqnarray*}
Finally, for the last term of the energy difference, we decompose $q=\qc+\pc\qp$, which yields
\begin{eqnarray}
|\eqref{eqn:E_kin:GP:7}|&\ls&
N\left|\llr{\charBo\psi, \qc_1q_2\Ubt^{(12)}p_1p_2\charBo\psi}\right|\label{eqn:7:1}
+N\left|\llr{\psi,\qc_1\qp_2\pc_2\Ubt^{(12)}p_1p_2\psi}\right|\\
&&+N\left|\llr{\charBbaro\psi,\qc_2\qp_1\pc_1\Ubt^{(12)}p_1p_2\psi}\right|\label{eqn:7:2}\\
&& 
+N\left|\llr{\charBo\psi,\qc_2\qp_1\pc_1\Ubt^{(12)}p_1p_2\charBbaro\psi}\right|\quad\label{eqn:7:3}\\
&&+N\left|\llr{\psi,\qp_1\qp_2\pc_1\pc_2\Ubt^{(12)}p_1p_2\psi}\right|\label{eqn:7:4}\\
&&+N\left|\llr{\charBbaro\psi,\qp_1\qp_2\pc_1\pc_2\Ubt^{(12)}p_1p_2\psi}\right|\label{eqn:7:5}\\
&&+N\left|\llr{\charBo\psi,\qp_1\qp_2\pc_1\pc_2\Ubt^{(12)}p_1p_2\charBbaro\psi}\right|,\label{eqn:7:6}
\end{eqnarray}
where we used the symmetry under the exchange $1\leftrightarrow 2$ of the second term in the first line.
For~\eqref{eqn:7:1}, note that $\charBo\psi$ is symmetric in $\{2\mydots N\}$ and commutes with $\nabla_1$ and $\qc_1$, hence we obtain, analogously to the estimate of~\eqref{gamma_b_2:1} (Section~\ref{subsec:gamma_b_2}), the bound
$$|\eqref{eqn:7:1}|\ls \efrako^3(t)\left(N^{\frac{\bt}{2}}\varepsilon^\frac{3-\bt}{2}+\mu^\frac{1-\bt}{2}\right)
	<\efrako^3(t)\left(\left(\tfrac{\varepsilon^\vartheta}{\mu}\right)^{\frac{\bt}{2}}+\mu^\frac{1-\bt}{2}\right)$$
since 
$$N^{\frac{\bt}{2}}\varepsilon^\frac{3-\bt}{2}
=\left(N\varepsilon^{\vartheta-1}\right)^{\frac{\bt}{2}}\varepsilon^{\frac{3-\vartheta\bt}{2}} 
\leq\left(N\varepsilon^{\vartheta-1}\right)^{\frac{\bt}{2}}$$
for $\bt\leq\frac{3}{\vartheta}$. 
For the second line and third line, note that $\pc_1\Ubt^{(12)}\pc_1=\pc_1\overline{\Ubt}(x_1-x_2,y_2)$, with $\overline{\Ubt}$ as in Definition~\ref{def:bar}, which is sensible since $\Ubt\in\Wbt$ for any $\eta>0$. Hence, with $\vbar$ and $\hrt$ as in Definition~\ref{def:bar} and Lemma~\ref{lem:bar},
we obtain with the choice $\rho=1$
\begin{eqnarray*}
|\eqref{eqn:7:2}|
&\ls&N\left|\llr{\qp_1\charBbaro\psi,\qc_2\vbar(x_1-x_2,y_2)p_1p_2\psi}\right|\\
&&+N\left|\llr{\charBbaro\nabla_{x_1}\qp_1\psi,\qc_2(\nabla_{x_1}\hrto(x_1-x_2,y_2))p_1p_2\psi}\right|\\
&&+N\left|\llr{\qp_1\charBbaro\psi,\qc_2(\nabla_{x_1}\hrto^{(12)})\nabla_{x_1}p_1p_2\psi}\right|\\
&\ls&N\norm{\charBbaro\psi}\onorm{\vbar(x_1-x_2,y_2)\pp_1}\left(\norm{\qc_2\psi}^2+N^{-1}\right)^\frac12\\
&&+N\norm{\nabla_{x_1}\qp_1\psi}\onorm{(\nabla_{x_1}\hrto(x_1-x_2,y_2))\pp_1}\left(\norm{\qc_2\psi}^2+N^{-1}\right)^\frac12\\
&&+N\norm{\charBbaro\psi}\onorm{(\nabla_{x_1}\hrto(x_1-x_2,y_2))\cdot\nabla_{x_1}\pp_1}\left(\norm{\qc_2\psi}^2+N^{-1}\right)^\frac12\\
&\ls&\efrako^3(t)(\ln\mu^{-1})^\frac12(\varepsilon+N^{-\frac12})
\end{eqnarray*}
by Lemmas~\ref{lem:Gamma:Lambda:3}, \ref{lem:a_priori:4}, \ref{lem:bar} and~\ref{lem:cutoffs:5}. Similarly, but without the need for Lemma~\ref{lem:Gamma:Lambda:3}, we obtain with $\rho=1$
\begin{eqnarray*}
|\eqref{eqn:7:3}|&\ls&\efrako^3(t)N^{-d+\frac56}\varepsilon^{d-\frac13}(\ln\mu^{-1})^\frac12\,.
\end{eqnarray*}
Analogously to the bound of~\eqref{gamma_b_4:1} in Section~\ref{subsec:gamma_b_4}, using $\hbbart$ with the choice $\rho=N^{-\frac14}$ and suitably inserting $\hat{n}^\frac12\hat{n}^{-\frac12}$,
we obtain
$$|\eqref{eqn:7:4}|\ls \efrako^2(t)\llr{\psi,\hat{n}\psi}+\efrako^3(t)N^{-\frac14^-}.$$
Finally, with the choice $\rho=N^{-\frac12}$, the last two lines can be bounded as
\begin{eqnarray*}
|\eqref{eqn:7:5}|
&\ls&N\left|\llr{\charBbaro\psi,\qp_1\qp_2\vbbar^{(12)}p_1p_2\psi}\right|
+N\left|\llr{\charBbaro\qp_1\psi,\qp_2(\nabla_{x_1}\hbbart^{(12)})\cdot\nabla_{x_1}p_1p_2\psi}\right|\\
&&+N\left|\llr{\charBbaro\nabla_{x_1}\qp_1\psi,\qp_2(\nabla_{x_1}\hbbart^{(12)})p_1p_2\psi}\right|\\
&\leq& N\norm{\charBbaro\psi}\left(\onorm{\pp_1\vbbar^{(12)}\pp_2}^2+N^{-1}\onorm{\vbbar^{(12)}\pp_1}^2\right)^\frac12+N\norm{\charBbaro\psi}\onorm{(\nabla_{x_1}\hbbart^{(12)})\cdot\nabla_{x_1}\pp_1}\\
&&+N\norm{\nabla_{x_1}\qp_1\psi}\bigg(\llr{\qp_2(\nabla_{x_1}\hbbart^{(12)})p_1p_2\psi,\qp_3(\nabla_{x_1}\hbbart^{(13)})p_1p_3\psi}\\
&&\qquad\qquad\qquad+N^{-1}\onorm{(\nabla_{x_1}\hbbart^{(12)})\pp_1}^2\bigg)^\frac12\\
&\ls&\efrako^3(t)\left(N^{-d+\frac56}\varepsilon^{d-\frac13}+N^{-\frac12}\right)\left(\ln\mu^{-1}\right)^\frac12\\
|\eqref{eqn:7:6}|&\ls&
N\norm{\charBbaro\psi}\left(\onorm{\pp_2\vbbar^{(12)}\pp_1}+N^{-\frac12}\onorm{\pp_2\vbbar^{(12)}}+\onorm{(\nabla_{x_1}\hbbarot)\pp_1}\efrako(t)\right)\\
&\ls&\efrako^2(t)N^{-d+\frac56}\varepsilon^{d-\frac13}(\ln\mu^{-1})^\frac12\,,
\end{eqnarray*}
where we used~\eqref{eqn:q23trick} with $s_1=p_1$ as well as~\eqref{eqn:p_1Up_2} and Lemmas~\ref{lem:bar}, \ref{lem:cutoffs:5},~\ref{lem:log} and~\ref{lem:a_priori:4}.
Hence, we obtain with Lemma~\ref{lem:log}
\begin{eqnarray*}
|\eqref{eqn:E_kin:GP:7}|
&\ls&\efrako^3(t)\left(\left(\tfrac{\varepsilon^\vartheta}{\mu}\right)^\frac{\bt}{2}+\mu^\frac{1-\bt}{2}+\tfrac{\gamma}{\gamma-1}N^{-\frac12^-}+\varepsilon^{1^-}+N^{-d+\frac56}\varepsilon^{(d-\frac13)^-}\right)+\efrako^2(t)\llr{\psi,\hat{n}\psi}\\
&\ls&\efrako^3(t)\left(\left(\tfrac{\varepsilon^\vartheta}{\mu}\right)^\frac{\bt}{2}+\mu^\frac{1-\bt}{2}+N^{-d+\frac56}\right)+\efrako^2(t)\llr{\psi,\hat{n}\psi}\,,
\end{eqnarray*}
where we have used that $-\frac14<-d+\frac56$ and that $\varepsilon\ll N^{-d+\frac56}\varepsilon^{d-\frac13}$, which follows because
$$\varepsilon N^{d-\frac56}\varepsilon^{\frac13-d}=\left(N\varepsilon^{\vartheta-1}\right)^{d-\frac56}\varepsilon^{\frac12-\vartheta(d-\frac56)}\leq \left(\tfrac{\varepsilon^{\vartheta}}{\mu}\right)^{d-\frac56}\ll 1$$
since $\vartheta\leq 3$. 
All estimates together imply
\begin{eqnarray*}
|E^\psi_{\wm}(t)-\mathcal{E}^\Phi_{b_1}(t)|
&\geq&\norm{\charAo\nabla_{x_1}\qp_1\psi}^2-\efrako^2(t)\llr{\psi,\hat{n}\psi}\\
&&-\efrako^3(t)\left(\mu^\frac{1-\bt}{2}+\left(\tfrac{\varepsilon^\vartheta}{\mu}\right)^\frac{\bt}{2}+N^{-d+\frac56}+\left(\tfrac{\mu}{\varepsilon^\gamma}\right)^\frac{1}{\bt\gamma^2}\right)\,,
\end{eqnarray*}
where we have used that $3d-1>1-\bt$ as $\bt>d>\frac56$ and that $\tfrac{\mu^\bt}{\varepsilon}<\left(\tfrac{\mu}{\varepsilon^\gamma}\right)^\frac{1}{\bt\gamma^2}$
because, since $\bt>\frac12+\frac{1}{2\gamma}>\frac{1}{\gamma}$,
$$\tfrac{\mu^\bt}{\varepsilon}=\left(\tfrac{\mu}{\varepsilon^\gamma}\right)^\bt\varepsilon^{\gamma\bt-1}<\left(\tfrac{\mu}{\varepsilon^\gamma}\right)^\bt<\left(\tfrac{\mu}{\varepsilon^\gamma}\right)^\frac{1}{\bt\gamma^2}$$
\end{proof}

\subsection{Proof of Proposition~\ref{prop:correction}}\label{subsec:GP:prop:correction}
Recalling that $\hat{r}=p_1p_2\hat{m}^b+(p_1q_2+q_1p_2)\hat{m}^a$, we conclude immediately 
$$N^2\left|\llr{\mathbbm{1}_{\supp\gb}(z_1-z_2)\psi,\gbot(p_1p_2\hat{m}^b+(p_1q_2+q_1p_2)\hat{m}^a)\psi}\right|
\ls \efrako^2(t)N^{-\frac{3\bt}{2}+\xi}\varepsilon^{\frac16+\frac{3\bt}{2}}<\efrako^2(t)\varepsilon^\frac{17}{12}
$$
by Lemmas~\ref{lem:g} and~\ref{lem:l:1} and because $\bt>\frac56$.
For fixed $t\in[0,\Tex)$ and sufficiently small $\varepsilon$, $\efrako^2(t)\varepsilon^\frac{5}{12}\ls 1$, hence  this is bounded by $\varepsilon$.
 
\subsection{Proof of Proposition~\ref{prop:dt_alpha:GP}}
\label{subsec:GP:prop:dt_alpha:GP}
This proof is analogous to the proof of \cite[Proposition 3.2]{GP}, and we sketch the main steps for convenience of the reader. In the sequel, we abbreviate $\psi^{N,\varepsilon}\equiv\psi$ and $\Phi(t)\equiv\Phi$.
Since
\begin{equation*}
\tfrac{\d}{\d t}\awm(t)=\tfrac{\d}{\d t}\alwm(t)-N(N-1)\Re\left(\tfrac{\d}{\d t}\llr{\psi,\gbot\hat{r}\psi}\right)\,,
\end{equation*}
Proposition~\ref{prop:alpha^<} implies that for almost every $t\in[0,\Tex)$,
\begin{equation}
\left|\tfrac{\d}{\d t} \awm\right|\leq |\gamma_{a,<}(t)|+\left|\gamma_{b,<}(t)-N(N-1)\Re\left(\tfrac{\d}{\d t}\llr{\psi,\gbot\hat{r}\psi}\right)\right|\,.\label{eqn:dt_alpha:GP:1}
\end{equation}
The second term in~\eqref{eqn:dt_alpha:GP:1} gives
\begin{eqnarray}
-N(N-1)\Re\left(\tfrac{\d}{\d t}\llr{\psi,\gbot\hat{r}\psi}\right)
&=&N(N-1)\Im\llr{\psi,\gbot\Big[\Hm(t)-\sum\limits_{j=1}^N h_j(t),\hat{r}\Big]\psi}\qquad\label{eqn:dt_alpha:GP:5}\\
&&+N(N-1)\Im\llr{\psi,\left[\Hm(t),\gbot\right]\hat{r}\psi}.\label{eqn:dt_alpha:GP:6}
\end{eqnarray}
In~\eqref{eqn:dt_alpha:GP:5}, we write
$\sum_{i<j}\wm^{(ij)}=\wm^{(12)}+\sum_{j=3}^N\left(\wm^{(1j)}+\wm^{(2j)}\right)+\sum_{3\leq i<j\leq N}\wm^{(ij)}$
and use the identity
$\wm^{(12)}-b_1(|\Phi(x_1)|^2+|\Phi(x_2)|^2)=Z^{(12)}-\tfrac{N-2}{N-1}b_1(|\Phi(x_1)|^2+|\Phi(x_2)|^2)$.
This yields
$$
\eqref{eqn:dt_alpha:GP:5}
=\gamma_a(t)+\gamma_d(t)+\gamma_e(t)+\gamma_f(t)+N(N-1)\Im\llr{\psi,\gbot\left[Z^{(12)},\hat{r}\right]\psi}\,.
$$
For~\eqref{eqn:dt_alpha:GP:6}, note that
\begin{eqnarray*}
\left[\Hm(t),\gbot\right]\hat{r}\psi
&=&\left(\wm^{(12)}- U_\bo^{(12)}\right)\fbot\hat{r}\psi-2(\nabla_1\gbot)\cdot\nabla_1\hat{r}\psi-2(\nabla_2\gbot)\cdot\nabla_2\hat{r}\psi,
\end{eqnarray*}
hence
$$
\eqref{eqn:dt_alpha:GP:6}=\gamma_c(t)+N(N-1)\Im\llr{\psi,\left(\wm^{(12)}- \Ubt^{(12)}\right)\fbot\hat{r}\psi}\,.
$$
The expressions $\gamma_{a,<}(t)$, $\gamma_{b,<}(t)$ together with the remaining terms from \eqref{eqn:dt_alpha:GP:5} and \eqref{eqn:dt_alpha:GP:6} yield
\begin{eqnarray*}
&&\hspace{-1cm}\gamma_{a,<}(t)+N(N-1)\Im\bigg(-\llr{\psi,Z^{(12)}\hat{r}\psi}+\llr{\psi,(1-\fbot)\left[Z^{(12)},\hat{r}\right]\psi}\\
&&\hspace{6cm}+\llr{\psi,(\wm^{(12)}- \Ubt^{(12)})\fbot\hat{r}\psi}\bigg)\\
&=&\gamma_{a,<}(t)-N(N-1)\Im\llr{\psi,\gbot\hat{r}Z^{(12)}\psi}\\
&&-N(N-1)\Im\llr{\psi,\left( \Ubt^{(12)}-\tfrac{b_1}{N-1}\left(|\Phi(x_1)|^2+|\Phi(x_2)|^2\right)\right)(1-\gbot)\hat{r}\psi}\\
&=&\gamma^<(t)+\gamma_b(t)\,,
\end{eqnarray*}
where we used that $\Im\llr{\psi,\tilde{Z}^{(12)}\hat{r}\psi}=\Im\llr{\psi,\tilde{Z}^{(12)}\hat{m}\psi}$ and that
$$Z^{(12)}\fbot=
\left(\wm^{(12)}- \Ubt^{(12)}\right)\fbot+ \Ubt^{(12)}\fbot-\tfrac{b_1}{N-1}\left(|\Phi(x_1)|^2+|\Phi(x_2)|^2\right)\fbot.\qed$$ 

\subsection{Proof of Proposition~\ref{prop:gamma:GP}}
\subsubsection{Estimate of $\gamma^<(t)$}
\label{subsec:GP:gamma<}
To estimate $\gamma^<(t)$, we apply Proposition~\ref{prop:gamma^<} to the interaction potential $\Ubt\fb$, which makes sense since $\Ubt\fb\in\Wbt$ for $\eta\in(0,1-\bt)$ by Lemma~\ref{lem:Uf:in:W}.
Besides, we need to verify that the sequence $(N,\varepsilon)$, which satisfies \emph{A4} with $(\Theta,\Gamma)_1=(\vartheta,\gamma)$, is also admissible and moderately confining with parameters $(\Theta,\Gamma)_\bt=(\delta/\bt,1/\bt)$ for some $\delta\in(1,3)$. 
We show that this holds for $\delta=\vartheta\bt$.

By assumption, $1>\bt>\frac{\gamma+1}{2\gamma}>\frac{1}{\gamma}>\frac{1}{\vartheta}$. 
Hence, $\delta=\vartheta\bt\in(1,3)$ and we find
$$
\frac{\varepsilon^{\delta/\bt}}{\mu}=\frac{\varepsilon^\vartheta}{\mu}\,,\qquad  \quad
\frac{\mu}{\varepsilon^{1/\bt}}=\frac{\mu}{\varepsilon^\gamma}\,\varepsilon^{\gamma-1/\bt}\leq \frac{\mu}{\varepsilon^\gamma}\,.
$$
Since Proposition~\ref{prop:gamma^<} requires  the parameter $0<\xi<\min\left\{\frac13\,,\,\frac{1-\bt}{2}\,,\,\bt\,,\,\frac{3-\delta}{2}\cdot\frac{\bt}{\delta-\bt}\right\}$, we choose
$0<\xi<\min\left\{\frac{1-\bt}{2}\,,\,\frac{3-\vartheta\bt}{2(\vartheta-1)}\right\}$.

Proposition~\ref{prop:gamma^<} provides a bound for $\gamma^<(t) $, which, however, depends on $\alUf(t)$ and consequently on the energy difference $|E^\psi_{\Ubt\fb}(t)-\mathcal{E}_{\Ubt\fb}^\Phi(t)|$.
Note that $\alUf(t)$ enters only in the estimate of
 \begin{eqnarray*}
|\eqref{gamma_b_4:2}|\leq N\left|\llr{\hat{l}\qp_1\psi,\qp_2\pc_1\pc_2(\Ubt\fb)^{(12)} \pc_1\pc_2\pp_2\qp_1\psi}\right|
\end{eqnarray*}
in $\gamma_{b,<}^{(4)}(t)$.
Hence, we need a new estimate of \eqref{gamma_b_4:2} by means of Lemma~\ref{lem:E_kin:GP} to obtain a bound in terms of $|\Eb^\psi(t)-\Ecal^\Phi(t)|$.
Since $\Ubt\fb\in\Wbt$, we can define $\Ufbbar\in\Vbbar_{\Rbt}$ as in Definition~\ref{def:bar}, 
$$\pc_1\pc_2(\Ubt\fb)^{(12)}\pc_1\pc_2=\Ufbbar^{(12)}\pc_1\pc_2\,,$$
and perform an integration by parts in two steps: first, we replace $\Ufbbar$ by the potential $\vbbarbz\in\Vbbar_{\mu^\bz}$ from Definition~\ref{def:bar}, namely
\begin{equation*}
\vbbarbz(x)=
\begin{cases} \frac{1}{\pi}\mu^{-2\bz}\norm{\Ufbbar}_{L^1(\R^2)}&\text{ for }|x|<\mu^\bz\,,\\[2mm]
0 & \text{ else}\,,
\end{cases}
\end{equation*}
where we have chosen $\rho=\mu^\bz$ for some $\bz\in(0,\bt)$. 
Subsequently, we replace this potential by $\nbbaro\in\Vbbar_1$ with $\rho=1$, where $\vbbarbz$ plays the role of $\Ufbbar$, i.e.,
\begin{equation*}
\nbbaro(x):=
\begin{cases} \frac{1}{\pi}\norm{\vbbarbz}_{L^1(\R^2)}&\text{ for }|x|<1\,,\\[2mm]
0 & \text{ else}.
\end{cases}
\end{equation*}
By construction,
$$\norm{\Ufbbar}_{L^1(\R^2)}=\norm{\vbbarbz}_{L^1(\R^2)}=\norm{\nbbaro}_{L^1(\R^2)}, $$
hence, by Lemma~\ref{lem:bar:3}, the functions
$\overline{\overline{h}}_{\Rbt,\mu^\bz}$ and $\overline{\overline{h}}_{\mu^\bz,1}$ as defined in~\eqref{def:hbbar} satisfy the equations
$$\Delta_x\overline{\overline{h}}_{\Rbt,\mu^\bz}=\Ufbbar-\vbbarbz\,,\qquad
\Delta_x \overline{\overline{h}}_{\mu^\bz,1}=\vbbarbz-\nbbaro\,.$$
Hence,
$$\pc_1\pc_2(\Ubt\fb)^{(12)}\pc_1\pc_2 = \left(\Delta_x \overline{\overline{h}}_{\Rbt,\mu^\bz}+\Delta_x \overline{\overline{h}}_{\mu^\bz,1}+\nbbaro\right)\pc_1\pc_2\,,$$
and consequently
\begin{eqnarray}
|\eqref{gamma_b_4:2}|
&\leq& N\left|\llr{\pc_1\nabla_{x_1}\qp_1\psi,\qp_2(\nabla_{x_1}\hbbarbz^{(12)}) p_2\qp_1\hat{l}_1\psi}\right|\label{eqn:gamma<:1}\\
&&+N\left|\llr{\pc_1\hat{l}\qp_1\psi,\qp_2(\nabla_{x_1}\hbbarbz^{(12)})p_2\nabla_{x_1}\qp_1\psi}\right|\label{eqn:gamma<:2}\\
&&+N\left|\llr{\nabla_{x_1}\qp_1\psi,\qp_2(\nabla_{x_1}\hbbaro^{(12)}) \pc_1\pc_2\pp_2\hat{l}_1\qp_1\psi}\right|\label{eqn:gamma<:3}\\
&&+N\left|\llr{\hat{l}\qp_1\psi,\qp_2(\nabla_{x_1}\hbbaro^{(12)}) \pc_1\pc_2\pp_2\nabla_{x_1}\qp_1\psi}\right|\label{eqn:gamma<:4}\\
&&+N\left|\llr{\hat{l}\qp_1\qp_2\psi,\nbbaro^{(12)} \pc_1\pc_2\pp_2\qp_1\psi}\right|\,.\label{eqn:gamma<:5}
\end{eqnarray}
With Lemma~\ref{lem:Gamma:Lambda:1}, the first two lines can be bounded as
\begin{eqnarray*}
\eqref{eqn:gamma<:1}
&\ls&N\efrako(t)\Big(\llr{\qp_2(\nabla_{x_2}\hbbarbz^{(12)})p_2\qp_1\hat{l}_1\psi,\qp_3(\nabla_{x_3}\hbbarbz^{(13)})p_3\qp_1\hat{l}_1\psi}\\
&&\qquad\qquad+N^{-1}\onorm{\nabla_{x_1}\hbbarbz^{(12)}\pp_2}^2\Big)^\frac12\\
&\ls&\efrako^3(t)\left(\mu^\bz+N^{-\frac12}\right)\left(\ln\mu^{-1}\right)^\frac12\,,\\
\eqref{eqn:gamma<:2}
&\ls&N\norm{\nabla_{x_1}\qp_1\psi}\norm{\pp_2(\nabla_{x_2}\hbbarbz^{(12)})\hat{l}\qp_1\qp_2\psi}
\;\ls\; \efrako^2(t)\mu^\bz\ln\mu^{-1}\,,
\end{eqnarray*}
where we used for~\eqref{eqn:gamma<:1} the estimate~\eqref{eqn:q23trick} with $s_1=\qp_1$ and $\tilde{\psi}=\hat{l}_1\psi$ 
and for~\eqref{eqn:gamma<:2} the estimate~\eqref{eqn:p23trick}
and applied Lemma~\ref{lem:bar:4}.
To estimate~\eqref{eqn:gamma<:3} and~\eqref{eqn:gamma<:4}, we insert identities $\mathbbm{1}=\charAo+\charAbaro$ to be able to use Lemma~\ref{lem:E_kin:GP}:
\begin{eqnarray}
\eqref{eqn:gamma<:3}+\eqref{eqn:gamma<:4}
&\leq&N\left|\llr{\nabla_{x_1}\qp_1\psi,\charAbaro\qp_2(\nabla_{x_1}\hbbaro^{(12)})  p_2\pc_1\qp_1\hat{l}_1\psi}\right|\label{eqn:gamma<:6}\\
&&+N\left|\llr{\nabla_{x_1}\qp_1\psi,\charAbaro  p_2 \pc_1(\nabla_{x_1}\hbbaro^{(12)})  \hat{l}\qp_1\qp_2\psi}\right|\label{eqn:gamma<:7}\\
&&+N\left|\llr{\charAo\nabla_{x_1}\qp_1\psi,\qp_2(\nabla_{x_1}\hbbaro^{(12)}) p_2 \pc_1\qp_1\hat{l}_1\psi}\right|\label{eqn:gamma<:8}\\
&&+N\left|\llr{\hat{l}\qp_1\psi,\qp_2(\nabla_{x_1}\hbbaro^{(12)}) p_2 \pc_1\charAo\nabla_{x_1}\qp_1\psi}\right|\,.\label{eqn:gamma<:9}
\end{eqnarray}
By Lemma~\ref{lem:cutoffs:2}, we find for $\tilde{\psi}\in L^2(\R^{3N})$ and with $x=(x^{(1)},x^{(2)})$
\begin{eqnarray*}
&&\hspace{-0.5cm}\norm{\charAbaro\qp_2(\nabla_{x_1}\hbbaro^{(12)})\pp_2\tilde{\psi}}^2\\
&&\;=\;\norm{\charAbaro\qp_2(\partial_{x^{(1)}_1}\hbbaro^{(12)})\pp_2\tilde{\psi}}^2
+\norm{\charAbaro\qp_2(\partial_{x^{(2)}_1}\hbbaro^{(12)})\pp_2\tilde{\psi}}^2\\
&&\;\ls\;\mu^{d-\frac13}\bigg(\norm{\big(\partial_{x^{(1)}_1}^2\hbbaro^{(12)}\big)\pp_2\tilde{\psi}}^2+\norm{\big(\partial_{x^{(2)}_1}^2\hbbaro^{(12)}\big)\pp_2\tilde{\psi}}^2
+\norm{\big(\partial_{x^{(1)}_1}\partial_{x^{(2)}_1}\hbbaro^{(12)}\big)\pp_2\tilde{\psi}}^2\\
&&\qquad\quad+\norm{\big(\partial_{x^{(1)}_1}\hbbaro^{(12)}\big)\pp_2\nabla_{x_1}\tilde{\psi}}^2
+\norm{\big(\partial_{x^{(2)}_1}\hbbaro^{(12)}\big)\pp_2\nabla_{x_1}\tilde{\psi}}^2
+\varepsilon^2\norm{(\nabla_{x_1}\hbbaro^{(12)}\pp_2)\partial_{y_1}\tilde{\psi}}^2\bigg)\,,
\end{eqnarray*}
and analogously for the respective expression in~\eqref{eqn:gamma<:7}.
Note that for $i,j\in\{1,2\}$ and $F\in L^2(\R^2)$ with Fourier transform $\hat{F}(k)$, it holds that $\norm{\partial_{x^{(j)}}F}^2_{L^2(\R^2)}\leq\norm{\nabla_{x}F}^2_{L^2(\R^2)}$ and that
\begin{equation*}
\norm{\partial_{x^{(i)}}\partial_{x^{(j)}}F}^2_{L^2(\R^2)}
=\norm{k^{(i)} k^{(j)}\hat{F}}^2_{L^2(\R^2)}
\ls\norm{((k^{(1)})^2+(k^{(2)})^2)\hat{F}}^2_{L^2(\R^2)}
=\norm{\Delta_x F}^2_{L^2(\R^2)}\,.
\end{equation*}
Hence, we conclude with Lemma~\ref{lem:pfp:4} that
\begin{eqnarray*}
\eqref{eqn:gamma<:6}+\eqref{eqn:gamma<:7}&\ls&
N\norm{\nabla_{x_1}\qp_1\psi}\mu^{d-\frac13}\efrako(t)
\bigg(\norm{\Delta_{x}\hbbaro}_{L^2(\R^2)}\norm{\hat{l}\qp_1\psi}\\
&&+\norm{\nabla_{x}\hbbaro}_{L^2(\R^2)}\norm{\nabla_{x_1}\hat{l}\qp_1\psi}
+\varepsilon\norm{\nabla_{x}\hbbaro}_{L^2(\R^2)}\onorm{\partial_{y_1}\pc_1}\norm{\hat{l}\qp_1\psi}\bigg)\\
&\ls&\efrako^3(t)\left(\mu^{d-\bz-\frac13}+N^\xi\mu^{d-\frac13}(\ln\mu^{-1})^\frac12\right)\,,
\end{eqnarray*}
which follows because $\Delta_{x}\hbbaro=\vbbarbz-\nbbaro$.
For the next two lines, note that $\charAo\nabla_{x_1}\qp_1\psi$ is symmetric in$\{2\mydots N\}$, hence we can apply Lemma~\ref{lem:commutators:1}. Similarly to the estimate that led to~\eqref{eqn:q23trick}, integrating by parts twice yields
\begin{eqnarray*}
\eqref{eqn:gamma<:8}
&\ls&N\norm{\charAo\nabla_{x_1}\qp_1\psi}\Big(\norm{\hat{l}_1\qp_1\qp_2\psi}^2\onorm{\hbbaro^{(12)}\nabla_{x_2}\pp_2}^2+\norm{p_2\hbbaro^{(12)}\nabla_{x_2}\hat{l}_1\qp_1\qp_2\psi}^2\\
&&\hspace{8cm}+N^{-1}\onorm{(\nabla_{x_1}\hbbaro)p_2}^2\Big)^\frac12\,.
\end{eqnarray*}
Further, proceeding as in~\eqref{eqn:p23trick}, we find
\begin{eqnarray*}
\eqref{eqn:gamma<:9}
&\ls&N\norm{\charAo\nabla_{x_1}\qp_1\psi}\left(\onorm{\hbbaro^{(12)}\nabla_{x_1}\pp_1}\norm{\hat{l}\qp_1\qp_2\psi}+\norm{p_1\hbbaro^{(12)}\nabla_{x_1}\hat{l}\qp_1\qp_2\psi}\right)\,.
\end{eqnarray*}
By Lemmas~\ref{lem:fqq:2},~\ref{lem:bar:4} and~\ref{lem:cutoffs:2}, we obtain for $j\in\{0,1\}$
\begin{eqnarray*}
&&\hspace{-1cm}\norm{p_1\hbbaro^{(12)}\hat{l}_j\qp_2\nabla_{x_1}\qp_1\psi}^2\\
&\ls&\norm{p_1\hbbaro^{(12)}\hat{l}_j\qp_2\charAo\nabla_{x_1}\qp_1\psi}^2
+\left|\llr{\nabla_{x_1}\qp_1\psi,  \charAbaro\hat{l}_j\qp_2\hbbaro^{(12)} p_1\hbbaro^{(12)}\hat{l}_j\qp_2\nabla_{x_1}\qp_1\psi}\right|\\
&&+\left|\llr{\nabla_{x_1}\qp_1\psi,  \charAbaro\hat{l}_j\qp_2\hbbaro^{(12)} p_1\hbbaro^{(12)}\hat{l}_j\qp_2\charAo\nabla_{x_1}\qp_1\psi}\right|\\
&\ls&\onorm{p_1\hbbaro^{(12)}}^2\norm{\charAo\nabla_{x_1}\qp_1\psi}^2\\
&&+\norm{\nabla_{x_1}\qp_1\psi}^2\mu^{d-\frac13}\onorm{\hat{l}}\Big(
\onorm{(\nabla_{x_1}\hbbaro^{(12)})\pp_1}+\onorm{\hbbaro^{(12)}\nabla_{x_1}\pp_1}\\
&&\hspace{6.5cm}+\varepsilon\onorm{\partial_{y_1}\pc_1}\onorm{\hbbaro^{(12)}\pp_1}
\Big)\onorm{\hbbaro^{(12)}\pp_1}\\
&\ls& \efrako^2(t)N^{-2}\norm{\charAo\nabla_{x_1}\qp_1\psi}^2+\efrako^4(t)N^{-2+\xi}\mu^{d-\frac13}(\ln\mu^{-1})^\frac12\,.
\end{eqnarray*}
Combining these estimates, we conclude with Lemma~\ref{lem:E_kin:GP}
\begin{eqnarray*}
\eqref{eqn:gamma<:8}+\eqref{eqn:gamma<:9}
&\ls& \efrako(t)\left(\norm{\charAo\nabla_{x_1}\qp_1\psi}^2+\llr{\psi,\hat{n}\psi}+N^{-1}\ln\mu^{-1}\right)+\efrako^3(t)N^{\xi}\mu^{d-\frac13}(\ln\mu^{-1})^\frac12\\
&\ls&\efrako^3(t)\alwm(t)+\efrako^4(t)\left(\left(\tfrac{\varepsilon^\vartheta}{\mu}\right)^\frac{\bt}{2}+\left(\tfrac{\mu}{\varepsilon^\gamma}\right)^\frac{1}{\bt\gamma^2}+\mu^\frac{1-\bt}{2}+N^{-d+\frac56}\right)\,.
\end{eqnarray*}
Finally,
\begin{eqnarray*}
\eqref{eqn:gamma<:5}&\ls&N\norm{\hat{l}\qp_1\qp_2\psi}\norm{\qp_1\psi}\onorm{\nbbaro^{(12)}\pp_2}
\ls\efrako(t)\llr{\psi,\hat{n}\psi}
\end{eqnarray*}
by Lemmas~\ref{lem:fqq},~\ref{lem:pfp:4} and by Definition~\ref{def:Wbar} of $\Vbbar_1$.
With the choice $\bz=\frac{3d-1}{6}>\frac{1-\bt}{2}$, all estimates together yield
\begin{equation*}
|\eqref{gamma_b_4:2}|\ls \efrako^3(t)\alwm+\efrako^4(t)\left(\left(\tfrac{\varepsilon^\vartheta}{\mu}\right)^\frac{\bt}{2}+\left(\tfrac{\mu}{\varepsilon^\gamma}\right)^\frac{1}{\bt\gamma^2}+\mu^\frac{1-\bt}{2}+N^{-d+\frac56} \right)\,.
\end{equation*}
In combination with the remaining bounds from Proposition~\ref{prop:gamma^<}, evaluated for $\bt$, $\eta=(1-\bt)^-$ and $\delta=\vartheta\bt$, we obtain
\begin{eqnarray*}
|\gamma^<(t)|&\ls&\efrako^3(t)\alwm+\efrako^4(t)\left(\left(\tfrac{\varepsilon^\vartheta}{\mu}\right)^{\frac{\bt}{2}}+\left(\tfrac{\mu}{\varepsilon^\gamma}\right)^\frac{1}{\bt\gamma^2}+\varepsilon^\frac{1-\bt}{2}+N^{-d+\frac56} \right)\,.
\end{eqnarray*}

\subsubsection{Estimate of the remainders $\gamma_a(t)$ to $\gamma_f(t)$}
\label{subsec:GP:remainders}
The estimates of $\gamma_a(t)$, $\gamma_b(t)$ as well as the bounds for $\gamma_d(t)$ to $\gamma_f(t)$ work mostly analogously to the respective estimates in~\cite[Section 4.5]{GP}, hence we merely sketch the main steps for completeness.

Recalling that $\hat{r}:=\hat{m}^bp_1p_2+\hat{m}^a(p_1q_2+q_1p_2)$, one concludes with Lemmas~\ref{lem:taylor}, \ref{lem:g:2} and \ref{lem:l:2} that
\begin{eqnarray*}
|\gamma_a(t)|
&\ls& N^3\norm{(\Vp(t,z_1)-\Vp(t,(x_1,0)))\psi}\onorm{\gbot p_1}\left(\onorm{\hat{m}^a}+\onorm{\hat{m}^b}\right)\\
&\ls&\efrako^4(t)N^{1+\xi-\frac{\bt}{2}}\varepsilon^{\frac{3+\bt}{2}}
<\efrako^4(t)\left(\tfrac{\varepsilon^\vartheta}{\mu}\right)^{1+\xi-\frac{\bt}{2}}
\end{eqnarray*}
since $\bt>\frac56$,  $\xi<\frac{1}{12}$ and $\vartheta\leq 3$.
To estimate $\gamma_b(t)$, note first that $b_\bt=b(\Ubt\fb)=b_1$ by \eqref{b=b}, hence $\eqref{gamma:GP:b:1:2}=0$. The two remaining terms can be controlled as
\begin{eqnarray*}
|\eqref{gamma:GP:b:1:1}|&\ls& N\norm{\Phi}^2_{L^\infty(\R)}\onorm{\gbot p_1}\left(\onorm{\hat{m}^a}+\onorm{\hat{m}^b}\right)\\
&\ls&\efrako^3(t)N^{-1-\frac{\bt}{2}+\xi}\varepsilon^\frac{1+\bt}{2}\;<\;\efrako^3(t)\varepsilon^\frac{1+\bt}{2}\,,\\
|\eqref{gamma:GP:b:2}|
&\ls &N^2\onorm{p_1\gbot}\left(\onorm{\hat{m}^a}+\onorm{\hat{m}^b}\right)\norm{p_1\left(\wm^{(12)}-\tfrac{b_1}{N-1}(|\Phi(x_1)|^2+|\Phi(x_2)|^2)\right)\psi}\\
&\ls&\efrako^3(t)N^{-1-\frac{\bt}{2}+\xi}\varepsilon^\frac{1+\bt}{2}\;<\;\efrako^3(t)\varepsilon^\frac{1+\bt}{2}
\end{eqnarray*}
as a consequence of Lemmas~\ref{lem:l:2}, \ref{lem:Phi:1}, \ref{lem:w12:4} and \ref{lem:g:2}.
The first term of $\gamma_d(t)$ yields 
\begin{eqnarray*}
|\eqref{gamma:GP:d:1}|
&\ls& N^3\norm{\mathbbm{1}_{\supp\gb}(z_1-z_2)\psi}\onorm{\gbot p_1}\norm{\Phi}^2_{L^\infty(\R)}\left(\onorm{\hat{m}^a}+\onorm{\hat{m}^b}\right)\\
&\ls& \efrako^4(t)N^{1+\xi-\frac{3\bt}{2}}\varepsilon^{\frac{3\bt}{2}+\frac16}<\efrako^4(t)\varepsilon
\end{eqnarray*}
since $\bt>\frac56$ and $\xi<\frac{1}{12}$. 
For the second term of $\gamma_d(t)$, we write
$\hat{r}=\hat{m}^a(p_1+p_2)+(\hat{m}^b-2\hat{m}^a)p_1p_2$, apply Lemma~\ref{lem:commutators:5} with $\hat{m}^c$ and $\hat{m}^d$ from Definition~\ref{def:weights}, and observe that $\gbot \wm^{(13)}\neq0$ implies $|z_2-z_3|\leq 2\Rbt$ because $|z_1-z_2|\leq \Rbt$ for $z_1-z_2\in\supp\gb$ and $|z_1-z_3|\leq \mu$ for $z_1-z_3\in\supp\wm$.
This leads to
\begin{eqnarray*}
|\eqref{gamma:GP:d:2}|
&\ls&  N^3\left|\llr{\psi,\gbot p_2\left[\mathbbm{1}_{\supp{\wm}}(z_1-z_3)\wm^{(13)},p_1p_3\hat{m}^d+(p_1q_3+q_1p_3)\hat{m}^c\right]\psi}\right|\\
&&+N^3\left|\llr{p_1\mathbbm{1}_{\supp{\wm}}(z_1-z_3)\gbot\wm^{(13)}\psi,\mathbbm{1}_{B_{2\Rbt}(0)}(z_2-z_3)\hat{m}^a\psi}\right|\\
&&+N^3\left|\llr{\psi,\gbot p_1(\hat{m}^a+ p_2(\hat{m}^b-2\hat{m}^a))p_1\wm^{(13)}\psi}\right| \\
&&+N^3\left|\llr{\wm^{(13)}\psi,\gbot p_2\mathbbm{1}_{\supp{\wm}}(z_1-z_3)p_1(\hat{m}^b-2\hat{m}^a)\psi}\right|\\
&\ls&\efrako^3(t)\left( N^{-1-\frac{\bt}{2}+3\xi}\varepsilon^\frac{1+\bt}{2}
+N^{1+\xi-\bt}\varepsilon^{\bt-\frac13}
+N^{-\frac{\bt}{2}+\xi}\varepsilon^\frac{1+\bt}{2}
\right)\\
&<&\efrako^3(t)\left(\left(\tfrac{\varepsilon^\vartheta}{\mu}\right)^{1+\xi-\bt}
+\varepsilon^\frac{1+\bt}{2}
\right)
\end{eqnarray*}
since $\bt>\frac56$ and $\xi<\frac{1}{12}$ and where we have estimated $\norm{\mathbbm{1}_{B_{2\Rbt}(0)}(z_2-z_3)\hat{m}^a\psi}^2$ analogously to Lemma~\ref{lem:g:5}.
Using Lemma~\ref{lem:commutators:5}, the relation
\begin{eqnarray*}
p_3p_4(\hat{r}-\hat{r}_2)+(p_3q_4+q_3p_4)(\hat{r}-\hat{r}_1)
&=&(p_1q_2+q_1p_2)(p_3q_4+q_3p_4)\hat{m}^c
+(p_1q_2+q_1p_2)p_3p_4\hat{m}^d\\
&&+p_1p_2(p_3q_4+q_3p_4)\hat{m}^e+
p_1p_2p_3p_4\hat{m}^f\,,
\end{eqnarray*}
and the symmetry of $\psi$, we obtain
\begin{eqnarray*}
|\gamma_e(t)|
&\ls& N^4\left|\llr{\psi,\gbot p_1q_2\left[\wm^{(34)},p_3q_4\hat{m}^c+p_3p_4\hat{m}^d\right]\psi}\right|
\\
&&+N^4\left|\llr{\psi,\gbot p_1p_2\left[\wm^{(34)},p_3q_4\hat{m}^e+p_3p_4\hat{m}^f\right]\psi}\right|\\
&\ls& N^4\norm{p_3\wm^{(34)}\psi}\onorm{\gbot p_1}\left(\onorm{\hat{m}^c}+\onorm{\hat{m}^d}\right)\\
&\ls&\efrako^3(t)N^{-\frac{\bt}{2}+3\xi}\varepsilon^\frac{1+\bt}{2}\;<\;\efrako^3(t)\varepsilon^\frac{1+\bt}{2}
\end{eqnarray*}
by Lemmas~\ref{lem:w12:4}, \ref{lem:g:2} and Lemma~\ref{lem:l:2}.
Finally,
\begin{eqnarray*}
|\gamma_f(t)|
&\ls& N^2\efrako^2(t)\onorm{p_2\gbot}\left(\onorm{\hat{m}^a}+\onorm{\hat{m}^b}\right)
\;\ls\; \efrako^3(t)N^{-\frac{\bt}{2}+\xi}\varepsilon^\frac{1+\bt}{2}
\;<\;\efrako^3(t)\varepsilon^\frac{1+\bt}{2}\,.
\end{eqnarray*}
The last remaining term left to estimate is $\gamma_c(t)$, where we follow a different path than in~\cite{GP}: we decompose the scalar product of the gradients into its $x$- and $y$-component and subsequently integrate by parts, making use of the fact that $\nabla_{x_1}\gbot=-\nabla_{x_2}\gbot$ and analogously for $y$.
Taking the maximum over $s_2\in\{p_2,q_2\}$ and $\hat{l}\in\mathcal{L}$ from \eqref{eqn:mathcal:L}, this results in
\begin{eqnarray}
|\gamma_c(t)|
&\ls& N\left|\llr{\psi,(\nabla_{x_1}\gbot)\cdot\nabla_{x_1}p_1\hat{l}s_2\psi}\right|
+N\left|\llr{\psi,(\nabla_{x_2}\gbot) p_2\cdot\nabla_{x_1}\hat{l}q_1\psi}\right|\label{eqn:gamma:c:1}\\
&&+N\left|\llr{\pc_2\psi,(\partial_{y_2}\gbot)\partial_{y_1}p_1\hat{l}s_2\psi}\right|
+N\left|\llr{\pc_2\psi,(\partial_{y_2}\gbot)p_2\partial_{y_1}\hat{l}q_1\psi}\right|\label{eqn:gamma:c:2}\\
&&+N\left|\llr{\qc_2\psi,(\partial_{y_2}\gbot)\partial_{y_1}p_1\hat{l}s_2\psi}\right|
+N\left|\llr{\qc_2\psi,(\partial_{y_2}\gbot)p_2\partial_{y_1}\hat{l}q_1\psi}\right|\,.\label{eqn:gamma:c:3}
\end{eqnarray}
With Lemmas~\ref{lem:l:2},~\ref{lem:pfp},~\ref{lem:a_priori:4} and~\ref{lem:g},
the first line is easily estimated as
\begin{eqnarray*}
\eqref{eqn:gamma:c:1}&\ls&N\left|\llr{\nabla_{x_1}\psi,\gbot\nabla_{x_1}p_1\hat{l}s_2\psi}\right|
+N\left|\llr{\nabla_{x_2}\psi,\gbot\nabla_{x_1}p_1\hat{l}s_2\psi}\right|\\
&&+N\left|\llr{\psi,\gbot\Delta_{x_1}p_1\hat{l}s_2\psi}\right|+N\left|\llr{\psi,\gbot\nabla_{x_2}p_2\nabla_{x_1}\hat{l}q_1\psi}\right|\\
&\ls& \efrako^3(t)N^{-\frac{\bt}{2}+\xi}\varepsilon^{\frac{1+\bt}{2}}
\;<\;\efrako^3(t)\varepsilon^{\frac{1+\bt}{2}}\,.
\end{eqnarray*}
For the second line, we conclude with Lemma~\ref{lem:g:6} that for any fixed $p\in(1,\infty)$,
\begin{eqnarray*}
\eqref{eqn:gamma:c:2}
&\ls&N\left|\llr{\partial_{y_2}\pc_2\mathbbm{1}_{\supp\gb(\cdot,y_1-y_2)}(x_1-x_2)\psi,\gbot\partial_{y_1}p_1\hat{l}s_2\psi}\right|\\
&&+N\left|\llr{\partial_{y_2}\pc_2\mathbbm{1}_{\supp\gb(\cdot,y_1-y_2)}(x_1-x_2)\psi,\gbot p_2\partial_{y_1}\hat{l}q_1\psi}\right|\\
&&+N\left|\llr{\pc_2\mathbbm{1}_{\supp\gb(\cdot,y_1-y_2)}(x_1-x_2)\psi,\gbot\partial_{y_1}p_1\partial_{y_2}\hat{l}s_2\psi}\right|\\
&&+N\left|\llr{\pc_2\mathbbm{1}_{\supp\gb(\cdot,y_1-y_2)}(x_1-x_2)\psi,\gbot\partial_{y_2}p_2\partial_{y_1}\hat{l}q_1\psi}\right|\\
&&\ls N^{1+\xi}\varepsilon^{-1}\norm{\mathbbm{1}_{\supp\gb(\cdot,y_1-y_2)}(x_1-x_2)\psi}\left(\onorm{\gbot\partial_{y_1}p_1}+\onorm{\gbot p_1}\varepsilon^{-1}\right)\\
&&\ls\efrako^2(t)N^{\xi-\frac{3\bt}{2}+\frac{\bt}{p}}\varepsilon^{-\frac32+\frac{3\bt}{2}-\frac{\bt}{p}}\,.
\end{eqnarray*}
With the choice $p=\frac{\gamma+1}{\gamma-1}$, we obtain
\begin{eqnarray*}
N^{\xi-\frac{3\bt}{2}+\frac{\bt}{p}}\varepsilon^{-\frac32+\frac{3\bt}{2}-\frac{\bt}{p}}
&=&(N^{-1}\varepsilon^{1-\gamma})^{\frac{3\bt}{2}-\xi-\frac{\bt}{p}}\varepsilon^{\gamma\bt(\frac32-\frac{\gamma-1}{\gamma+1})-\frac32-\xi(\gamma-1)}\\
&\leq&(\tfrac{\mu}{\varepsilon^\gamma})^{\frac{\bt}{2}-\xi}\varepsilon^{(\gamma-1)(\frac14-\xi)}
\;<\;(\tfrac{\mu}{\varepsilon^\gamma})^{\frac{\bt}{2}-\xi}
\end{eqnarray*}
since $\bt>\frac{\gamma+1}{2\gamma}$ and $\xi<\frac14$.
Finally, the last line yields
\begin{eqnarray*}
\eqref{eqn:gamma:c:3}
&\ls&N\left|\llr{\partial_{y_2}\qc_2\psi,\gbot\partial_{y_1}p_1\hat{l}s_2\psi}\right|
+N\left|\llr{\qc_2\psi,\gbot\partial_{y_1}p_1\partial_{y_2}\hat{l}s_2\psi}\right|\\
&&+N\left|\llr{\partial_{y_2}\qc_2\psi,\gbot p_2\partial_{y_1}\hat{l}q_1\psi}\right|
+N\left|\llr{\qc_2\psi,\gbot \partial_{y_2}p_2\partial_{y_1}\hat{l}q_1\psi}\right|\\
&\ls&\efrako^2(t)N^{-\frac{\bt}{2}+\xi}\varepsilon^{-\frac{1-\bt}{2}}
\;<\;\left(\tfrac{\mu}{\varepsilon^\gamma}\right)^{\frac{\bt}{2}-\xi}\,,
\end{eqnarray*}
where the last inequality follows because
$$N^{-\frac{\bt}{2}+\xi}\varepsilon^{-\frac{1-\bt}{2}}
\;=\;(N^{-1}\varepsilon^{1-\gamma})^{\frac{\bt}{2}-\xi}\,\varepsilon^{\frac{\gamma\bt}{2}-\frac12-\xi(\gamma-1)}
\;<\;(N^{-1}\varepsilon^{1-\gamma})^{\frac{\bt}{2}-\xi}
$$
as $\bt>\frac{\gamma+1}{2\gamma}$ and $\xi<\frac14$.

\section*{Acknowledgments}
\begin{wrapfigure}{l}{0.088\textwidth}
 \vspace{-15pt}
\includegraphics[scale=0.27]{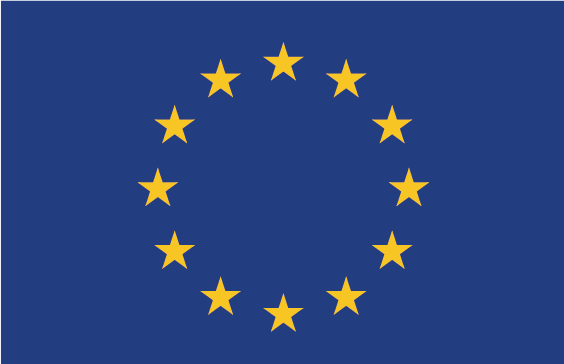}
  \vspace{-11pt}
\end{wrapfigure}
I thank Stefan Teufel for helpful remarks and for his involvement in the closely related joint project \cite{GP}.
Helpful discussions with Serena Cenatiempo and Nikolai Leopold are gratefully acknowledged.
This work was supported by the German Research Foundation within the Research Training Group 1838 ``Spectral Theory and Dynamics of Quantum Systems''
and has received funding from the European Union’s Horizon 2020 research and innovation programme under the Marie Sk{\l}odowska-Curie Grant Agreement No.\ 754411.

\renewcommand{\bibname}{References}
\bibliographystyle{abbrv}
    \bibliography{bib_PhD}

\begin{thebibliography}{10}

\bibitem{adami2007}
R.~Adami, F.~Golse, and A.~Teta.
\newblock Rigorous derivation of the cubic {NLS} in dimension one.
\newblock {\em J. Stat. Phys.}, 127(6):1193--1220, 2007.

\bibitem{adams}
R.~A. Adams and J.~J.~F. Fournier.
\newblock {\em Sobolev spaces. {Pure} and applied mathematics series, vol.
  140}.
\newblock Academic Press, 2003.

\bibitem{anapolitanos2017}
I.~Anapolitanos, M.~Hott, and D.~Hundertmark.
\newblock Derivation of the {Hartree} equation for compound {Bose} gases in the
  mean field limit.
\newblock {\em Rev. Math. Phys.}, 29(07):1750022, 2017.

\bibitem{abdallah2005_2}
N.~Ben~Abdallah, F.~M{\'e}hats, C.~Schmeiser, and R.~Weish{\"a}upl.
\newblock The nonlinear {Schr{\"o}dinger} equation with a strongly anisotropic
  harmonic potential.
\newblock {\em SIAM J. Math. Anal.}, 37(1):189--199, 2005.

\bibitem{benedikter2015}
N.~Benedikter, G.~de~Oliveira, and B.~Schlein.
\newblock Quantitative derivation of the {Gross--Pitaevskii} equation.
\newblock {\em Comm. Pure Appl. Math.}, 68(8):1399--1482, 2015.

\bibitem{boccato2017}
C.~Boccato, C.~Brennecke, S.~Cenatiempo, and B.~Schlein.
\newblock Complete {Bose--Einstein} condensation in the {Gross--Pitaevskii}
  regime.
\newblock {\em Comm. Math. Phys.}, 359(3):975--1026, 2018.

\bibitem{boccato2018}
C.~Boccato, C.~Brennecke, S.~Cenatiempo, and B.~Schlein.
\newblock Bogoliubov theory in the {Gross--Pitaevskii} limit.
\newblock {\em Acta Mathematica}, 222(2):219--335, 2019.

\bibitem{boccato2018_2}
C.~Boccato, C.~Brennecke, S.~Cenatiempo, and B.~Schlein.
\newblock Optimal rate for {Bose--Einstein} condensation in the
  {Gross-Pitaevskii} regime.
\newblock {\em Comm. Math. Phys.}, pages 1--85, 2019.

\bibitem{NLS}
L.~Bo{\ss}mann.
\newblock Derivation of the 1d nonlinear {Schr{\"o}dinger} equation from the 3d
  quantum many-body dynamics of strongly confined bosons.
\newblock {\em J. Math. Phys.}, 60(3):031902, 2019.

\bibitem{GP}
L.~Bo{\ss}mann and S.~Teufel.
\newblock Derivation of the 1d {Gross--Pitaevskii} equation from the 3d quantum
  many-body dynamics of strongly confined bosons.
\newblock {\em Ann. Henri Poincar{\'e}}, 20(3):1003--1049, 2019.

\bibitem{brennecke2017}
C.~Brennecke and B.~Schlein.
\newblock {Gross--Pitaevskii} dynamics for {Bose--Einstein} condensates.
\newblock {\em Analysis \& PDE}, 12(6):1513--1596, 2019.

\bibitem{carles2013}
R.~Carles and J.~Drumond~Silva.
\newblock Large time behaviour in nonlinear {Schrödinger} equation with time
  dependent potential.
\newblock {\em Comm. Math. Sci.}, 13(2):443--460, 2015.

\bibitem{chen2013}
X.~Chen and J.~Holmer.
\newblock On the rigorous derivation of the 2d cubic nonlinear
  {Schr{\"o}dinger} equation from 3d quantum many-body dynamics.
\newblock {\em Arch. Ration. Mech. Anal.}, 210(3):909--954, 2013.

\bibitem{chen2017}
X.~Chen and J.~Holmer.
\newblock Focusing quantum many-body dynamics {II}: The rigorous derivation of
  the 1d focusing cubic nonlinear {Schr{\"o}dinger} equation from 3d.
\newblock {\em Anal. PDE}, 10(3):589--633, 2017.

\bibitem{chen2017_2}
X.~Chen and J.~Holmer.
\newblock The rigorous derivation of the {2D} cubic focusing {NLS} from quantum
  many-body evolution.
\newblock {\em Int. Math. Res. Not.}, 2017(14):4173--4216, 2017.

\bibitem{chong2016}
J.~Chong.
\newblock Dynamics of large boson systems with attractive interaction and a
  derivation of the cubic focusing {NLS} in $\mathbb{R}^3$.
\newblock {\em arXiv:1608.01615}, 2016.

\bibitem{deoliveira2014}
G.~de~Oliveira.
\newblock Quantum dynamics of a particle constrained to lie on a surface.
\newblock {\em J. Math. Phys.}, 55(9):092106, 2014.

\bibitem{erdos2007}
L.~Erd{\H{o}}s, B.~Schlein, and H.-T. Yau.
\newblock Derivation of the cubic non-linear {Schr{\"o}dinger} equation from
  quantum dynamics of many-body systems.
\newblock {\em Invent. Math.}, 167(3):515--614, 2007.

\bibitem{erdos2010}
L.~Erd{\H{o}}s, B.~Schlein, and H.-T. Yau.
\newblock Derivation of the {Gross--Pitaevskii} equation for the dynamics of
  {Bose--Einstein} condensate.
\newblock {\em Ann. Math.}, 172(1):291--370, 2010.

\bibitem{evans}
L.~C. Evans.
\newblock {\em Partial Differential Equations}.
\newblock American Mathematical Society, 2010.

\bibitem{gorlitz2001}
A.~G{\"o}rlitz, J.~Vogels, A.~Leanhardt, C.~Raman, T.~Gustavson, J.~Abo-Shaeer,
  A.~Chikkatur, S.~Gupta, S.~Inouye, T.~Rosenband, D.~Pritchard, and
  W.~Ketterle.
\newblock Realization of {Bose--Einstein condensates} in lower dimensions.
\newblock {\em Phys. Rev. Lett.}, 87(13):130402, 2001.

\bibitem{griesemer2004}
M.~Griesemer.
\newblock Exponential decay and ionization thresholds in non-relativistic
  quantum electrodynamics.
\newblock {\em J. Funct. Anal.}, 210(2):321 -- 340, 2004.

\bibitem{griesemer2017}
M.~Griesemer and J.~Schmid.
\newblock Well-posedness of non-autonomous linear evolution equations in
  uniformly convex spaces.
\newblock {\em Math. Nachr.}, 290(2--3):435--441, 2017.

\bibitem{hadzibabic2006}
Z.~Hadzibabic, P.~Kr{\"u}ger, M.~Cheneau, B.~Battelier, and J.~Dalibard.
\newblock {Berezinskii--Kosterlitz--Thouless} crossover in a trapped atomic
  gas.
\newblock {\em Nature}, 441(7097):1118, 2006.

\bibitem{hadzibabic2008}
Z.~Hadzibabic, P.~Kr{\"u}ger, M.~Cheneau, S.~P. Rath, and J.~Dalibard.
\newblock The trapped two-dimensional {Bose} gas: from {Bose--Einstein}
  condensation to {Berezinskii--Kosterlitz--Thouless} physics.
\newblock {\em New J. Phys.}, 10(4):045006, 2008.

\bibitem{matplotlib}
J.~D. Hunter.
\newblock Matplotlib: {A 2D} graphics environment.
\newblock {\em Computing in Science \& Engineering}, 9(3):90--95, 2007.

\bibitem{jastrow}
R.~Jastrow.
\newblock Many-body problem with strong forces.
\newblock {\em Phys. Rev.}, 98(5):1479, 1955.

\bibitem{jeblick2016}
M.~Jeblick, N.~Leopold, and P.~Pickl.
\newblock Derivation of the time dependent {Gross--Pitaevskii} equation in two
  dimensions.
\newblock {\em Commun. Math. Phys.}, 372(1):1--69, 2019.

\bibitem{jeblick2018}
M.~Jeblick and P.~Pickl.
\newblock Derivation of the time dependent {Gross--Pitaevskii} equation for a
  class of non purely positive potentials.
\newblock {\em arXiv:1801.04799}, 2018.

\bibitem{jeblick2017}
M.~Jeblick and P.~Pickl.
\newblock Derivation of the time dependent two dimensional focusing {NLS}
  equation.
\newblock {\em J. Stat. Phys.}, 172(5):1398--1426, 2018.

\bibitem{keler2016}
J.~v. Keler and S.~Teufel.
\newblock The {NLS} limit for bosons in a quantum waveguide.
\newblock {\em Ann. Henri Poincar{\'e}}, 17(12):3321--3360, 2016.

\bibitem{kirkpatrick2011}
K.~Kirkpatrick, B.~Schlein, and G.~Staffilani.
\newblock Derivation of the two-dimensional nonlinear {Schr{\"o}dinger}
  equation from many body quantum dynamics.
\newblock {\em Amer. J. of Math.}, 133(1):91--130, 2011.

\bibitem{knowles2010}
A.~Knowles and P.~Pickl.
\newblock Mean-field dynamics: singular potentials and rate of convergence.
\newblock {\em Comm. Math. Phys.}, 298(1):101--138, 2010.

\bibitem{lewin2017}
M.~Lewin, P.~T. Nam, and N.~Rougerie.
\newblock A note on {2D} focusing many-boson systems.
\newblock {\em Proc. Amer. Math. Soc.}, 145(6):2441--2454, 2017.

\bibitem{lieb_loss}
E.~H. Lieb and M.~Loss.
\newblock {\em Analysis. {Graduate} studies in mathematics, vol. 14}.
\newblock American Mathematical Society, 2001.

\bibitem{lieb_stability}
E.~H. Lieb and R.~Seiringer.
\newblock {\em The Stability of Matter in Quantum Mechanics}.
\newblock Cambridge University Press, 2010.

\bibitem{LSSY}
E.~H. Lieb, R.~Seiringer, J.~P. Solovej, and J.~Yngvason.
\newblock {\em The Mathematics of the Bose Gas and its Condensation}.
\newblock Birkh{\"a}user, 2005.

\bibitem{mehats2017}
F.~M{\'e}hats and N.~Raymond.
\newblock Strong confinement limit for the nonlinear {Schr{\"o}dinger} equation
  constrained on a curve.
\newblock {\em Ann. Henri Poincar{\'e}}, 18(1):281--306, 2017.

\bibitem{michelangeli2017_2}
A.~Michelangeli and A.~Olgiati.
\newblock {Gross--Pitaevskii} non-linear dynamics for pseudo-spinor
  condensates.
\newblock {\em J. Nonlinear Math. Phys.}, 24(3):426--464, 2017.

\bibitem{michelangeli2017}
A.~Michelangeli and A.~Olgiati.
\newblock Mean-field quantum dynamics for a mixture of {Bose--Einstein
  condensates}.
\newblock {\em Anal. Math. Phys.}, 7(4):377--416, 2017.

\bibitem{nirenberg1959}
L.~Nirenberg.
\newblock On elliptic partial differential equations.
\newblock {\em Ann. Sc. Norm. Super. Pisa Cl. Sci.}, 13(2):115--162, 1959.

\bibitem{pickl2008}
P.~Pickl.
\newblock On the time dependent {Gross--Pitaevskii-} and {Hartree} equation.
\newblock {\em arXiv:0808.1178}, 2008.

\bibitem{pickl2015}
P.~Pickl.
\newblock Derivation of the time dependent {Gross--Pitaevskii} equation with
  external fields.
\newblock {\em Rev. Math. Phys.}, 27(01):1550003, 2015.

\bibitem{rychtarik2004}
D.~Rychtarik, B.~Engeser, H.-C. N{\"a}gerl, and R.~Grimm.
\newblock Two-dimensional {Bose--Einstein} condensate in an optical surface
  trap.
\newblock {\em Phys. Rev. Lett.}, 92(17):173003, 2004.

\bibitem{schnee2007}
K.~Schnee and J.~Yngvason.
\newblock Bosons in disc-shaped traps: From 3d to 2d.
\newblock {\em Commun. Math. Phys.}, 269(3):659--691, 2007.

\bibitem{smith2005}
N.~L. Smith, W.~H. Heathcote, G.~Hechenblaikner, E.~Nugent, and C.~J. Foot.
\newblock Quasi-{2}d confinement of a {BEC} in a combined optical and magnetic
  potential.
\newblock {\em J. Phys. B}, 38(3):223, 2005.

\bibitem{sohinger2012}
V.~Sohinger.
\newblock Bounds on the growth of high {Sobolev} norms of solutions to 2d
  {Hartree} equations.
\newblock {\em Discrete Contin. Dyn. Syst. A}, 32(10):3733--3771, 2012.

\bibitem{triay2019}
A.~Triay.
\newblock Derivation of the time-dependent {Gross--Pitaevskii} equation for
  dipolar gases.
\newblock {\em arXiv:1904.04000}, 2019.

\bibitem{wachsmuth}
J.~Wachsmuth and S.~Teufel.
\newblock {\em Effective Hamiltonians for constrained quantum systems}.
\newblock American Mathematical Society, 2014.

\bibitem{yefsah2011}
T.~Yefsah, R.~Desbuquois, L.~Chomaz, K.~J. G{\"u}nter, and J.~Dalibard.
\newblock Exploring the thermodynamics of a two-dimensional {Bose} gas.
\newblock {\em Phys. Rev. Lett.}, 107(13):130401, 2011.

\end{thebibliography}
\end{document}